\documentclass[a4paper,a4paper,11pt]{article}
\usepackage[utf8]{inputenc}
\usepackage{physics}
\usepackage[margin=2.54cm]{geometry}

\usepackage[
backend=biber,
style=alphabetic,
sorting=ynt,
maxbibnames=10
]{biblatex}

\usepackage[export]{adjustbox}
\usepackage{tabularx}
\usepackage{ccaption}
\usepackage{float}
\usepackage{subfig}
\usepackage{multirow}
\graphicspath{{Figures/}{../Figures/}}
\usepackage{tikz}
\usepackage{subfiles}
\usepackage{wrapfig}
\usepackage{appendix}
\usepackage{mathtools}
\usepackage{amssymb}
\usepackage{amsmath}
\usepackage{xcolor}
\usepackage{qcircuit}
\usepackage{dsfont}
\usepackage[pdftex,colorlinks=true,linkcolor=black,citecolor=blue,urlcolor=black]{hyperref}
\usepackage{thm-restate}
\usepackage{enumitem} 
\usepackage{authblk}

\usepackage{amsthm}
\usepackage{thmtools,thm-restate}
\newtheorem{theorem}{Theorem}[section]

\newtheorem{lemma}{Lemma}[section]
\newtheorem{corollary}{Corollary}[section]
\newtheorem{definition}{Definition}[section]
\newtheorem{proposition}{Proposition}[section]

\newtheorem{conjecture}{Conjecture}[section]
\newtheorem{fact}{Fact}[section]

\makeatletter
\newtheorem*{rep@theorem}{\rep@title}
\newcommand{\newreptheorem}[2]{%
\newenvironment{rep#1}[1]{%
 \def\rep@title{#2 \ref{##1}}%
 \begin{rep@theorem}}%
 {\end{rep@theorem}}}
\makeatother

\newtheorem{subdefinition}{Definition}[definition]

\newtheorem{innercustomthm}{Theorem}
\newenvironment{customthm}[1]
  {\renewcommand\theinnercustomthm{#1}\innercustomthm}
  {\endinnercustomthm}

\newtheorem{innercustomdef}{Definition}
\newenvironment{customdef}[1]
  {\renewcommand\theinnercustomdef{#1}\innercustomdef}
  {\endinnercustomdef}

\newtheorem{innercustomconj}{Conjecture}
\newenvironment{customconj}[1]
  {\renewcommand\theinnercustomconj{#1}\innercustomconj}
  {\endinnercustomconj}

\newreptheorem{theorem}{Theorem}
\newreptheorem{lemma}{Lemma}

\theoremstyle{definition}
\newtheorem{example}{Example}[section]
\newtheorem{remark}{Remark}[section]

\newcommand{\mO}{\mathcal{O}}

\newcommand{\Pclass}{\mathsf{P}}
\newcommand{\RP}{\mathsf{RP}}
\newcommand{\BPP}{\mathsf{BPP}}
\newcommand{\NP}{\mathsf{NP}}
\newcommand{\NqP}{\mathsf{NqP}}
\newcommand{\MA}{\mathsf{MA}}
\newcommand{\AM}{\mathsf{AM}}
\newcommand{\UMA}{\mathsf{UMA}}
\newcommand{\QMA}{\mathsf{QMA}}
\newcommand{\UQMA}{\mathsf{UQMA}}
\newcommand{\QCMA}{\mathsf{QCMA}}
\newcommand{\UQCMA}{\mathsf{UQCMA}}
\newcommand{\PCP}{\mathsf{PCP}}
\newcommand{\QPCP}{\mathsf{QPCP}}
\newcommand{\QCPCP}{\mathsf{QCPCP}}
\newcommand{\SPCP}{\mathsf{SPCP}}
\newcommand{\PH}{\mathsf{PH}}
\newcommand{\BQP}{\mathsf{BQP}}

\newcommand{\StoqMA}{\mathsf{StoqMA}}

\newcommand{\poly}{\mathrm{poly}}

\newcommand{\kb}[2]{|#1\rangle\langle#2|}

\title{Guidable Local Hamiltonian Problems with Implications to Heuristic Ansatz State Preparation and the Quantum PCP Conjecture}

\author[1,2]{Jordi Weggemans}

\author[1]{Marten Folkertsma}

\author[2,3]{Chris Cade}

\affil[1]{QuSoft \& CWI, Amsterdam, the Netherlands}
\affil[2]{Fermioniq, Amsterdam, the Netherlands}
\affil[3]{QuSoft \& University of Amsterdam (UvA), Amsterdam, the Netherlands}

\date{\today}

\begin{document}
\maketitle

\begin{abstract}
\noindent We study `Merlinized' versions of the recently defined Guided Local Hamiltonian problem, which we call `\emph{Guidable} Local Hamiltonian' problems. Unlike their guided counterparts, these problems do not have a guiding state provided as a part of the input, but merely come with the promise that one \emph{exists}. We consider in particular two classes of guiding states: those that can be prepared efficiently by a quantum circuit; and those belonging to a class of quantum states we call \emph{classically evaluatable}, for which it is possible to efficiently compute expectation values of local observables classically. We show that guidable local Hamiltonian problems for both classes of guiding states are $\QCMA$-complete in the inverse-polynomial precision setting, but lie within $\NP$ (or $\NqP$) in the constant precision regime when the guiding state is classically evaluatable.

Our completeness results show that, from a complexity-theoretic perspective, classical Ansätze selected by classical heuristics are just as powerful as quantum Ansätze prepared by quantum heuristics, as long as one has access to quantum phase estimation. In relation to the quantum PCP conjecture, we (i) define a complexity class capturing quantum-classical probabilistically checkable proof systems and show that it is contained in $\BQP^{\NP[1]}$ for constant proof queries; (ii) give a no-go result on `dequantizing' the known quantum reduction which maps a $\QPCP$-verification circuit to a local Hamiltonian with constant promise gap; (iii) give several no-go results for the existence of quantum gap amplification procedures that preserve certain ground state properties; and (iv) propose two conjectures that can be viewed as stronger versions of the NLTS theorem. Finally, we show that many of our results can be directly modified to obtain similar results for the class $\MA$.
\end{abstract}

\newpage
\tableofcontents

\newpage

\section{Introduction}
\label{sec:intro}
Quantum chemistry and quantum many-body physics are generally regarded as two of the most promising application areas of quantum computing~\cite{Aaronson2009ComputationalCW,Bauer2020chemical}. Whilst perhaps the original vision of the early pioneers of quantum computing was to simulate the \emph{time-dynamics} of quantum systems~\cite{Benioff1980computer,feynman1982simulating}, for many applications one is interested in \emph{stationary} properties.
One particularly noteworthy quantity is the \textit{ground state energy} (which corresponds to the smallest eigenvalue) of a local Hamiltonian describing a quantum mechanical system of interest, say a small molecule or segment of material.
The precision to which one can estimate the ground state energy plays a crucial role in practice: for instance, in chemistry the relative energies of molecular configurations enter into the exponent of the term computing reaction rates, making the latter exceptionally sensitive to small (non-systematic) errors in energy calculations. Indeed, to match the accuracy obtained by experimentation for such values one aims for an accuracy that is smaller than so-called \emph{chemical accuracy}, which is about 1.6 millihartree.\footnote{This quantity, which is $\approx$1 kcal/mol, is chosen to match the accuracy achieved by thermochemical experiments.} This quantity -- which reads as a constant -- is defined with respect to a (physical) Hamiltonian whose norm grows polynomially in the system size and particle dimension, and thus chemical accuracy is in fact a quantity that scales inverse polynomially in the system size when one considers (sub-)normalized Hamiltonians, which is often the case in the quantum computing / Hamiltonian complexity literature.

The problem of estimating the smallest eigenvalue of a local Hamiltonian up to some additive error (the decision variant of which is known as the \emph{local Hamiltonian problem}) is well-known to be $\QMA$-hard when the required accuracy scales inversely with a polynomial, where $\QMA$ is the quantum analogue of the class $\NP$, also known as Quantum Merlin Arthur. Therefore, it is generally believed that, without any additional help or structure, quantum computers are not able to accurately estimate the smallest eigenvalues of general local Hamiltonians, and there is some evidence that this hardness carries over to those Hamiltonians relevant to chemistry and materials science~\cite{o2022intractability}. A natural question to ask is then the following:  how much `extra help' needs to be provided in order to accurately estimate ground state energies using a quantum computer? 

\

\noindent In the quantum chemistry community, it is often suggested that this extra help could come from a classical heuristic that first finds some form of  \emph{guiding state}: a classical description of a quantum state that can be used as an input to a quantum algorithm to compute the ground state energy accurately~\cite{liu2022prospects}. Concretely, this comes down to the following two-step procedure~\cite{cade2023improved}:
\begin{itemize}
    \item Step~1 (Guiding state preparation): A classical heuristic algorithm is applied to obtain a \emph{guiding state} $\ket{\psi}$, which is hoped to have `good'\footnote{`Good' here means at least inverse polynomial in the number of qubits the Hamiltonian acts on.} fidelity with the ground space.
    \item Step~2: (Ground state energy approximation): The guiding state $\ket{\psi}$ is used as input to Quantum Phase Estimation (QPE) to efficiently and accurately compute the corresponding ground state energy.
\end{itemize}
Step~2 of the above procedure can be formalised by the \emph{Guided $k$-local Hamiltonian problem ($k$-GLH)}, which was introduced in~\cite{gharibian2021dequantizing} and shown to be $\BQP$-complete under certain parameter regimes that were subsequently improved and tightened in~\cite{cade2023improved}. The problem $k$-GLH is stated informally as follows:  given a $k$-local Hamiltonian $H$, an appropriate classical `representation' of a guiding state $\ket{u}$ promised to have $\zeta$-fidelity with the ground space of $H$, and real thresholds $b> a$, decide if the ground state energy of $H$ lies above or below the interval $[a,b]$. In a series of works~\cite{gharibian2021dequantizing,cade2023improved}, it was shown that $2$-GLH is $\BQP$-complete for \emph{inverse polynomial} precision and fidelity, i.e.\ $b-a\geq 1/\poly(n)$ and  $\zeta=1-1/\poly(n)$ respectively. In contrast, when $b-a \in\Theta(1)$ and $\zeta = \Omega(1)$, $k$-GLH can be efficiently solved \emph{classically} by using a dequantised version of the quantum singular value transformation. 

\

\noindent The GLH problem forms the starting point of this work. We study `\emph{Merlinized}' versions of GLH -- in which guiding states are no longer given as part of the input but instead are only promised to exist -- and use these as a way to gain some insight into important theoretical questions in quantum chemistry and complexity theory. In the subsequent paragraphs, we introduce some of the motivating questions guiding the study of the complexity of these so-called `guidable' local Hamiltonian problems.

\paragraph[title]{Ansätze\footnote{An Ansatz (plural Ansätze) is a German word often used in physics and mathematics for an assumption about the form of an unknown function or solution which is made in order to facilitate the solution of some problem. An Ansatz for state preparation in our context refers to the family of quantum states considered to be prepared on a quantum computer, for example those of matrix product states with polynomially bounded bond dimension or stabilizer states.}} for state preparation.
Step~1 of the aforementioned two-step procedure generally requires one to have access to classical heuristics capable of finding guiding states whose energies can be estimated classically (as a metric to test whether candidate states are expected to be close to the actual ground state or not). Furthermore, these `trial states' should also be preparable as quantum states on a quantum computer, so that they can be used as input to phase estimation in Step~2. In~\cite{gharibian2021dequantizing}, inspired by a line of works that focused on the dequantization of quantum machine learning algorithms~\cite{Tang2019AQuantum,Chia2020Sampling,Jethwani2020Quantum}, a particular notion of `sampling-access' to the guiding state $u$ is assumed. Specifically, it is assumed that one can both query the amplitude of arbitrary basis states, and additionally that one can sample basis states according to their $l_2$ norm with respect to the overall state $u$.\footnote{In this work we slightly abuse notation by making a distinction between the vector representing a quantum state, which we will denote as `$u$', and that same vector instantiated as a quantum state (e.g. living on a quantum computer), which we will denote by `$\ket{u}$'. Of course, these are the same mathematical object ($u = \ket{u} \in \mathcal{C}^{2^n}$), and we only use the different notation to make our theorem statements and proofs clearer.} Whilst this can be a somewhat powerful model~\cite{Jordan2021Revisiting}, it is closely related to the assumption of QRAM access to classical data, and thus in the context of quantum machine learning (where such access is commonly assumed), it makes sense to compare quantum machine learning algorithms to classical algorithms with sampling access to rule out quantum speed-ups that come merely from having access to quantum states that are constructed from exponential-size classical data. 

However, for quantum chemistry and quantum many-body applications, this type of access to quantum states seems to be somewhat artificial. From a theoretical perspective, one might wonder to what extent this sampling access model `hides' some complexity, allowing classical algorithms to perform well on the problem when they otherwise would not.

Finally, one may ask whether the fact that the ground state preparation in Step~1 considers only \emph{classical} heuristics might be too restrictive. \emph{Quantum} heuristics for state preparation, such as variational quantum eigensolvers~\cite{tilly2022the} and adiabatic state preparation techniques~\cite{albash2018adiabatic}, have contained considerable attention as possible quantum approaches within the NISQ era. However, one can argue that even in the fault-tolerant setting, such heuristics will likely still be viable approaches to state preparation, in particular when used in conjunction with Quantum Phase Estimation.

\paragraph{The quantum PCP conjecture.}
Arguably the most fundamental result in classical complexity theory is the Cook-Levin Theorem~\cite{cook1971the,levin1973universal}, which states that constraint satisfaction problems (CSPs) are $\NP$-complete. The PCP theorem~\cite{Arora1998proof,Arora1998probabilistic}, which originated from a long line of research on the complexity of interactive proof systems, can be viewed as a `strengthening' of the Cook-Levin theorem. In its proof-checking form, it states that all decision problems in $\NP$ can be decided, with a constant probability of error, by only checking a constant number of bits of a polynomially long proof string $y$ (selected randomly from the entries of $y$).   There are also alternative equivalent formulations of the PCP theorem. One is in terms of \emph{hardness of approximation}: it states that it remains $\NP$-hard to decide whether an instance of CSP is either completely satisfiable, or whether no more than a constant fraction of its constraints can be satisfied.\footnote{The transformation of a CSP to another one which is hard to approximate is generally referred to as \emph{gap amplification}, and is realised in Dinur's proof of the PCP theorem~\cite{Dinur2007thepcp}.} It is straightforward to show that this formulation is equivalent to the aforementioned proof-checking version: one simply samples a clause at random and checks whether it is satisfied, which with constant probability detects a violated clause.

Naturally, quantum complexity theorists have proposed proof-checking and hardness of approximation versions of PCP in the quantum setting. Given the close relationship between $\QMA$ and the local Hamiltonian problem, the most natural formulation is in terms of hardness of approximation: in this context, the \emph{quantum} PCP conjecture roughly states that energy estimation of a (normalized) local Hamiltonian up to \emph{constant} precision, relative to the operator norm of the Hamiltonian, remains $\QMA$-hard. This conjecture is arguably one of the most important open problems in quantum complexity theory and has remained unsolved for nearly two decades. Under the assumption $\NP \neq \QMA$, the quantum PCP conjecture implies that there exists Hamiltonians for which all low-energy states have no efficient classical description from which their energy can be evaluated efficiently classically. In a recent breakthrough result, the NLTS conjecture was proven to be true, which (amongst other things) means that constant-depth quantum circuits -- for which the energies can be computed efficiently, as shown by a standard lightcone argument -- are not expressive enough to estimate the ground state energies of all Hamiltonians up to even constant precision~\cite{anshu2022nlts}. However, there have also been some no-go results: for example,  a quantum PCP statement cannot hold for local Hamiltonians defined on a grid, nor on high-degree or expander graphs~\cite{brandao2013product}. 

One way to shed light on the validity of the quantum PCP conjecture can be to study PCP-type conjectures for other `Merlinized' complexity classes. Up until this point, PCP-type conjectures have not been considered for other classes besides $\NP$ and $\QMA$.\footnote{This is barring a result by Drucker which proves a PCP theorem for the class $\AM$~\cite{drucker2011apcp}; though there is no direct relationship between $\QMA$ and $\AM$ and hence it is not clear whether this gives any intuition about the likely validity of the quantum PCP conjecture.} However, there is the beautiful result of~\cite{Aharonov2019stoquastic}, which studies the possibility of a gap amplification procedure for the class $\MA$ by considering a particular type of Hamiltonian: uniform stoquastic local Hamiltonians. The authors show that deciding whether the energy of such a Hamiltonian is exactly zero or inverse polynomially bounded away from zero is $\MA$-hard, but that the problem is in $\NP$ when this interval is increased to be some constant. Consequently, this implies that there can exist a gap-amplification procedure for uniform stoquastic Local Hamiltonians (in analogy to the gap amplification procedure for constraint satisfaction problems in the original PCP theorem) if and only if $\MA = \NP$ -- i.e.~if $\MA$ can be derandomized. Since $\MA \subseteq \QMA$, this result also shows that if a gap amplification procedure for the general local Hamiltonian problem would exist that `preserves stoquasticity', then it could also be used to derandomize $\MA$.

\subsection{Summary of main results}\label{sec:summary_of_main_results}
\subsubsection{Completeness results for the guidable local Hamiltonian problem}
Inspired by classical heuristics that work with Ansätze to approximate the ground states of local Hamiltonians, we define a general class of states that we call \emph{classically evaluatable and quantumly preparable}.

\begin{customdef}{1.1 (Informal)}[Classically evaluatable and quantumly preparable states, from Definition~\ref{def:cds}] We say that an $n$-qubit state $u$ is classically evaluatable if
\begin{enumerate}[label=(\roman*)]
\item it has an efficient classical description which requires at most a polynomial number of bits to write down and
\item one can, given such a description, classically efficiently compute expectation values of $\mO(\log n)$-local observables of $u$.
\end{enumerate}
 In addition, we say that the state is also quantumly preparable if (iii) there exists a quantum circuit that prepares $u$ as a quantum state $\ket{u}$ using only a polynomial number of two-qubit gates.
\end{customdef}

In the main text we consider a more general version of the definition above, which also allows for probabilistic estimation of expectation values, and we provide four concrete examples of Ansätze that satisfy all three conditions: matrix product states (MPS), stabilizer states, constant-depth quantum circuits and IQP circuits~\cite{bremner2011classical}. We also relate classically evaluatable states to the samplable states of~\cite{gharibian2021dequantizing}, and show that if one allows for an error in the estimation of local observables, it forms in fact a larger class of quantum states (Theorem~\ref{thm:sampl_vs_ces}).

\ 

\noindent Our main focus is on a new family of local Hamiltonian problems, which we call \emph{Guidable local Hamiltonian problems}, in which we are promised that the ground state is close (with respect to fidelity) to some state from a particular class of states.

\begin{customdef}{1.2 (Informal)}[Guidable Local Hamiltonian problems, from Definition~\ref{def:GaLH}]  Guidable Local Hamiltonian Problems are problems defined by having the following input, promise, output and some extra promise to be precisely defined below for each of the problems separately:\\
\textbf{Input:} A $k$-local Hamiltonian $H$ with $\|H\|\leq 1$ acting on $n$ qubits, threshold parameters $a,b \in \mathbb{R}$ such that $b-a \geq \delta > 0$ and a fidelity parameter $\zeta \in (0,1]$.\\
\textbf{Promise:} We have that either $\lambda_0(H) \leq a$ or $\lambda_0(H) \geq b$ holds, where $\lambda_0(H)$ denotes the ground state energy of $H$. \\
\textbf{Extra promises:} 
Let $\Pi_\text{gs}$ be the projection on the subspace spanned by the ground states of $H$.  Then for each problem, we have that either one of the following promises hold:
\begin{enumerate}
    \item \textbf{Classically Guidable and Quantumly Preparable $k$-LH} ($\mathsf{CGaLH}^{*}(k,\delta,\zeta))$:  there exists a classically evaluatable and quantumly preparable state $u\in \mathbb{C}^{2^n}$ for which $\norm{\Pi_\text{gs} u}^2 \geq \zeta$. 
    \item \textbf{Quantumly Guidable $k$-LH} ($\mathsf{QGaLH}(k,\delta,\zeta)$): There exists a quantum circuit of polynomially many two-qubit gates that produces the state $\ket{\phi}$ for which $\norm{\Pi_\text{gs} \ket{\phi}}^2 \geq \zeta$.
\end{enumerate}
\textbf{Output:} \begin{itemize}
    \item If $\lambda_0(H) \leq a$, output {\sc yes}.
    \item If $\lambda_0(H) \geq b$, output {\sc no}.
\end{itemize}
\label{def:GaLH_inf}
\end{customdef}
\noindent A guidable local Hamiltonian problem variant for a different class of guiding states was already introduced in~\cite{gharibian2021dequantizing} without giving any hardness results. Using techniques from Hamiltonian complexity we obtain the following completeness results.\footnote{In fact $\mathsf{QGaLH}(k,\delta,\zeta)$ remains $\QCMA$-hard all the way up to $\zeta = 1$.}
\begin{customthm}{1.1 (Informal)}[Complexity of guidable local Hamiltonian problems, from Corollary~\ref{cor:CGaLHstar} and Theorem~\ref{thm:QGaLH}] For $k=2$ and $\delta=1/\poly(n)$, we have that both $\mathsf{CGaLH}^{*}(k,\delta,\zeta)$ and $\mathsf{QGaLH}(k,\delta,\zeta)$ are $\QCMA$-complete when $\zeta \in (1/\poly(n),1-1/\poly(n))$.
\end{customthm}
We also obtain similar complexity results for a guidable version of the quantum satisfiability problem (see Appendix~\ref{app:GaQSAT}).

A direct corollary of the above theorem is the following.
\begin{corollary}[Classical versus quantum state preparation] When one has access to a quantum computer (and in particular quantum phase estimation), then having the ability to prepare any quantum state preparable by a polynomially-sized quantum circuit is no more powerful than the ability to prepare states from the family of classically evaluatable and quantumly preparable states, when the task is to decide the local Hamiltonian problem with precision $1/\poly(n)$.
\label{cor:c_cs_q_stateprep}
\end{corollary}
It should be noted that our result does \emph{not} imply that all Hamiltonians which have efficiently quantumly preparable guiding states also necessarily have guiding states that are classically evaluatable. All this result says is that for any instance of the guidable local Hamiltonian problem with the promise that there exist guiding states that can be efficiently prepared by a quantum computer, there exists an (efficient) \emph{mapping} to another instance of the guidable local Hamiltonian problem with the promise that there exist guiding states that are classically evaluatable and quantumly preparable. Whilst this reduction is efficient in the complexity-theoretic sense, it might not be for practical purposes, as it would likely remove all the physical structure present in the original Hamiltonian. Hence, the main implication of our result is not that these kinds of reductions are of practical merit, but that at least from a complexity-theoretic point of view the aforementioned classical-quantum hybrid approach of guiding state selection through \emph{classical} heuristics combined with \emph{quantum} energy estimation is at least as powerful as using quantum heuristics for state preparation instead.

\ 

\noindent We complement our quantum hardness results with classical containment results (of the classically guidable local Hamiltonian problem), obtained through a deterministic dequantized version of Lin and Tong's ground state energy estimation algorithm~\cite{Lin2020nearoptimalground}. Here $\mathsf{CGaLH}$ is just as $\mathsf{CGaLH}^*$ but without the promise of the guiding state being quantumly preparable (see Definition~\ref{def:GaLH} in the main text).

\begin{customthm}{1.2 (Informal)}[Classical containment of the classically guidable local Hamiltonian problem, from Theorem~\ref{thm:cp_NP_NqP}.] Let $k = \mO(\log n)$. When $\delta$ is constant, we have that $\mathsf{CGaLH}(k,\delta,\zeta)$ is in $\NP$ when $\zeta$ is constant and is in $\NqP$ when $\zeta =  1/\poly(n)$. Here $\NqP$ is just as $\NP$ but with the Turing machine being allowed to run in quasi-polynomial time.
\end{customthm}
\noindent Through a more careful analysis of when exactly the quantum hardness vanishes, the  picture of Figure~\ref{fig:charac} emerges, which characterises the complexity of $\mathsf{CGaLH}^{*}(k,\delta,\zeta)$ for relevant parameter settings in the desired precision and promise on the fidelity. One additional result to mention, is that when the overlap between the guiding state becomes very close to one, $\zeta = 1 - 1/\exp(n)$, the problem remains in $\NP$ even when the promise gap becomes polynomially small, $\delta = 1/\poly(n)$ (Theorem~\ref{thm:CGaLH_close}).\footnote{After this work, Jiang published a work on a similar problem for a different class of states then we consider~\cite{jiang2023local}. Jiang shows that if the ground state admits a polynomial-time algorithm to compute the amplitudes, the class of states for which this is possible being called \emph{succinct}, the corresponding local Hamiltonian problem is $\MA$-complete even in the inverse polynomial precision setting. However, since our proof of Theorem~\ref{thm:CGaLH} uses as a witness polynomial-sized subset states, which are succinct, it also shows that Jiang's problem is $\QCMA$-hard when it is only promised that there exists a succinct state with only at most $1-1/\poly(n)$ overlap with a ground state. Hence, Jiang's result is similar to our Theorem~\ref{thm:CGaLH_close} in that if the ground state itself becomes (exponentially close to) a special class of states, the problem becomes classically solvable.}

\tikzset{every picture/.style={line width=0.75pt}} %set default line width to 0.75pt        
\begin{figure}[h!]
    \centering
\begin{tikzpicture}[x=0.75pt,y=0.75pt,yscale=-0.85,xscale=0.85]
%uncomment if require: \path (0,302); %set diagram left start at 0, and has height of 302

%Shape: Rectangle [id:dp26956949296292454] 
\draw  [color={rgb, 255:red, 155; green, 155; blue, 155 }  ,draw opacity=0.53 ][fill={rgb, 255:red, 227; green, 12; blue, 12 }  ,fill opacity=0.28 ] (105,26) -- (205.2,26) -- (205.2,119.25) -- (105,119.25) -- cycle ;
%Shape: Rectangle [id:dp26413029258527865] 
\draw  [color={rgb, 255:red, 155; green, 155; blue, 155 }  ,draw opacity=0.53 ][fill={rgb, 255:red, 227; green, 12; blue, 12 }  ,fill opacity=0.53 ] (105,119.25) -- (205.2,119.25) -- (205.2,212.5) -- (105,212.5) -- cycle ;
%Shape: Rectangle [id:dp5507751890577569] 
\draw  [color={rgb, 255:red, 155; green, 155; blue, 155 }  ,draw opacity=0.53 ][fill={rgb, 255:red, 0; green, 124; blue, 255 }  ,fill opacity=0.44 ] (205.03,26.18) -- (305.23,26.18) -- (305.23,119.43) -- (205.03,119.43) -- cycle ;
%Shape: Rectangle [id:dp7874685527382818] 
\draw  [color={rgb, 255:red, 155; green, 155; blue, 155 }  ,draw opacity=0.53 ][fill={rgb, 255:red, 255; green, 93; blue, 0 }  ,fill opacity=0.32 ] (205.2,119.25) -- (305.39,119.25) -- (305.39,212.5) -- (205.2,212.5) -- cycle ;
%Shape: Rectangle [id:dp5500016084956724] 
\draw  [color={rgb, 255:red, 155; green, 155; blue, 155 }  ,draw opacity=0.53 ][fill={rgb, 255:red, 0; green, 124; blue, 255 }  ,fill opacity=0.22 ] (305.39,26) -- (405.59,26) -- (405.59,119.25) -- (305.39,119.25) -- cycle ;
%Shape: Rectangle [id:dp9658879247772332] 
\draw  [color={rgb, 255:red, 155; green, 155; blue, 155 }  ,draw opacity=0.53 ][fill={rgb, 255:red, 0; green, 124; blue, 255 }  ,fill opacity=0.22 ] (405.59,26) -- (505.78,26) -- (505.78,119.25) -- (405.59,119.25) -- cycle ;
%Shape: Rectangle [id:dp4138237703258706] 
\draw  [color={rgb, 255:red, 155; green, 155; blue, 155 }  ,draw opacity=0.53 ][fill={rgb, 255:red, 0; green, 124; blue, 255 }  ,fill opacity=0.22 ] (505.78,26) -- (605.98,26) -- (605.98,119.25) -- (505.78,119.25) -- cycle ;
%Shape: Rectangle [id:dp4529462179541187] 
\draw  [color={rgb, 255:red, 155; green, 155; blue, 155 }  ,draw opacity=0.53 ][fill={rgb, 255:red, 0; green, 124; blue, 255 }  ,fill opacity=0.22 ] (505.62,119.25) -- (605.81,119.25) -- (605.81,212.5) -- (505.62,212.5) -- cycle ;
%Shape: Right Triangle [id:dp5748768002775084] 
\draw  [color={rgb, 255:red, 155; green, 155; blue, 155 }  ,draw opacity=0.53 ][fill={rgb, 255:red, 255; green, 93; blue, 0 }  ,fill opacity=0.32 ] (405.42,119.43) -- (505.62,212.68) -- (405.42,212.68) -- cycle ;
%Shape: Right Triangle [id:dp619094024006891] 
\draw  [color={rgb, 255:red, 155; green, 155; blue, 155 }  ,draw opacity=0.53 ][fill={rgb, 255:red, 0; green, 124; blue, 255 }  ,fill opacity=0.22 ] (505.78,212.5) -- (405.42,119.43) -- (505.62,119.25) -- cycle ;
%Shape: Rectangle [id:dp4689015701122713] 
\draw  [color={rgb, 255:red, 155; green, 155; blue, 155 }  ,draw opacity=0.53 ][fill={rgb, 255:red, 255; green, 93; blue, 0 }  ,fill opacity=0.32 ] (305.23,119.43) -- (405.42,119.43) -- (405.42,212.68) -- (305.23,212.68) -- cycle ;

% Text Node
\draw (124,150.9) node [anchor=north west][inner sep=0.75pt]  [font=\LARGE]  {$\mathsf{QMA}$};
% Text Node
\draw (115,65.9) node [anchor=north west][inner sep=0.75pt]  [font=\small]  {$\mathsf{QPCP}[\mathcal{O}( 1)]$};
% Text Node
\draw (230,59.9) node [anchor=north west][inner sep=0.75pt]  [font=\Large]  {$\mathsf{NqP}^\dagger$};
% Text Node
\draw (275.5,150.4) node [anchor=north west][inner sep=0.75pt]  [font=\LARGE]  {$\mathsf{QCMA}$};
% Text Node
\draw (485,103.9) node [anchor=north west][inner sep=0.75pt]  [font=\LARGE]  {$\mathsf{NP}$};
% Text Node
\draw (214,87) node [anchor=north west][inner sep=0.75pt]  [color={rgb, 255:red, 100; green, 100; blue, 100 }  ,opacity=1 ] [align=left] {{\footnotesize Theorem~\ref{thm:cp_NP_NqP}}};
% Text Node
\draw (272,179) node [anchor=north west][inner sep=0.75pt]  [color={rgb, 255:red, 100; green, 100; blue, 100 }  ,opacity=1 ] [align=left] {{\footnotesize Corollary~\ref{cor:CGaLHstar}}};
% Text Node
\draw (433,131) node [anchor=north west][inner sep=0.75pt]  [color={rgb, 255:red, 100; green, 100; blue, 100 }  ,opacity=1 ] [align=left] {{\footnotesize Theorems~\ref{thm:cp_NP_NqP} and~\ref{thm:CGaLH_close}}};
% Text Node
\draw (342,223.4) node [anchor=north west][inner sep=0.75pt]    {$\Omega( 1)$};
% Text Node
\draw (55,68.4) node [anchor=north west][inner sep=0.75pt]    {$\Omega( 1)$};
% Text Node
\draw (405,223.4) node [anchor=north west][inner sep=0.75pt]  [font=\small]  {$1-1/\mathrm{poly}( n)$};
% Text Node
\draw (513,223.4) node [anchor=north west][inner sep=0.75pt]  [font=\small]  {$1-1/\mathrm{exp}( n)$};
% Text Node
\draw (222,223.4) node [anchor=north west][inner sep=0.75pt]  [font=\small]  {$1/\mathrm{poly}( n)$};
% Text Node
\draw (122,223.4) node [anchor=north west][inner sep=0.75pt]  [font=\small]  {$1/\mathrm{exp}( n)$};
% Text Node
\draw (33,157.4) node [anchor=north west][inner sep=0.75pt]  [font=\small]  {$1/\mathrm{poly}( n)$};
% Text Node
\draw (324,253.4) node [anchor=north west][inner sep=0.75pt]  [font=\large]  {$\mathrm{Fidelity} \ \zeta $};
% Text Node
\draw (5.5,185.44) node [anchor=north west][inner sep=0.75pt]  [font=\large,rotate=-270.36]  {$\mathrm{Promise\ gap} \ \delta $};
\end{tikzpicture}
    \caption{Complexity characterization of $\mathsf{CGaLH}^{*}(k,\delta,\zeta)$ over parameter regime $\delta$ and $\zeta$, for $k = \mO(1)$. Any classification indicates completeness for the respective complexity class, except for $\NqP$, for which we only know containment (indicated by the $`\dagger'$). Here completeness for certain parameter combinations means that for all functions of the indicated form, the problem is contained in the complexity class, and for a subset of these functions the problem is also hard. The results for $\QPCP[\mO(1)]$ and $\QMA$ follow directly from~\cite{Aharonov2008The} and~\cite{Kitaev2002ClassicalAQ}.}
    %\jordi{Modify this picture such that it includes (clickable?) theorem numbers.}\chris{If you know how to do that, go ahead :)}\jordi{Added it, need to write down the references for QPCP and QMA}
    \label{fig:charac}
\end{figure}
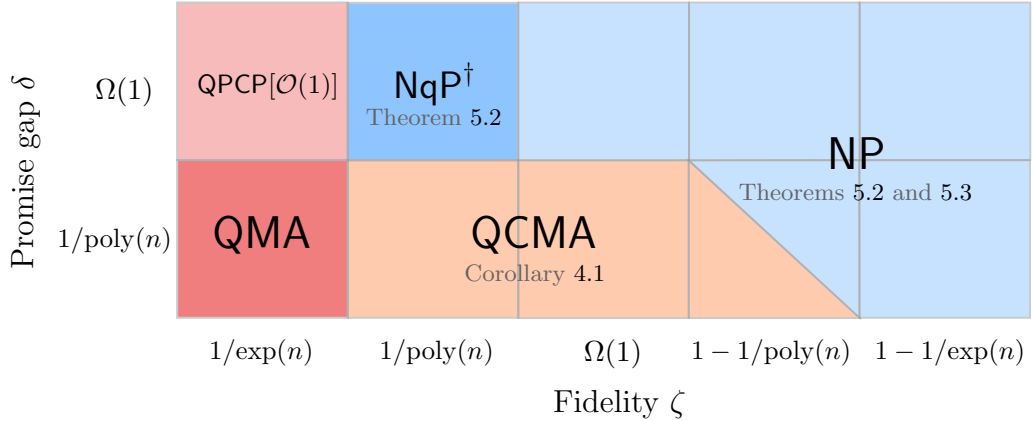

\subsubsection{Quantum-classical probabilistically checkable proofs}
We introduce the notion of a \emph{quantum-classical probabilistically checkable proof system} in the following way.

\begin{customdef}{1.3 (Informal)}[Quantum-classical PCP, from Definition~\ref{def:QCPCP}] 
 A $\QCPCP[q]$ protocol consists of a polynomial-time quantum verifier $V$ that uses $\poly(n)$ ancilla qubits and is given an input $x \in \{0,1\}^n$ and a classical proof $y \in \{0,1\}^{p(n)}$, where $p(n) \leq \poly(n)$, from which it queries at most $q(n)$ bits non-adaptively. The verifier measures the first qubit and accepts only if the outcome is $\ket{1}$. A promise problem $A = (A_\text{yes},A_\text{no})$ belongs to the class $\QCPCP[q]$ if it has a $\QCPCP[q]$ verifier system with the following properties
\begin{itemize}
\item[~] \textbf{Completeness.} If $x\in A_\text{yes}$, then there exists a classical proof $y$ such that the verifier accepts with probability at least $2/3$.
\item[~] \textbf{Soundness.} If $x \in A_\text{no}$, then for all classical proofs $y$ the verifier accepts with probability at most $1/3$. 
\end{itemize}
\label{def:QCPCP_inf}
\end{customdef}
\noindent Note that we have that $\QCPCP[\mO(1)]$ trivially contains $\mathsf{CGaLH}$ with any constant promise gap and any constant fidelity,\footnote{Recall that the problem is $\QCMA$-hard when the promise gap is inverse polynomial in the number of Hamiltonian terms instead, even when the fidelity is constant (but $< 1$).} since we have shown that this problem is in $\NP$ (Theorem~\ref{thm:cp_NP_NqP}) and therefore admits a classical PCP system to solve the problem. Note that replacing the classically evaluatable states with samplable states in the promise, as in~\cite{gharibian2021dequantizing}, does not necessarily mean that the problem is in $\QCPCP [\mO(1)]$ as we do not know whether $\MA$ admits a (quantum-classical) PCP.

\

\noindent We first prove two basic facts about $\QCPCP$: we show that it allows for weak error reduction (Proposition~\ref{prop:error_reduction}) and that the non-adaptiveness restriction does not limit the power of the class when the number of proof queries is constant (Theorem~\ref{thm:QCPCP_A}). Our `quantum-classical PCP conjecture' then posits that $\QCPCP[\mO(1)] = \QCMA$, analogously to the quantum PCP conjecture which states that $\QPCP[\mO(1)] = \QMA$ (here $\QPCP[q]$ denotes the complexity class associated with \textit{quantum} probabilistically checkable proof systems).\footnote{The question of whether a PCP can be shown for QCMA was also raised briefly in~\cite{bittel2022optimizing}.}  

\

\noindent Our main result regarding $\QCPCP[\mO(1)]$ is that we can provide a non-trivial upper bound on the complexity of the class.
\begin{customthm}{1.3 (Informal)}[Upper bound on $\QCPCP$, from Theorem~\ref{thm:QCPCP_NP_qr}] 
\begin{align*}
    \QCPCP[\mO(1)] \subseteq \BQP^{\NP[1]}.
\end{align*}
\end{customthm}
Here $\BQP^{\NP[1]}$ is the class of all problems that can be solved by a $\BQP$-verifier that makes a single query to an $\NP$-oracle. The key idea behind the proof is that a quantum reduction can be used to transform a $\QCPCP$ verification circuit to a local Hamiltonian that is \textit{diagonal} in the computational basis, and thus can be solved with a single query to an $\NP$ oracle. Using this upper bound, we then show that if our quantum-classical PCP conjecture is true, then $\NP^\BQP \subseteq \BQP^\NP$. This inclusion would have the consequence that $\NP \subseteq \BQP$ implies $\PH \subseteq \BQP$, i.e.~that if $\NP$ is contained in $\BQP$ then so is the entire polynomial hierarchy.\footnote{This would contrast the result by~\cite{aaronson2021acrobatics} which shows that $\NP^\BQP \not\subset \BQP^\NP$ relative to an oracle. However, since our inclusion goes via $\QCMA \subseteq \QCPCP[\mO(1)]$, which would likely require non-relativizing techniques just as was the case for the classical PCP Theorem, the conjecture and this result could simultaneously be true.} Such a result would provide strong evidence that quantum computers are indeed not capable of solving $\NP$-hard problems.

\subsubsection{Three implications for the quantum PCP conjecture}
Finally, we use our obtained results on $\QCPCP$ and $\mathsf{CGaLH}$ to obtain two interesting results and a new conjecture with respect to the quantum PCP conjecture. First, we give evidence that it is unlikely that there exists a \emph{classical} reduction from a $\QPCP$-system (see Definition~\ref{def:QPCP} for a formal definition) to a local Hamiltonian problem with a constant promise gap having the same properties as the known \textit{quantum} reduction (see for example~\cite{Grilo2018thesis}), unless $\BQP \subseteq \QCPCP[\mO(1)] \subseteq \NP$, something that is not expected to hold~\cite{aaronson2010bqp,raz2022oracle}.
\begin{customthm}{1.4 (Informal)}[No-go for classical polynomial-time reductions, from Theorem~\ref{thm:clas_red_no_go}] For any $\epsilon <1/6 $ there cannot exist a classical polynomial-time reduction from a $\QPCP[\mO(1)]$ verification circuit $V$ to a local Hamiltonian $H$ such that, given a proof $\ket{\psi}$,
\begin{align*}
    \abs{\mathbb{P}[V \text{ accepts } \ket{\psi}] - \left(1 - \bra{\psi}H\ket{\psi}\right)}\leq \epsilon,
\end{align*}
unless $\QCPCP[ \mO(1)] \subseteq \NP$ (which would imply $\BQP \subseteq \NP$).
\end{customthm}
This provides strong evidence that allowing for reductions to be quantum is indeed necessary to show equivalence between the gap amplification and proof verification formulations of the quantum PCP conjecture~\cite{aharonov2013guest}.

\

\noindent Second, our classical containment results of $\mathsf{CGaLH}^{*}$ with constant promise gap can be viewed as no-go theorems for a gap amplification procedure for $\QPCP$ having certain properties, as illustrated by the following result.
\begin{customthm}{1.5 (Informal)}[No-go results for Hamiltonian gap amplification, from Theorem~\ref{thm:QCPCP_nogos}]
There cannot exist a gap amplification procedure for the local Hamiltonian problem that preserves the
fidelity between the ground space of the Hamiltonian and any classically evaluatable state up to a
\begin{itemize}
    \item multiplicative constant, unless $\QCMA = \NP$, or
    \item multiplicative inverse polynomial, unless $\QCMA \subseteq \NqP$.
\end{itemize}
\end{customthm}
This result is analogous to the result of~\cite{Aharonov2019stoquastic}, which rules out a gap amplification procedure that preserves stoquasticity under the assumption that $\MA \neq \NP$.\footnote{Or taking a different view, proving the existence of such gap amplifications would allow one to simultaneously prove that $\MA$ can be derandomized (or even $\RP$ if it exhibits some additional properties)~\cite{Aharonov2019stoquastic}.} Moreover, we point out that many Hamiltonian gadget constructions \textit{do} satisfy such fidelity-preserving conditions, and indeed are precisely those that were used in~\cite{cade2023improved} to improve the hardness results for the guided local Hamiltonian problem.\footnote{For a quantum version of gap amplification one would typically expect locality-reducing Hamiltonian gadgets as part of the procedure, to compensate for a ``powering step'' which consists of taking powers of the Hamiltonian (which therefore increases locality). It is already known that the current best-known locality-reducing gadgets~\cite{bravyi2008quantum} cannot be used because they increase the norm of the Hamiltonian by a constant factor, which results in an unmanageable decrease of the relative promise gap. Our result shows that even if one would find better constructions that don't have this effect, they would still have to satisfy the additional constraints as described in Theorem~\ref{thm:QCPCP_nogos}.} We obtain similar results for the class $\MA$ by considering a variant of $\mathsf{CGaLH}$ that restricts the Hamiltonian to be stoquastic (Appendix~\ref{app:MA}). 

\ 

\noindent Third, we can use our results to formulate a stronger version of the NLTS theorem (and an alternative to the NLSS conjecture~\cite{gharibian2021dequantizing}), which we will call the \emph{No Low-energy Classically evaluatable States conjecture}. This conjecture can hopefully provide a new stepping stone towards proving the quantum PCP conjecture. 

\begin{customconj}{1.1 (Informal)}[NLCES conjecture, from Conjecture~\ref{conj:NLCES}] There exists a family of local Hamiltonians $\{H_n\}_{n \in \mathbb{N}}$ on $n$ qubits, and a constant $\beta > 0$, such that for sufficiently large $n$ for every classically evaluatable state $u \in \mathbb{C}^{2^n}$ as per Definition~\ref{def:cds}, we have that
\[
\bra{u}H_n\ket{u} \geq \lambda_0(H_n) + \beta\,.
\]
\label{conj:NLCES_inf}
\end{customconj}
Just as is the case for the NLSS conjecture and the NLTS theorem, the NLCES conjecture would, if proven to be true, not necessarily imply the quantum PCP conjecture. For example, it might be that there exist states that can be efficiently described classically but for which computing expectation values is hard (just as, for example, tensor network contraction is $\mathsf{\# P}$-hard in the worst case~\cite{schuch2007computational, biamonte2015tensor}). Furthermore, as we have shown in this work, states with high energy but also a large fidelity with the ground state suffice as witnesses to decision problems on Hamiltonian energies, and these would not be excluded by a proof of the NLCES conjecture above. To make this more concrete, we also formulate an even stronger version of the NLCES conjecture, which states that there must be a family of Hamiltonians, for which no classically evaluatable state has good fidelity with the low energy spectrum (Conjecture~\ref{conj:strong_NLCES}). \\

Our results can also be interpreted as tests of the `robustness' of the definitions of quantum PCP conjectures with respect to adding some extra notion of `classicallity', either in the proof-checking or local Hamiltonian formulation. In the local Hamiltonian formulation, the existence of classically evaluatable states as guiding states makes complexity drop from $\QCMA$ to $\NP$ (or $\NqP$) when the promise gap goes from inverse polynomial to constant. In the proof-checking formulation, making the proof classical allows one to show containment for the corresponding class of quantum-classical PCP systems in $\BQP^{\NP}$, a class not known to contain $\QCMA$.

\subsection{Overview of techniques}

\paragraph{QCMA-hardness proof for guidable local Hamiltonian problems.}
We follow a similar proof structure as used in the $\BQP$-hardness proofs of the Guided Local Hamiltonian problem~\cite{gharibian2021dequantizing,cade2023improved}.There are several obstacles which prevent one from directly adopting the same proof in the $\QCMA$ setting, i.e.~when starting with a $\QCMA$ verification circuit $U$. This mostly comes down to the fact that $U$, unlike a $\BQP$-circuit, has an additional input register for the witness. This creates many valid `history states' (which are $0$-eigenvectors of the Hamiltonian $H_{\text{in}} + H_{\text{clock}} + H_{\text{prop}}$), giving us less control over and knowledge about the actual ground state of the Hamiltonian generated by the circuit-to-Hamiltonian construction.  To work around this, we use several tricks in our new construction. First, we use the CNOT-trick, introduced in~\cite{Wocjan2003two}, to `force' all witnesses to be classical. Second, use the result by~\cite{Aharonov2022pursuitof} which shows that there exists a randomized reduction from a $\QCMA$ protocol with verification circuit $U$ to one in which the verification circuit $\tilde{U}$ such that there exists a unique accepting witness in the {\sc yes}-case. Next, we apply the small-penalty circuit-to-Hamiltonian mapping of~\cite{Deshpande2020} which, together with error reduction on the verification circuit $U$, gives us fine control over the bounds on the energies in the low-energy subspace of the Hamiltonian. Combining this with the aforementioned randomized reduction, we find that the ground space of $H$ is now one-dimensional and can be made to have exponentially close fidelity with the history state corresponding to the uniquely accepting witness in the {\sc yes}-case. This allows us to construct a corresponding polynomial-sized subset state (which we show to be classically evaluatable and quantumly preparable) that has good fidelity with this history state, and use this as our guiding state. We also apply the pre-idling and block-encoding tricks from~\cite{gharibian2021dequantizing} to increase fidelity with the guiding state and to handle the {\sc no}-case, respectively. Finally, by using locality-reducing Hamiltonian gadget constructions that preserve the `classical evaluatibility' of the guiding state, we arrive at our final result.

\paragraph{A deterministic spectral amplification algorithm.} Our classical algorithm is inspired by the techniques developed in~\cite{gharibian2021dequantizing}, which in a more general setting dequantizes the quantum singular value transformation~\cite{gilyen2019quantum} for sparse matrices. Our technique can essentially be viewed as a dequantization of Lin and Tong's ground state energy estimation algorithm~\cite{Lin2020nearoptimalground}: here one assumes access to a unitary $U_H$ that implements a block-encoding of $H$. Since $H$ is Hermitian, a polynomial function applied to $H$ can be viewed as acting on the eigenvalues of $H$. An approximate low-energy subspace projector can then be constructed using a polynomial which approximates the sign function, using a result from~\cite{LowHamiltonian2017}. We construct a similar algorithm, but this time in a classical deterministic setting, where we assume the input states are of the form of classically evaluatable states (see Section~\ref{sec:summary_of_main_results}). We measure the complexity of our dequantization algorithm by counting the number of expectation values of local observables that have to be computed, which follows straightforwardly counting the number of Hermitian terms in which the polynomial approximation we consider can be expanded. Finally, we derive the complexity of the algorithm when it is applied to solving the Hamiltonian energy decision problems considered in this paper.

\paragraph{Upper bounds on quantum-classical PCPs} Our proof of $\QCPCP[q] \subseteq \BQP^{\NP[1]}$ for $q \in \mO(1)$ uses a quantum reduction from a $\QCPCP$-verifier $V$ to a \emph{diagonal} local Hamiltonian problem. The $\BQP$-verifier can then perform the reduction with high success probability, and -- if the reduction succeeded -- solve the diagonal local Hamiltonian problem with a single query to the $\NP$-oracle.
At its core, the quantum reduction iterates over all possible local configurations of the proof $y$ (which are all $2^q$ bit strings) and runs $V$ to collect information on both (i) the likeliness that the proof should be queried at certain index locations and (ii) that it accepts seeing a certain local configuration of a proof. By careful analysis of the error bounds one can use results from learning theory to bound the number of runs of $V$ that should be performed to learn the Hamiltonian up to a desired precision in operator norm.

\subsubsection{Relation to previous work}
The starting point of this paper is the \textit{guided} local Hamiltonian problem, introduced in~\cite{gharibian2021dequantizing}. Our work diverges from theirs in two principle directions: 1) whereas their work focuses predominantly on the case in which a guiding state is given as a part of the input, we focus here on the `guidable' version of the problem -- i.e.~when a guiding state is only \textit{promised} to exist\footnote{This was briefly touched upon in~\cite{gharibian2021dequantizing}, where it was shown that the local Hamiltonian problem for all Hamiltonians whose ground space has constant fidelity with samplable states can be estimated up to constant precision in $\MA$, but without any hardness results.}; and 2) we consider a more general and natural family of guiding states, namely what we term \textit{classically evaluatable} states. 

Direction 1) allows us to introduce and consider the idea of a quantum-classical PCP (QCPCP) conjecture and, combined with our results on the hardness of the guidable local Hamiltonian problem, obtain results about the relationship between problems admitting QCPCPs and other computational complexity classes. The move from `guided' to `guidable' here is necessary: without the notion of a witness, a PCP framework cannot be considered.

Direction 2) puts our results in a more general setting, and in particular one that is somewhat more relevant to questions surrounding the application of quantum computers to hard problems in chemistry and physics. In particular, as we show in Section~\ref{sec:ces_and-GaLH}, our notion of classically evaluatable states captures many Ansätze commonly used for estimating ground state energies of physically relevant Hamiltonians. Moreover, the set of $\epsilon$-classically evaluatable states captures a larger set of states than those that are sampleable, as we discuss in Section~\ref{sec:ces_and-GaLH}. Moving to this class of states allows us to weaken the fairly stringent assumptions of the original guided local Hamiltonian problem.

In \cite{Wocjan2003two}, the authors consider a $\QCMA$-complete problem of a similar flavour to ours, namely deciding whether a 3-local Hamiltonian has low energy states that can be prepared using a polynomially bounded number of elementary quantum gates. Apart from the specificity of the requirement (preparable via polynomial-time quantum circuits vs. classically evaluatable), this differs in an important way from the type of constraint that we consider in this work: that the requirement on the ground states of the Hamiltonian is regarding their \textit{fidelity} with a particular class of quantum states, not that they themselves belong to that class. We elaborate more on these differences at the end of Section~\ref{sec:ces_and-GaLH}.

\subsection{Open questions and future work}

\paragraph{A non-trivial lower bound on quantum PCP.}
A trivial lower bound on the computational power of quantum PCP is $\NP$, which follows from the standard PCP theorem. Our formulation of the $\QCPCP$ provides a way to prove the first non-trivial lower bound on quantum PCP. Since the proofs in $\QCPCP$ are always (or can be forced to be) classical, one might hope to do this by using some of the techniques used to prove (formulations of the) PCP theorem, like exponentially long PCPs, PCPs of proximity, alphabet reduction etc., which could carry over more easily to the $\QCPCP$ setting as compared to $\QPCP$.

\paragraph{Proof checking versus the local Hamiltonian formulations of quantum PCPs.}
Another obvious lower bound to $\QPCP$ (or $\QCPCP$) comes from $\BQP \subseteq \QPCP$, since the verifier can simply ignore the proof. However, the relationship between $\BQP$ and $\NP$ is very much unclear: it is generally believed that for both classes there exist problems that are exclusively contained in only one of them. In this work we show that it is unlikely that there exists a classical reduction from a $\QPCP$ verifier circuit to a local Hamiltonian problem with a constant promise gap that has the same properties as the known quantum reduction. This means that it is entirely possible that the generic local Hamiltonian problem with constant promise gap is contained $\NP$, whilst the \textit{proof checking} version of $\QPCP$ is not, provided that the quantum reduction from the proof checking formulation to the local Hamiltonian problem can indeed not be `dequantized'.\footnote{Indeed, it could be that the local Hamiltonian problem with constant promise gap is contained in some complexity class $\mathsf{C}$. Then so long as $\BQP \not\subset \mathsf{C}$, it is possible that the proof checking version of $\QPCP$ is strictly more powerful than the local Hamiltonian version (i.e.~$\QPCP \not\subset \mathsf{C}$), since the quantum reduction cannot necessarily be performed `inside' $\mathsf{C}$.} That is, despite results that show `equivalence' of the proof-checking and local Hamiltonian variants of the quantum PCP conjecture, the two variants could actually have quite different computational power since equivalence is shown only under quantum reductions. It would be interesting to explore the possibility of different complexities for the proof checking and local Hamiltonian variants of $\QPCP$ further.

\paragraph{The (strong) NLCES conjecture.} It would be interesting to see whether the family of Hamiltonians used to prove the well-known NLTS conjecture, or constructions inspired by the proof thereof (in particular Hamiltonians that arise from error-correcting codes), can also be used to prove (weaker versions of) our NLCES conjecture (see Conjecture~\ref{conj:NLCES}). Note that our NLCES conjecture is strictly stronger than NLTS, since it includes all states that can be prepared by constant depth quantum circuits (i.e.~those states covered by the NLTS conjecture), but also includes states that require super-constant quantum depth, for example arbitrary Clifford circuits\footnote{This has in fact recently been proven for Clifford circuits, see~\cite{coble2023local}.}, matrix-product states, etc.

\paragraph{$\MA$ containment of guidable stoquastic LH.} It is well-known that for stoquastic Hamiltonians, deciding if the ground state energy is $\leq a$ or $\geq b$ with $b-a = 1/\poly(n)$ is $\StoqMA$-complete, for arbitrary $b \geq a$ inverse polynomially separated, but $\MA$-complete when $a=0$ and $b = 1/\poly(n)$~\cite{bravyi2006,bravyi2010complexity}. In~\cite{BravyiMonte2015}, it is shown that for a much stronger type of assumption on the existence of a guiding state than what we consider, the problem is also $\MA$-complete for arbitrary $b \geq a$ inverse polynomially separated. Showing $\MA$-containment for our definition of \emph{guidable} stoquastic local Hamiltonian problems (with arbitrary $a,b$, inverse polynomially separated) could provide a way to study the exact relationship between $\StoqMA$ and $\MA$.

\paragraph{The classical guiding state existence assumption.} As discussed in~\cite{cade2023improved}, the existence of practical quantum advantage based on the previously mentioned two-step procedure is only expected if there exist guiding states, quantum or classical, that have not too much (exponentially close) but also not too little (exponentially small) fidelity with the ground space of the Hamiltonian under study.  Whilst there is some literature that (partially) explores this direction~\cite{vonBurgquantum,2018arXiv180905523T,lee2023evaluating}, it would be useful and interesting to study this assumption in the special case of Ansätze that describe classically evaluatable and quantumly preparable states. This could provide numerical evidence to support the results that we have shown from a complexity-theoretic perspective: that classical heuristics combined with quantum phase estimation is indeed the right way to approach fault-tolerant quantum advantage in chemistry.

\subsection*{Guide for readers}
This work can be viewed as a collection of results related to `toned-down' versions of the quantum PCP conjecture, both in its proof verification and local Hamiltonian problem formulation, where all results are unified when their implications to the actual quantum PCP conjecture are discussed in Section~\ref{sec:PCP}. This work introduces several new concepts, and for readers only interested in specific concepts the following guide can be used:
\begin{enumerate}
    \item For those interested in the introduced class of classically evaluatable states and its relation to other classes of quantum states, consult Section~\ref{subsec:ces}.
    \item For guidable local Hamiltonian problems, the general definition can be found in Section~\ref{subsec:GaLH_def}, the $\QCMA$-hardness proofs in Section~\ref{sec:QCMA-hardness} and classical containment results in Section~\ref{sec:clas_cont}.
    \item For the definition and results regarding quantum-classical PCPs, see Section~\ref{sec:QCPCP}.
\end{enumerate}

\section{Preliminaries}
\label{sec:prelim}
\subsection{Notation}

We write $\lambda_i(A)$ to denote the $i$th eigenvalue of a Hermitian matrix $A$, ordered in non-decreasing order, with $\lambda_0(A)$ denoting the smallest eigenvalue (ground state energy). When we write $\norm{\cdot}$ we refer to the operator norm when its input is a matrix, and Euclidean norm for a vector. 

\subsection{Some basic definitions and results from complexity theory}

Let us first recall a couple of basic definitions and results from (quantum) complexity theory, which is central to this work. All complexity classes will be defined with respect to promise problems (and not languages). To this end, we take a (promise) problem $A = (A_\text{yes},A_\text{no})$ to consist of two non-intersecting sets $A_\text{yes},A_\text{no} \subseteq \{ 0,1\}^*$ (the {\sc yes} and {\sc no} instances, respectively). We have that $A_\text{inv} = \{0,1\}^* \setminus A_\text{yes} \cup A_\text{no}$ is the set of all invalid instances, and we do not care how a verifier behaves on problem instances $x \in A_\text{inv}$ (it can accept or reject arbitrarily, see the paragraph `oracle access' for a more elaborate discussion on what this entails).

\begin{definition}[$\Pclass$] A promise problem $A = (A_\text{yes},A_\text{no})$ is in $\Pclass$ if and only if there exists a deterministic polynomial-time Turing machine $M$ which takes as input a string $x\in \{0,1\}^*$ and decides on acceptance or rejection of $x$ such that
\begin{itemize}
    \item if $x \in A_\text{yes} $ then $M$ accepts $x$.
    \item if $x \in A_\text{no}$ then $M$ rejects $x$.
\end{itemize}
\end{definition}

\begin{definition}[$\NP$] A promise problem $A = (A_\text{yes},A_\text{no})$ is in $\NP$ if and only if there exists a deterministic polynomial-time Turing machine $M$ and a polynomial $p$, where $M$ takes as input a string $x\in \{0,1\}^*$ and a $p(|x|)$-bit witness $y$ and decides on acceptance or rejection of $x$ such that
\begin{itemize}
    \item if $x \in A_\text{yes} $ then there exists a $y \in \{0,1\}^{p(n)}$ such that $M$ accepts $(x,y)$.
    \item if $x \in A_\text{no}$ then for every $y \in \{0,1\}^{p(n)}$ we have that $M$ rejects $(x,y)$.
\end{itemize}
\end{definition}

\begin{subdefinition}[$\NqP$] If the above Turing machine $M$ is instead allowed to run in quasi-polynomial time, i.e.~$2^{\mO(\log^c(n) )}$ for some constant $c >0$, $A$ is in $\NqP$ (Non-deterministic quasi-Polynomial time).
\end{subdefinition}

\begin{definition}[$\MA$]
A promise problem $A = (A_\text{yes},A_\text{no})$ is in $\MA[c,s]$ if and only if there exists a probabilistic polynomial-time Turing machine $M$ and a polynomial $p$, where $M$ takes as input a string $x\in \{0,1\}^*$ and a $p(|x|)$-bit witness $y$ and decides on acceptance or rejection of $x$ such that
\begin{itemize}
    \item if $x \in A_\text{yes} $ then there exists a $y \in \{0,1\}^{p(n)}$ such that $M$ accepts $(x,y)$ with probability $\geq c$,
    \item if $x \in A_\text{no}$ then for every $y \in \{0,1\}^{p(n)}$ we have that $M$ accepts $(x,y)$ with probability $\leq s$,
\end{itemize}
where $c-s = 1/\poly(n)$. When $c=2/3$ and $s=1/3$ we omit the $[c,s]$ notation and call the class $\MA$.
\end{definition}
\begin{subdefinition}[$\UMA$]
 The class $\UMA [c,s]$ has the same definition as $\MA$ but with the extra constraint that if $x\in A_\text{yes} $ then there exists only a single $y^*$ such that $M$ accepts $(x,y^*)$ with probability $\geq c(=2/3)$, and otherwise for all $y \neq y^*$ we have that $M$ accepts $(x,y)$ with probability $\leq s(=1/3)$.
\end{subdefinition}

All quantum complexity classes will be defined in the quantum circuit model.\footnote{In order to capture the fact that the length of the inputs $x$ is allowed to vary, whilst an input size to a given circuit is fixed, one considers the notion of a polynomial-time uniform family of quantum circuits. Specifically, one says that a set of quantum circuits $\{V_n\}$ is a polynomial-time uniform if there exists a polynomial-time deterministic Turing machine, which on input $1^n$, outputs a description of the circuit $V_n$.}

\begin{definition}[$\QMA$] 
A promise problem $A = (A_\text{yes},A_\text{no})$ is in $\QMA$[c,s] if and only if there exists a polynomial-time uniform family of quantum circuits $\{V_n\}$ and a polynomial $p$, where $V_n$ takes as input a string $x\in \{0,1\}^*$ with $|x|=n$, and a $p(n)$-qubit witness quantum state $\ket{\psi}$ and decides on acceptance or rejection of $x$ such that
\begin{itemize}
    \item if $x \in A_\text{yes} $ then there exists a witness state $\ket{\psi} \in \left( \mathbb{C}^2\right)^{\otimes p(n)}$ such that $V_n$ accepts $(x,\ket{\psi})$ with probability $\geq c$,
    \item if $x \in A_\text{no}$ then for every witness state $\ket{\psi} \in \left( \mathbb{C}^2\right)^{\otimes p(n)}$, $V_n$ accepts $(x,\ket{\psi})$ with probability $\leq s$,
\end{itemize}
where $c-s = 1/\poly(n)$. If $c=2/3$ and $s=1/3$, we abbreviate it as $\QMA$.
\end{definition}

\begin{definition}[$\QCMA$] 
A promise problem $A = (A_\text{yes},A_\text{no})$ is in $\QCMA$[c,s] if and only if there exists a polynomial-time uniform family of quantum circuits $\{V_n\}$ and a polynomial $p$, where $V_n$ takes as input a string $x\in \{0,1\}^*$ with $|x|=n$, and a $p(n)$-qubit witness quantum state $\ket{\psi}$ and decides on acceptance or rejection of $x$ such that
\begin{itemize}
    \item if $x \in A_\text{yes} $ then there exists a $y \in \{0,1\}^{p(n)}$ such that $V_n$ accepts $(x,\ket{y})$ with probability $\geq c$,
    \item if $x \in A_\text{no}$ then for every witness state $y \in \{0,1\}^{p(n)}$, $V_n$ accepts $(x,\ket{y})$ with probability $\leq s$,
\end{itemize}
where $c-s = 1/\poly(n)$. If $c=2/3$ and $s=1/3$, we abbreviate to $\QCMA$.
\label{def:QCMA}
\end{definition}
\begin{subdefinition}[$\UQCMA$]
 The class $\UQCMA [c,s]$ has the same definition as $\QCMA$ but with the extra constraint that if $x\in A_\text{yes} $ then there exists only a single string $y*$ such that $V_n$ accepts $(x,\ket{y^*})$ with probability $\geq c(=2/3)$, and otherwise for all $y \neq y^*$ we have that $V$ accepts $(x,\ket{y})$ with probability $\leq s(=1/3)$.
\end{subdefinition}
Unlike $\QMA$, it is known that a lot of the behaviours exhibited by the classical complexity classes $\NP$ and $\MA$ hold for $\QCMA$ as well. An example of this, and one that we use later, is a result from~\cite{Aharonov2022pursuitof} stating that there exists a randomized reduction from $\QCMA$ to $\UQCMA$, anagolous to Valiant-Vazirani theorem for $\NP$~\cite{valiant1985np}.

\begin{lemma}[Randomized reduction from $\QCMA$ to $\UQCMA$~\cite{Aharonov2022pursuitof}] Let $\langle U_n, p_1,p_2 \rangle$ describe a promise problem in $\QCMA$, where $U_n$ is the description of a quantum circuit which takes an input $x$ of length $|x| = n$ and a witness $y$ with length $|y|=\poly(n)$. Denote $p_1$ and $p_2$ with $p_1-p_2 =1/\poly(n)$ for the completeness and soundness, respectively. Then there exists a randomized reduction to a $\UQCMA$ instance $\langle \tilde{U_n}, \tilde{p_1},\tilde{p_2} \rangle$, with $\tilde{p_1} - \tilde{p_2} = 1/\poly(n)$ such that:
\begin{itemize}
    \item If there exists a witness $y$ which makes $U_n$ accept $(x,y)$ with probability $\geq p_1$ then there exists a single $y^*$ which makes $\tilde{U}_n$ accept $(x,y^*)$ with probability $\geq \tilde{p_1}$ and accept $(x,y)$ for all other $y \neq y^*$ with probability $\leq \tilde{p_2}$.
    \item If $U_n$ accepts with probability $\leq p_2$ for all $y$ then $\tilde{U}_n$ accepts with probability $\leq \tilde{p_2}$ for all $y$. 
\end{itemize}
This randomized reduction succeeds with probability $\Omega(1/|y|)$.
\label{lem:UQCMA}
\end{lemma}
Another one of these properties is the equivalence of one-sided and two-sided error in the acceptance and rejectance probabilities, which just as in the $\MA$ setting holds for $\QCMA$ (assuming robustness under the choice of the universal gate-set that is used to construct the verification circuits). Formally, this is established via the following lemma.
\begin{lemma}[Perfect completeness $\QCMA$~\cite{Jordan2012}]  Let $\mathcal{G} = \{H,X,\text{Toffoli}\}$ be a fixed gate set. For any $c, s \in [0,1]$ satisfying $c-s := \delta \geq 1/q(n)$ for some polynomial $q :\mathbb{N} \rightarrow \mathbb{R}_{> 0}$, we have that
\[
\QCMA_\mathcal{G}[c,s] \subseteq \QCMA_\mathcal{G}[1,s'],
\]
where $s'= 1 - \frac{1}{2}\delta^2 - \frac{1}{2}\delta^3 $.
\label{lem:QCMA1}
\end{lemma}

\paragraph{Oracle access} For a (promise) class $\mathcal{C}$ with complete (promise) problem $A$, the class $\Pclass^\mathcal{C} = \Pclass^A$ is the class of all (promise) problems that can be decided by a polynomial-time verifier circuit $V$ with the ability to query an oracle for $A$. If $V$ makes invalid queries (i.e.~$x \in A_\text{inv}$), the oracle may respond arbitrarily. However, since $V$ is deterministic, it is required to output the same final answer regardless of how such invalid queries are answered ~\cite{gharibian2019complexity,goldreich2006promise}. Hence, the answer to any query outside of the promise set should not influence the final output bit. For a function $f$, we define $\Pclass^{\mathcal{C}[f]}$ to be just as $\Pclass^{\mathcal{C}}$ but with the additional restriction that $V$ may ask at most $f(n)$ queries on an input of length $n$.\footnote{This is different from the convention, where usually $\mO(f(n))$ is used instead.} One defines $\NP^\mathcal{C}$ or $\NP^{\mathcal{C}[f]}$ in the same way but replacing the polynomial-time deterministic verifier $V$ by a nondeterministic polynomial-time verifier $V'$, taking an additional input $y \in \{0,1\}^{p(n)}$ for some polynomial $p(n)$.

\subsection{Locality reducing perturbative gadgets}
Perturbative gadgets are standard techniques from the Hamiltonian complexity toolbox and are used to transform one Hamiltonian into another whilst approximately preserving the (low-energy) spectrum. We will use such gadgets here, and will be particularly interested in those that preserve not only the low-energy spectrum of the original Hamiltonian, but also the structure of the low-energy eigenstates. In~\cite{Cubitt2018Universal}, the authors introduce the following definition of simulation, and demonstrate via the use of perturbative gadgets that there are families of Hamiltonians which can be `reduced' to different families of Hamiltonians with simpler/lower locality interactions. Note that these results originally only applied to qubits, but can be extended to qudits~\cite{piddock2021universal}.
 
\begin{definition}[Approximate Hamiltonian simulation~\cite{Cubitt2018Universal}] We say that an $m$-qubit Hamiltonian $H'$ is a $(\Delta,\eta,\epsilon)$-simulation of an $n$-qubit Hamiltonian $H$ if there exists a local encoding $\mathcal{E}(H) = V(H \otimes P + \bar{H} \otimes Q) V^\dagger$ such that
\begin{enumerate}
 \item There exists an encoding $\tilde{\mathcal{E}}(H) = \tilde{V} (H \otimes P + \bar{H} \otimes Q) \tilde{V}^\dagger$ such that $\tilde{\mathcal{E}}(\mathds{1}) = P_{\leq \Delta (H')}$ and $\norm{V-\tilde{V}} \leq \eta$, where $P_{\leq \Delta (H')}$ is the projector onto the subspace spanned by eigenvectors of $H'$ with eigenvalue below $\Delta$,
 \item $\norm{H'_{\leq \Delta} -\tilde{\mathcal{E}}(H)} \leq \epsilon$, where $H'_{\leq \Delta}:=P_{\leq \Delta(H')}H'$.
\end{enumerate}
Here, $V$ is a local isometry that can be written as $V = \otimes_i V_i$, where each $V_i$ is an isometry acting on at most $1$ qubit, and $P$ and $Q$ are locally orthogonal projectors such that $P + Q = I$, and $\bar{M}$ is the complex conjugate of $M$. Moreover, we say that the simulation is efficient if $m$ and $\norm{H'}$ are at most $\mO\left(\poly(n,\eta^{-1},\epsilon^{-1},\Delta)\right)$, and the description of $H'$ can be computed in $\poly(n)$ time given the description of $H$.
\end{definition}
The full definition is only needed to specify the next lemma, and we would like to point readers to~\cite{Cubitt2018Universal} for a full explanation of the definition. 

For guided local Hamiltonian problems one is not just interested in the \emph{energy values}, but also in what happens to the actual \emph{eigenstates} throughout such transformations. In~\cite{cade2023improved}, Appendix B, the authors check for a large range of Hamiltonian transformations to what extent the initial eigenstates are affected. In order to obtain our results, we need the following lemma which summarizes a whole chain of reductions in~\cite{cade2023improved}.
 
\begin{lemma}[`Classical evaluability'-preserving eigenstate encodings] Suppose $H$ is an arbitrary $k$-local Hamiltonian on $n$ qubits with a non-degenerate ground state $\ket{g}$ separated from excited states by a gap $\gamma$. Then $H$ can be efficiently $(\Delta,\eta,\epsilon)$ simulated by a $2$-local Hamiltonian $H'$ on $m = \poly(n)$ qubits which has a non-degenerate ground state $\ket{g'}$, such that
\begin{align*}
    \norm{\mathcal{E}_\text{state}(\ket{g})-\ket{g'}} \leq \eta + \mO(\gamma^{-1} \epsilon),
\end{align*}
where $\mathcal{E}_\text{state}(\cdot)$ appends only states of a semi-classical form as a tensor product to $\ket{g}$, i.e.~preserves the classical evaluability as in Definition~\ref{def:cds}. 
\label{lem:2localsim}
\end{lemma}
\begin{proof} This follows immediately from the proofs of Proposition 2 and Proposition 3 in~\cite{cade2023improved}, while making the observation that all encodings up to the Spatially sparse $2$-local Hamiltonian (with Pauli interactions with no $Y$-terms) only append states that satisfy the definition of poly-sized subset states (see the proof of Theorem~\ref{thm:CGaLH} in the main text) to the original eigenstate of $H$.  
\end{proof}

\section{Guidable local Hamiltonian problems}
\label{sec:ces_and-GaLH}

\subsection{Classically evaluatable states}
\label{subsec:ces}
Let us first introduce Gharibian and Le Gall's definition of query and sampling access to quantum states~\cite{gharibian2021dequantizing}, which slightly generalizes the original definition as first proposed by Tang used to dequantize quantum algorithms for recommendation systems~\cite{Tang2019AQuantum}.

\begin{definition}[Query and sampling access, from~\cite{gharibian2021dequantizing}] We say that we have \emph{query and $\xi$-sampling access} to a vector $u \in \mathbb{C}^N$ if the following three conditions are satisfied:
\begin{enumerate}[label=(\roman*)]
    \item we have access to an $\mO\left(\poly(\log(N)\right)$-time classical algorithm $\mathcal{Q}_u$ that on input $i\in[N]$ outputs the entry $u_i$.
    \item we have access to an $\mO\left(\poly(\log(N)\right)$-time classical algorithm $\mathcal{S}\mathcal{Q}_u$ that samples from a probability distribution $p:[N]\rightarrow [0,1]$ such that
    \begin{align*}
        p(j) \in \left[ (1-\xi)\frac{|u_j|^2}{\norm{u}^2}, (1+\xi) \frac{|u_j|^2}{\norm{u}^2}\right]
    \end{align*}
    for all $j\in [N]$.
    \item we are given a real number $m$ satisfying $|m - \norm{u}| \leq \xi \norm{u}$.
\end{enumerate}
We simply say that we have sampling access to $u$ (without specifying $\xi$) if we have $0$-sampling access.
\label{def:sampaccess}
\end{definition}

In this work we propose a new class of quantum states, conceptually different from those of Definition~\ref{def:sampaccess}, which we will call \emph{classically evaluatable quantum states}. Our main motivations for doing so are the following:
\begin{enumerate}
    \item It seems rather difficult to find Ansätze that are used in practice for ground state energy estimation that satisfy all conditions of Definition~\ref{def:sampaccess}. As one of the main motivations of this work is to investigate the power of quantum versus classical state preparation when one has access to Quantum Phase Estimation, we wanted to define a class of states that can both be prepared efficiently on a quantum computer and which contains a large class of Ansätze commonly used in practice. 
    \item Analogous to Dinur's construction, one would expect that determining if a local Hamiltonian has ground state energy (exponentially close to) zero or some constant away from zero is $\QMA$-hard if the quantum PCP conjecture is true. However, there are arguments from physics\footnote{In this setting the LH problem becomes equivalent to determining whether the free energy of the system becomes negative at a finite temperature. One expects then that at such temperatures, the system loses its quantum characteristics on the large scale, making the effects of long-range entanglement become negligible. Hence, this means that the ground state of such a system should have some classical description, which places the problem in $\NP$~\cite{arad2011note}.} on why one might expect this problem to be in $\NP$~\cite{poulin2011markov}. To study the question of containment in $\NP$ it is necessary to be able to work with states within a deterministic setting, and therefore it does not make sense to rely on a form of sampling access which inherently relies on a probabilistic model of computation.
    \item To add to the previous point, being able to study containment in $\NP$ comes with the additional advantage of being able to make statements about whether the problem admits a PCP by the classical PCP theorem. No such theorem is currently known for $\MA$, and we exploit this further in Section~\ref{sec:QCPCP} where we introduce a new type of `quantum' PCP.
\end{enumerate}
We will define these quantum states in a slightly more general setting for completeness -- by allowing for probabilistic computation of expectation values as well  -- but this will not be important for the remainder of this work.

 \begin{definition}[$\epsilon$-classically evaluatable and quantumly preparable states] We say a state $u \in \mathbb{C}^{2^n}$ is \textbf{$\epsilon$-classically evaluatable} if 
\begin{enumerate}[label=(\roman*)]
    \item there exists a classical description of $u$, denoted as $\text{desc}(u)$, which requires at most $\poly(n)$ bits to write down, and
    \item  there exists a classical probabilistic algorithm $\mathcal{OQ}_u$ which, given $\text{desc}(u)$ and the matrix elements of some $k$-local observable $O$ with $\norm{O} \leq 1$, computes an estimate $\hat{z}$ such that $|\hat{z}-\bra{u}O\ket{u}|\leq \epsilon$ in time $\poly(n,\epsilon,2^k)$, with success probability $\geq 2/3$.
\end{enumerate} 
Furthermore, we say a state $u \in \mathbb{C}^{2^n}$ is also \textbf{quantumly preparable} if
\begin{enumerate}[label=(\roman*)]
  \setcounter{enumi}{2}
    \item there exists a quantum circuit $V$ of at most $\poly(n)$ $1$- and $2$-qubit gates that prepares $u$ as a quantum state, i.e.~$\ket{u}$. The description of $V$ can be computed efficiently using some efficient classical algorithm $\mathcal{A}_V$, which only takes $\text{desc}(u)$ as an input.
\end{enumerate} 
Finally, if $\epsilon = 0$ and the algorithm used in (ii) is deterministic instead of probabilistic, we simply say that $u$ is \textbf{classically evaluatable}.
\label{def:cds}
\end{definition}
Note that it is not required that $u$ is normalized, however by requirement (ii) it is possible to calculate the norm of $u$. Normalization is of course required for $u$ to be quantumly preparable. Also note that if condition (iii) holds, condition (ii) (for $\epsilon>0$) is no longer necessary in order to work with the class of states as a suitable Ansatz provided that one has access to a quantum computer, since there exist quantum algorithms to estimate the expectation values of the observables up to arbitrarily precise inverse polynomial precision. However, the current definition allows one to adopt the two-step classical-quantum procedure of \emph{classical} Ansatz generation and \emph{quantum} ground state preparation, as described in Section~\ref{sec:intro}.

\
 
To demonstrate the practical relevance of Definition~\ref{def:cds}, we give four examples of Ansätze which all satisfy the required conditions to be ($\epsilon$-)classically evaluatable and quantumly preparable. The first two examples will also be perfectly samplable, as in Definition~\ref{def:sampaccess}, of which the proofs are given in Appendix~\ref{app:samp}.
\begin{example}[Matrix-product states with bounded bond and physical dimensions] Matrix-product states are quantum states of the following form
\begin{align*}
    \ket{u} = \sum_{\{s\}} \text{Tr}[A_1^{(s_1)} A_2^{(s_2)} \dots A_n^{(s_n)}] \ket{s_1,\dots,s_{n}},
\end{align*}
where $s_i$ are qudits of `physical' dimension $p$ ($s_i \in \{0,1,\dots, p-1\}$), the $A_i^{(s_i)}$ are complex, square matrices of bond dimension $D$, and $n$ denotes the total number of qudits. We say that the bond dimension is bounded if it is at most polynomial in $n$, and that the physical dimension is bounded if it is taken to be some constant independent of $n$. MPS are also $0$-samplable, which is shown in Appendix~\ref{app:samp}.

\

\noindent \textit{Conditions check:}
\begin{enumerate}[label=(\roman*)]
    \item The MPS is fully determined by the set of matrices $\{A_i^{s_i}\}$, and can be described explicitly using at most $npD^2 = \poly(n)$ complex numbers.
    \item One can compute the inner product $\braket{u}{M|u}$ in time at most $n(2D^3\chi p + D^2\chi^2p^2)$ for any (even $n$-local) operator $M$ having a matrix product operator decomposition with bond dimension $\chi$~\cite{schollwock2011density,orus2014practical}. Since $O$ is $k$-local, it can be represented by an MPO with bond dimension at most $p^k$, and so $\braket{u}{O|u}$ can be computed in a maximum time of $n(2D^3 p^{k+1} + D^2 p^{2k+2}) = \poly(n, D, 2^k)$ when $p$ is constant.
    \item  An MPS on $n$ qubits with bond dimensions $D$ can be prepared on a quantum computer up to distance $\epsilon$ using at most $\mO(n D\log(D)^2 \log(n/\epsilon))$ $1$- and $2$-qubit gates and requiring $\lceil\log(D)\rceil$ additional ancilla qubits. A method for constructing such a circuit can be found in Appendix~\ref{app:MPS_to_Circuit} and is based on \cite{schon2005sequential}.
\end{enumerate}
\label{ex:MPS}
\end{example}

\begin{example}[Stabilizer states]
\label{exm:stabilizer_states}Gottesmann and Knill~\cite{gottesman1998heisenberg} showed that there exists a class of quantum states, containing states that exhibit large entanglement, that can be efficiently simulated on a classical computer. These states are called \emph{stabilizer states} and are those generated by circuits consisting of Clifford gates, $\mathcal C = \langle \text{CNOT}, H, S\rangle$ where $S = \sqrt{Z}$ is a phase gate, starting on a computational basis state. Any measurement of local Pauli's on these states can be efficiently classically simulated. Amongst other things, stabilizer states have been used to formulate error correcting codes~\cite{steane2003quantum}, study entanglement~\cite{Bennett_1996}, and in evaluating quantum hardware through randomised benchmarking~\cite{Knill_2008}. Stabilizer states are also $0$-samplable, again shown in Appendix~\ref{app:samp}.

\

\noindent \textit{Conditions check:}
\begin{enumerate}[label=(\roman*)]
\item Any stabilizer state can be described by a linear depth circuit consisting of Clifford gates starting on the $\ket{0^n}$ state~\cite{maslov2018shorter}. A possible description of such a circuit is a list of tuples $(q_1,q_2,t,g)$, where $q_1$ (resp. $q_2$) denotes the first (resp. second) qubit that $g \in \mathcal C$ acts on at depth $t$. This description takes at most $\tilde{\mO}(n^2)$ bits to write down.

\item The Gnottesman-Knill~\cite{gottesman1998heisenberg} theorem shows that stabilizer states allow for strong classical simulation and efficient classical computation of probabilities for Pauli measurements. This in particular allows for the calculation of expectation values in time $\poly(2^k)$. 

\item The description is given as a quantum circuit, which can be implemented to prepare the quantum state.

\end{enumerate}
\end{example}

We will now give two examples of Ans\"atze that have been shown to not be $\xi$-samplable, even up to some large constant values of $\xi$.

\begin{example}[Constant depth quantum circuits]
\label{exm:constant_depth_circuit}
Constant depth quantum circuits are circuits that, given some fixed gate set $\mathcal{G}$ with just local operations, are only allowed to apply at most $t = \mO(1)$ consecutive layers of operations from $\mathcal{G}$ on some initial quantum state, which we take to be the all-zero state $\ket{0\dots 0}$. An example of constant depth quantum circuits that are used as classical Ansätze would be the simple case of the \emph{product state Ansatz}, where one only considers one-qubit gates applied per site. Product state Ansätze are widely used in classical approximation algorithms to Local Hamiltonian problems, see for example~\cite{brandao2013product,Gharibian2019almost}. In~\cite{Terhal2002adaptive} it was shown that the ability to perform approximate weak sampling from the output of a constant depth quantum circuit up to relative error $0< \xi < 1/3$ implies that $\BQP \subseteq \AM$, which means that it is unlikely that constant depth quantum circuits are $\xi$-samplable for any $\xi < 1/3$.

\

\noindent \textit{Conditions check:}
\begin{enumerate}[label=(\roman*)]
    \item By definition.
    \item $\bra{u} O \ket{u} = \bra{0} U^\dagger O U\ket{0}$, where $U^\dagger O U$ is a $k 2^t$-local observable (via a light-cone argument), and hence we can compute $\bra{0 }U^\dagger  O U \ket{ 0}$ in time $\mO\left(2^{\mO( k 2^t)}\cdot \poly(n)\right ) $ which is $\poly(2^k)$ if $t = \mO(1)$.
    \item This holds again by definition.
\end{enumerate}
\label{ex:CDQC}
\end{example}
By combining Example~\ref{exm:stabilizer_states} and Example~\ref{exm:constant_depth_circuit} we find that any state of the form $UC\ket{0^n}$, with $U$ a constant-depth circuit and $C$ a Clifford circuit, is also classically evaluatable and quantumly preparable. Our final example is of a class of states that are not perfectly classically evaluatable, but are $\epsilon$-classically evaluatable for any $\epsilon = 1/\poly(n)$. 
\begin{example}[Instantaneous quantum polynomial (IQP) circuits] IQP circuits start in $\ket{0^n}$ and apply a polynomial number of local gates that are diagonal in the $X$-basis, followed by a computational basis measurement~\cite{bremner2011classical}. An equivalent definition would be to consider circuits with gates that are diagonal in the $Z$-basis, but then sandwiched in two layers of Hadamard gates (again followed by a measurement in the computational basis). It is well known that IQP circuits are difficult to sample from:  if IQP circuits could be weakly simulated to within multiplicative error $1\leq c < \sqrt{2}$, then the polynomial hierarchy would collapse to its third level~\cite{bremner2011classical}. Hence, they are not $\xi$-samplable for any $\xi < \sqrt{2}-1$. However, we will now show that states generated by IQP circuits are $\epsilon$-classically evaluatable for all $\epsilon = 1/\poly(n)$.

\

\noindent \textit{Conditions check:}
\begin{enumerate}[label=(\roman*)]
    \item This follows by definition, since all gates are local and there are only a polynomial number of them.
    \item This is a corollary from Theorem 3 in~\cite{bremner2011classical}, where it is shown that one can exactly sample basis states on $\mO(\log n)$ qubits according to their $l_2$-norm. Let $C$ be an IQP-circuit of $n$ qubits which produces the state $\ket{u} = C \ket{0^n}$ before the final measurement, and let $S \subseteq [n]$ with $|S| = k$ be the qubits on which a $k$-local observable $O$ acts. Following the proof of Theorem 3 in~\cite{bremner2011classical}, the state right before the last layer of Hadamard is given by
    \begin{align*}
        \ket{\phi} = \frac{1}{\sqrt{2^n}} \sum_{x\in S, y\in [n] \setminus S} e^{i f(x,y)} \ket{x,y},
    \end{align*}
    where the pair $x,y \in \{0,1\}^n$ denotes the bit string state corresponding to the concatenation (with the correct indexing) of the bit strings $x\in \{0,1\}^{|S|}$ and $y \in \{0,1\}^{n-|S|}$. Here $f(x,y)$ is a phase function which can be computed efficiently, by accumulating the relevant diagonal entries of the successive commuting gates. Since $O$ does not act on the qubits with indices $[n]\setminus S$, and they only get acted upon by Hadamards, further measurements on this register should not influence any POVM that only acts on $S$ by the no-signaling principle.  By this observation, the protocol is now very simple: one samples a random bit string $y' \in \{0,1\}^{n-|S|}$ and computes the random variable
    \begin{align*}
        X_i = \frac{1}{2^{|S|}} \sum_{x,x' \in \{0,1\}^{|S|}} \bra{x} e^{-i f(x,y')}H^{\otimes |S|} O H^{\otimes |S|} e^{i f(x,y')} \ket{x'},
    \end{align*}
    which can be done exactly in time $\poly(2^k)$ since $f(x,y')$ can be computed efficiently. Since $\norm{O} \leq 1$, We have that $\mathbb{E}[X_i^2] \leq 1$ and $|\mathbb{E}[X_i]| \leq 1$, and therefore $\text{Var}[X_i] = \mathbb{E}[X_i^2] - \mathbb{E}[X_i]^2 \leq 2$. Therefore, taking $s = c/\epsilon^2$ samples of $X_i$ (which are independent random variables) and computing $\hat{z}= \frac{1}{s} \sum_{i \in [s]} X_i$ ensures that
    \begin{align*}
        |\hat{z}-\bra{u}O\ket{u}|\leq \epsilon,
    \end{align*}
    with probability $\geq 2/3$, provided that $c \geq 6$. This follows from a simple application of Chebyshev's inequality.
    \item This follows also by definition.
\end{enumerate}
\label{ex:IQP}
\end{example}
In general quantum states will \emph{not} be classically evaluatable (as that would imply $\QMA = \NP$ as they could be used as witnesses for the $\QMA$-hard local Hamiltonian problem), and some other notable examples of classes of states which are not expected to be classically evaluatable are Projected Entangled Pair States (PEPS) (since computing expectation values of local observables is $\#\Pclass$-hard~\cite{schuch2007computational}) and collections of local reduced density matrices (to check whether they are consistent with a global quantum state is $\QMA$-hard~\cite{liu2007consistency,broadbent2022qma}).

We have seen that constant-depth quantum circuits are not even approximately samplable (under the conjecture that $\BQP \not\subset \AM$~\cite{Terhal2002adaptive}). We can formalize this in the following proposition which relates $\xi$-samplable states to $\xi$-classically evaluatable states. 
First, we need the following Lemma which is almost a direct corollary of the proof of Theorem 4.1 in~\cite{gharibian2021dequantizing}.\footnote{We cannot use their theorem directly, as it only works for even polynomials and we are interested in the polynomial $P(x)=x$.}

\begin{lemma}\label{lem:samplobsest}
Given query access to a $s$-sparse Hermitian matrix $A \in \mathbb{C}^{N \cross N}$ with $\norm{A} \leq 1$, query and $\xi$-sampling access to a vector $u \in \mathbb{C}^N$ with $\norm{u} \leq 1$ as per definition~\ref{def:sampaccess}, for any $\xi \leq \epsilon/8$ and $\epsilon \in (0,1]$ there exists a classical randomized algorithm which with probability $\geq 1-1/\poly(N)$ outputs an estimate $\hat{z} \in \mathbb{R}$ such that
\begin{align*}
    \abs{\hat{z} - u^\dagger A u} \leq \epsilon 
\end{align*}
in time $\mO^*(s/\epsilon^2)$.
\end{lemma}
\begin{proof} 
Since we have query access to the entries of $A$ and those of $u$, we can compute the $i$th entry of the vector $Au$ in time $\mO(s)$. The lemma follows then directly from the proof of Theorem 4.1 in~\cite{gharibian2021dequantizing}, taking $v=u$ and replacing their $P(\sqrt{A^\dagger A})$ with our Hermitian $A$ (this also makes the estimation of the imaginary part in the proof in~\cite{gharibian2021dequantizing} redundant). 
\end{proof}

\begin{theorem}\label{thm:sampl_vs_ces}
For any $\xi>0$, any $\xi$-samplable state is also $\mO(\xi)$-classically evaluatable. On the other hand, there exist states that are perfectly classically evaluatable but not $\xi'$-samplable for all $0< \xi'<1/3$, unless $\BQP \subseteq \AM$.
\end{theorem}
\begin{proof}
Let $u\in \mathbb C^N$ be a $\xi$-samplable state with $N=2^n$. The first part of the proposition follows by checking the two conditions. 
\begin{enumerate}[label=(\roman*)]
\item $u$ is described by giving the algorithms $\mathcal Q_u$ and $\mathcal{SQ}_u$. Both these algorithms run in $\mO(poly(log(N)))$-time, which implies that both have an efficient description of length at most  $\mO(poly(log(N)))$ (in terms of local classical operations, i.e.~logic gates). 
\item This follows directly from Lemma~\ref{lem:samplobsest}, since the global operator representation of a $k$-local observable $O$ acting on a $k$-subset of $n$ qubits can be written as $N \cross N$ Hermitian matrix $A = O \otimes I$ where $A$ has sparsity $s=2^k$.
\end{enumerate}
This shows that any $\xi$-samplable state is at least $8 \xi$-classically evaluatable. The second part follows directly from~\cite{Terhal2002adaptive}, Theorem 3, which shows that the ability to perform approximate weak sampling from the output of a constant depth quantum circuit up to relative error $0< \xi < 1/3$ implies that $\BQP \subseteq \AM$, obstructing the ability to satisfy condition (ii) in Definition~\ref{def:sampaccess}. By Example~\ref{ex:CDQC}, we already showed that constant-depth quantum circuits produce classically evaluatable states, completing the proof.
\end{proof}

This gives rise to a (conjectured) hierarchical structure of states as depicted in Figure~\ref{fig:hierarcy_of_states}.
An interesting observation is a supposedly significant leap in the hierarchy when we allow for a small error $\epsilon$ in the definition of $\epsilon$-classically evaluatable states. A straightforward way to explain this is by considering how it affects our ability to determine a global property of a quantum state, like its energy with respect to a Hamiltonian $H$. 

Let $H$ be a sum of $m$ log-local terms, i.e.~$H = \sum_{i=0}^{m-1} H_i$, satisfying $\norm{H} \leq 1$. If one wants to evaluate the energy of an $\epsilon$-classically evaluatable state with respect to $H$ up to accuracy $\epsilon'$, then $\epsilon$ has to be less than $\epsilon'/m$ since in the worst case the error grows linearly with the number of terms. Instead, $\xi$-samplable states have a requirement on the accuracy of sampling, which is a property of the global state. \cite{gharibian2021dequantizing} shows that this property can be used for energy estimation, where the requirement on $\xi$ only depends on the precision with which one wants to measure the energy. We see this reflected in Theorem~\ref{thm:sampl_vs_ces}, which shows that if a state has the property of being $\xi$-samplable this implies that the state is $\mO(\xi)$-classically evaluatable, but not the other way around. However, we are not aware of any classes of states which are provably only $\xi$-samplable for a constant, but small, $\xi >0$ (all examples that we give in this work are in fact $0$-samplable).

For the remainder of our work, we will focus on ($0$-)classically evaluatable states, which by Definition~\ref{def:cds} means that $\mathcal{OQ}_u$ is deterministic. A notable advantage of this approach, as opposed to $0$-samplable states, lies in its compatibility with deterministic algorithms, allowing us to give $\NP$ containment results (see Section~\ref{sec:clas_cont}). This is also a prerequisite to make connections to $\MA$ as well, see Appendix~\ref{app:MA}.

\begin{figure}
    \centering
    \includegraphics[width=\linewidth]{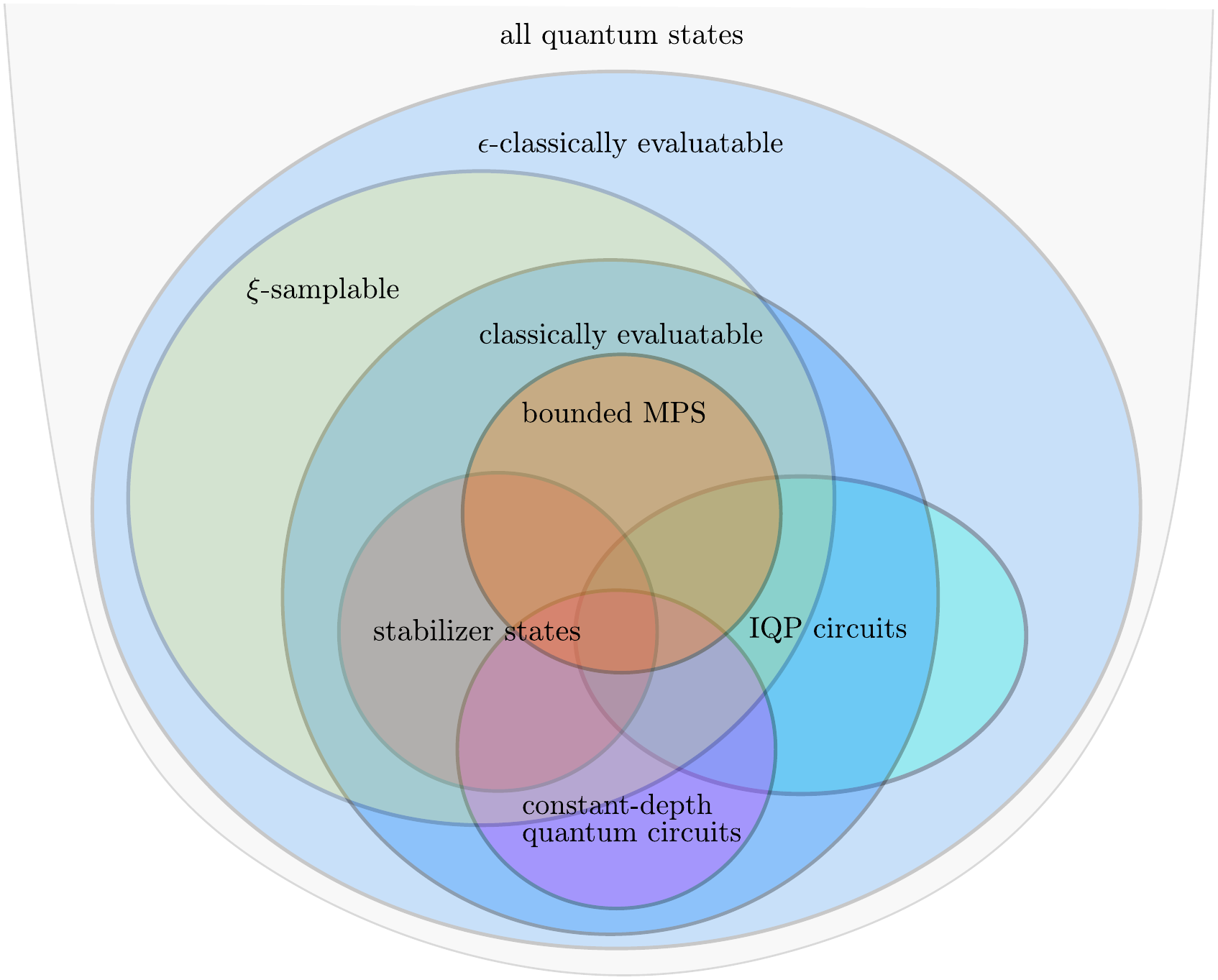}
    \caption{Visualization of the (conjectured) relations between classes of quantum states considered in this work, given a Hilbert space of a fixed dimension. For MPS, we only consider states with polynomially-bounded bond and local dimension. We take $\xi \leq \epsilon/8\leq 1/3$, such that by Theorem~\ref{thm:sampl_vs_ces} we have that (i) all $\xi$-samplable states are also $\epsilon$-classically evaluatable and (ii) constant-depth and IQP circuits are not $\xi$-samplable. One also expects that there are quantum states (which can be prepared by a polynomial time quantum circuit) which are neither classically evaluatable nor samplable, or else $\QMA$ ($\QCMA$) would be in $\NP$ or $\MA$, respectively.}
\label{fig:hierarcy_of_states}
\end{figure}

\subsection{Variants of guidable local Hamiltonian problems}
\label{subsec:GaLH_def}
Let us define the following class of local Hamiltonian problems, which can be viewed as `Merlinized' versions of the original guided local Hamiltonian problem. We make a distinction between different types of promises one can make with respect to the existence of guiding states: we either assume that the guiding states are of the form of Definition~\ref{def:cds} (with or without the promise that the states are also quantumly preparable), or that there exists an efficient quantum circuit that prepares the guiding state.

\begin{definition}[Guidable Local Hamiltonian Problems] Guidable Local Hamiltonian Problems are problems defined by having the following input, promise, one of the extra promises and output:\\
\textbf{Input:} A $k$-local Hamiltonian $H$ with $\|H\|\leq 1$ acting on $n$ qubits, threshold parameters $a,b \in \mathbb{R}$ such that $b-a \geq \delta > 0$ and a fidelity parameter $\zeta \in (0,1]$.\\
\textbf{Promise:} We have that either $\lambda_0(H) \leq a$ or $\lambda_0(H) \geq b$ holds, where $\lambda_0(H)$ denotes the ground state energy of $H$. \\
\textbf{Extra promises:} 
Denote  $\Pi_\text{gs}$ for the projection on the subspace spanned by the ground state of $H$.  Then for each problem class, we have that either one of the following promises hold:
\begin{enumerate}
    \item There exists a classically evaluatable state $u\in \mathbb{C}^{2^n}$ for which $\norm{\Pi_\text{gs} u}^2 \geq \zeta$. Then the problem is called the \textbf{Classically Guidable Local Hamiltonian Problem}, shortened as $\mathsf{CGaLH}(k,\delta,\zeta)$. If $u$ is also quantumly preparable, we call the problem the \textbf{Classically Guidable and Quantumly Preparable Local Hamiltonian Problem}, shortened as $\mathsf{CGaLH}^{*}(k,\delta,\zeta)$. 
    \item There exists a unitary $V$ implemented by a quantum circuit composed of at most $T=\poly(n)$ gates from a fixed gate set $\mathcal{G}$ that produces the state $\ket{\phi}=V\ket{0}$ (with high probability),  which has $\norm{\Pi_\text{gs} \ket{\phi}}^2 \geq \zeta$. Then the problem is called the \textbf{Quantumly Guidable Local Hamiltonian problem}, shortened as $\mathsf{QGaLH}(k,\delta,\zeta)$.
\end{enumerate}
\textbf{Output:} \begin{itemize}
    \item If $\lambda_0(H) \leq a$, output {\sc yes}.
    \item If $\lambda_0(H) \geq b$, output {\sc no}.
\end{itemize}
\label{def:GaLH}
\end{definition}
\begin{remark}
    If one removes the extra promise from the Guidable Local Hamiltonian Problem -- i.e.~the existence of a guiding state -- one recovers the usual definition of the Local Hamiltonian Problem, albeit with the normalization condition on the overall operator norm instead of the operator norm of the local terms. In the remainder of the paper when we refer to the Local Hamiltonian Problem, we mean the Local Hamiltonian Problem as defined in this way. Therefore, in any discussion related to the Local Hamiltonian formulation of the quantum PCP conjecture we consider the local Hamiltonian problem with a promise gap constant relative to the operator norm (which satisfies $\leq 1$). 
\end{remark}
One can also consider other types of guiding states, for example the samplable states as in Definition~\ref{def:sampaccess}. This guidable local Hamiltonian problem variant was already introduced in Section 5 of~\cite{gharibian2021dequantizing}.

The $\mathsf{QGaLH} (k,\delta,\zeta)$ problem is very similar to the low complexity low energy states problem from \cite{Wocjan2003two}, but differs in some key ways. In the low complexity low energy states problem one is promised that for all states $\{ \ket{\phi} \}$ that can be prepared from $\ket{0\dots 0}$ with a polynomially bounded number of gates from a fixed gate set, one has that either there exists at least one such $\ket{\phi}$ such that $\bra{\phi}H\ket{\phi}\leq a$ or for all these $\ket{\phi}$ we have $\bra{\phi}H\ket{\phi}\geq b$. Instead, in $\mathsf{QGaLH}(k,\delta,\zeta)$ one is promised that there exists a state $\ket{\psi}$ which can be prepared efficiently on a quantum computer that has fidelity $\zeta$ with the ground space of $H$. This promise in the fidelity does not imply that the energy of this $\ket{\psi}$ is necessarily low, as it might have a large fidelity with states in the high-energy spectrum of $H$. Nevertheless, it does imply that in the {\sc yes}-case there exists a low complexity low energy state $\ket{\phi}$. One can make use of the state $\ket{\psi}$ that has significant overlap with the ground state and use Lin and Tong's filtering method~\cite{Lin2020nearoptimalground} to project $\ket{\psi}$ onto a state $\ket{\phi}$ with energy at least inverse polynomially close to the ground state (which implies $\ket{\phi}$ can be prepared by a quantum circuit). However,  in the {\sc no}-case this promise on the fidelity implies that every possible state $\ket{\psi}$ has energy $\bra{\psi}H\ket{\psi}\geq b - \mO(1/\exp(n))$, as even in the {\sc no}-case it is still possible to approximate the ground state energy up to polynomial precision. This is different from the {\sc no}-case of the low complexity low energy states problem, where there might exist states with energy lower than $a$, as long as these states are not preparable by a polynomial-time quantum circuit, making the $\mathsf{QGaLH}(k,\delta,\zeta)$ problem more restrictive than the low complexity low energy states problem.  In principle, this could be remedied by relaxing the requirement in $\mathsf{QGaLH}(k,\delta,\zeta)$ from having fidelity with the ground space to having fidelity with the space of states with sufficiently low energy in the {\sc yes}-case only. All our results that follow would still hold, and this new problem could then be seen as a generalisation of the low complexity low energy states problem.

\

\noindent In the upcoming section we will characterize the complexity of these guidable local Hamiltonian problems in various parameter regimes.

\section{$\QCMA$-completeness of guidable local Hamiltonian problems}
\label{sec:QCMA-hardness}
In this section we prove that Guidable Local Hamiltonian problems are $\QCMA$-hard in the inverse polynomial precision regime. Our construction is based on a combination of the ideas needed to show $\BQP$-hardness for the Guided Local Hamiltonian problem~\cite{gharibian2021dequantizing,Cade2022Complexity,cade2023improved} and the small penalty clock construction of~\cite{Deshpande2020}.\\

\noindent The first obstruction one encounters in adopting the ideas from the $\BQP$-hardness proofs of the Guided Local Hamiltonian problem to the guidable setting is the fact that $\QCMA$ verifiers, unlike $\BQP$, have a proof register. In $\QCMA$ the promises of completeness and soundness are always with respect to computational basis state witnesses. Hence, these might no longer hold when \emph{any} quantum state can be considered as witness: for example, in the {\sc no}-case there might be highly entangled states which are accepted with probability $\geq 2/3$. When considering a circuit problem, the verifier can easily work around this by simply measuring the witness and then proceeding to verify with the resulting computational basis state. However, there is also another trick, which retains the unitarity of the verification circuit -- and which we will denote as the `CNOT-trick' from now on -- to force the witness to be classical, first used in proving $\QCMA$-completeness of the \emph{Low complexity low energy states} problem in~\cite{Wocjan2003two}. Since the authors do not explain the precise mechanism behind the workings of this CNOT-trick, we provide a short proof of the lemma below.
\begin{lemma}[The `$\text{CNOT}$-trick']
Let $p(n):\mathbb{N} \rightarrow \mathbb{R}_{> 0}, q(n):\mathbb{N} \rightarrow \mathbb{R}_{> 0}$ be polynomials. Let $U_n$ be a quantum polynomial-time verifier circuit that acts on an $n$-qubit input register $A$, a $p(n)$-qubit witness register $B$ and a $q(n)$-qubit workspace register $C$, initialized to $\ket{0}^{\otimes q(n)}$. Denote $\Pi_0$ for the projection on the first qubit being zero. Let $Q$ be the Marriott-Watrous operator of the circuit, defined as
\begin{align}
    Q = \left(\bra{x} \otimes I_{w} \otimes \bra{0}^{\otimes q(n)} \right) U_n^\dagger \Pi_0 U_n \left(\ket{x} \otimes I_{w} \otimes \ket{0}^{\otimes q(n)} \right).
    \label{eq:MWop}
\end{align}
Consider yet another additional $p(n)$-qubit workspace $D$ initialized to $\ket{0}^{\otimes p(n)}$, on which $U_n$ does not act. Then by prepending $U_n$ with $p(n)$ CNOT-operations, each of which is controlled by a single qubit in register $B$ and targeting the corresponding qubit in register $D$, the corresponding Marriott-Watrous operator becomes diagonal in the computational basis.
\label{lem:CNOT_trick}
\end{lemma}
\begin{proof} Denote $U_\text{{CNOT}}$ for the $2 p(n)$ qubit operation that acts on the two registers $B$ and $D$, and that for each $l \in [p(n)]$ applies a CNOT controlled by qubit $l$ in register $B$ and targets qubit $l$ in register $D$. Consider the new verifier circuit $\tilde{U}_n = U_n U_\text{{CNOT}}$ that acts on the registers $A,B,C$ and $D$, with the corresponding Mariott-Watrous operator $\tilde{Q}$. Let $\ket{i}$ and $\ket{j}$ for $i,j \in [2^{p(n)}]$ be arbitrary computational basis states.  Then we have 
\begin{align*}
    \bra{i} \tilde{Q} \ket{j}  &=  \left(\bra{x} \otimes \bra{i} \otimes \bra{0}^{\otimes q(n)} \otimes \bra{0}^{\otimes p(n)}\right) U_\text{{CNOT}} U_n^\dagger \Pi_0 U_n U_\text{{CNOT}} \left(\ket{x} \otimes \ket{j} \otimes \ket{0}^{\otimes q(n)} \otimes \ket{0}^{\otimes p(n)}\right)\\
    &=  \left(\bra{x} \otimes \bra{i} \otimes \bra{0}^{\otimes q(n)} \otimes \bra{i} \right) U_n^\dagger \Pi_0 U_n \left(\ket{x} \otimes \ket{j} \otimes \ket{0}^{\otimes q(n)} \otimes \ket{j}\right)\\
    &=  \bra{i} \ket{j}  \left(\bra{x} \otimes \bra{i} \otimes \bra{0}^{\otimes q(n)} \right)  U_n^\dagger \Pi_0 U_n  \left(\ket{x} \otimes \ket{j} \otimes \ket{0}^{\otimes q(n)} \right)\\
    &= \delta_{i,j}\bra{i} Q \ket{j},
\end{align*}
where we used the fact that $V$ and $\Pi_0$ themselves do not act on register $D$. Hence, the operator $\tilde{Q}$ is diagonal in the computational basis, where its entries are taken from the diagonal of $Q$.
\end{proof}

The next obstruction one faces is that in the $\QCMA$ setting there might be multiple proofs which all have exponentially close, or even identical, acceptance probabilities. The analysis of the $\BQP$-hardness proof fails to translate directly to this setting, and another technique is needed. For this, we resort to the \textit{small-penalty clock construction} of~\cite{Deshpande2020}. The key idea is to use a Feynman-Kiteav circuit-to-Hamiltonian mapping modified with a tunable parameter $\epsilon$, which maps a quantum verification circuit $U_n$, consisting of $T$ gates from a universal gate set of at most $2$-local gates, taking input $x$ and a quantum proof $\ket{\psi} \in \left( \mathbb{C}^2\right)^{\otimes \poly(n)}$ to a $k$-local Hamiltonian of the form
\begin{align}
H^x_{FK} = H_\text{in} + H_\text{clock} + H_\text{prop} + \epsilon H_\text{out}.
\label{eq:H_FK}
\end{align}
The value of $k$ depends on the used construction. Intuitively, the first three terms check that the Hamiltonian is faithful to the computation and the last term shifts the energy level depending on the acceptance probability of the circuit. Just as in~\cite{Deshpande2020}, we will use Kempe and Regev's $3$-local construction. A precise description of the individual terms in~\eqref{eq:H_FK} can be found in~\cite{kempe20033local}, and will not be relevant for our work, except for the fact that the $H^x_\text{FK}$ has a polynomially bounded operator norm.\footnote{For the results in Appendix~\ref{app:GaQSAT} it is important that all terms are weighted projectors, which is easily verified.} The ground state of the first three terms $H_0 = H_\text{in} + H_\text{clock} + H_\text{prop}$ is given by the so-called \textit{history state}, which is given in~\cite{kempe20033local} by
\begin{align}
    \ket{\eta(\psi)} = \frac{1}{\sqrt{T+1}} \sum_{t = 0}^T U_t \dots U_1 \ket{\psi} \ket{0} \ket{\hat{t}},
    \label{eq:hist_state}
\end{align}
where $\ket{\psi}$ is the quantum proof and $ \hat{t}$ the unary representation of the time step of the computation given by
\begin{align*}
    \hat{t} = | \underbrace{1 \dots 1}_t \underbrace{0 \dots 0}_{T-t} \rangle.
\end{align*}
From the construction in~\cite{kempe20033local}, it is easily verified that if $U_n$ accepts ($x,\ket{\psi}$) with probability $p$ then we have that the corresponding history state has energy
\begin{align}
\bra{\eta(\psi)} H^x_{FK} \ket{\eta(\psi)} = \epsilon \frac{1-p}{T+1}.
\label{eq:energy_hs}
\end{align}
Though the core idea behind the small-penalty clock construction is identical to the one used in the $\BQP$-hardness proof -- rescaling the weight of the $H_\text{out}$ term as compared to the other terms in a Feynman-Kiteav circuit-to-Hamiltonian mapping -- the analysis differs: using tools from the Schrieffer-Wolff transformation one can find precise bounds on intervals in which the energies in the low-energy sector must lie, gaining fine control over the relation between the acceptance probabilities of the circuit and the low-energy sector of the Hamiltonian. The main lemma we use from~\cite{Deshpande2020} is adopted from the proof of Lemma 26 in their work. 

\begin{lemma}[Small-penalty clock construction, adopted from Lemma 26 in~\cite{Deshpande2020}] Let $U_n$ be a quantum verification circuit for inputs $x$, $|x|=n$, where $U_n$ consists of $T = \poly(n)$ gates from some universal gate-set using at most $2$-local gates. Denote $P(\psi)$ for the probability that $U_n$ accepts $(x,\ket{\psi})$, and let $H^x_\text{FK}$ be the corresponding $3$-local Hamiltonian from the circuit-to-Hamiltonian mapping in~\cite{kempe20033local} with a $\epsilon$-factor in front of $H_\text{out}$, as in Eq.~\eqref{eq:H_FK}. Then for all $\epsilon \leq c/T^3$ for some constant $c >0$, we have that low-energy subspace $\mathcal{S}_\epsilon$ of $H$, i.e.
\begin{align*}
    S_\epsilon = \text{span} \{ \ket{\Phi} : \bra{\Phi} H \ket{\Phi} \leq \epsilon \}
\end{align*}
has that its eigenvalues $\lambda_i$ satisfy
\begin{align}
    \lambda_i \in \left[\epsilon \frac{1-P(\psi_i)}{T+1} - \mO(T^3\epsilon^2), \epsilon \frac{1-P(\psi_i)}{T+1} + \mO(T^3\epsilon^2)\right],
    \label{eq:spclock}
\end{align}
where $\{ \ket{\psi_i}\}$ are the eigenstates of the Mariott-Watrous operator of the circuit $U_n$ given by Eq.~\eqref{eq:MWop}.
\label{lem:spcc}
\end{lemma}
Having a $\QCMA$-verifier with the CNOT-trick of Lemma~\ref{lem:CNOT_trick} ensures that in Lemma~\ref{lem:spcc} all $\ket{\psi_i}$ are computational basis states, as the CNOT-trick diagonalizes the Mariott-Watrous operator. The key idea is now to exploit the fact that $\QCMA$, unlike for what is known for $\QMA$, is a `well-behaved' class in the sense that is equal to $\UQCMA$ (under randomized reductions) and has perfect completeness (see Lemmas~\ref{lem:UQCMA} and Lemma~\ref{lem:QCMA1}). Though it is not clear whether both properties can be satisfied at the same time\footnote{An earlier version of this work falsely claimed this to be case. Though the construction ensured that there was a unique witness that was accepted with probability $1$ in the {\sc yes}-case, the gap to the other witnesses was not constant but exponential.}, it turns out that both can be used to obtain separate results (see Appendix~\ref{app:GaQSAT}).

\begin{figure}
    \centering
    \includegraphics{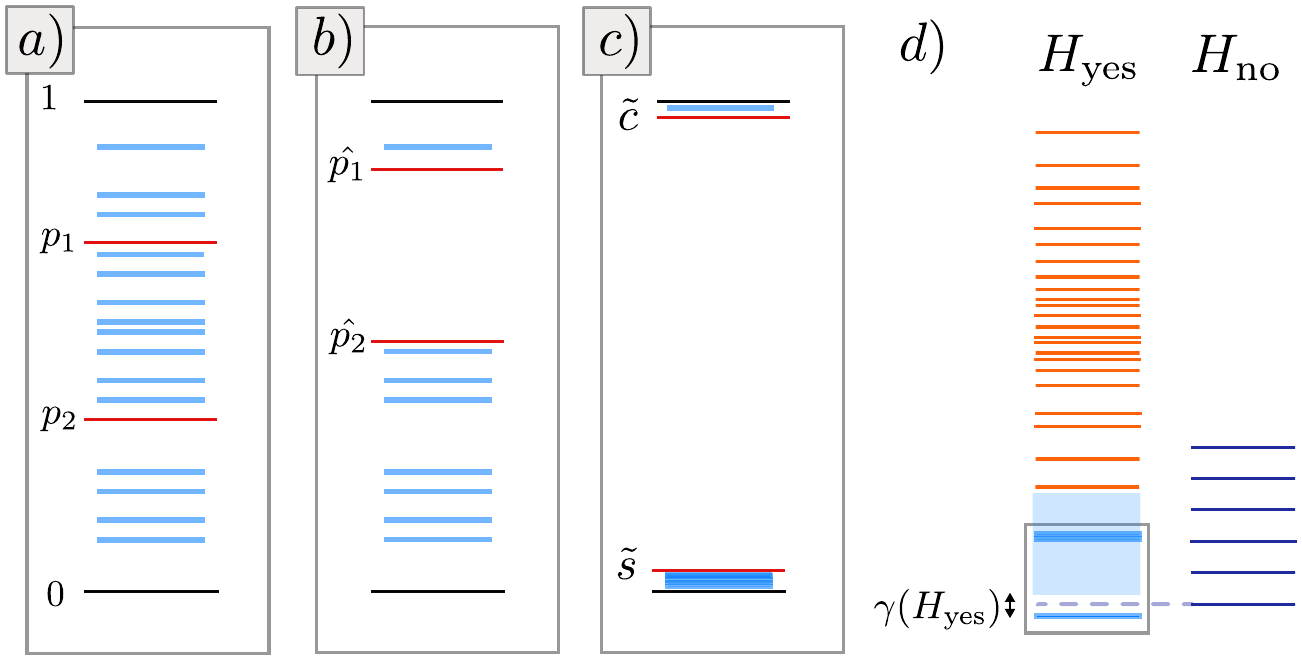}
    \caption{ Illustration of the key ideas to construct the desired witness distribution in the {\sc yes}-case in the first part of the reduction. The blue lines are witnesses, for which their position with respect to the $y$-axis represents the corresponding acceptance probabilities. The dark red lines represent the completeness and soundness parameters. $a) \rightarrow b)$ represents the randomized reduction from a $\QCMA$-problem to a $\UQCMA$ one, $b) \rightarrow c)$ the error reduction and finally $d) \rightarrow e)$ 
    The spectra of $H_\text{yes}$ and $H_\text{no}$ when $x\in A_\text{yes}$. $H_\text{yes}$ follows from the circuit-to-Hamiltonian mapping with the small penalty resulting in a Hamiltonian with fine control over its low-energy subspace, allowing one to ensure that its ground state is unique and can be made exponentially close to the history state corresponding to the unique accepting witness. The light blue shaded area represents the fact that we do not know the exact energy values corresponding to non-accepting witnesses, except for the fact that they are separated from $\lambda_0(H_\text{yes})$ by at least $\gamma(H_\text{yes}) = \Omega(1/\tilde{T}^6)$ for our choice of $\epsilon$. $H_\text{no}$ is chosen such that its ground state energy lies exactly in the gap of $H_\text{yes}$ in the {sc Yes}-case. Observe that if one was able to show that $\QMA \subseteq_r \UQMA $, one could use the same proof construction to show $\QMA$-hardness of inverse-poly-gapped Hamiltonians, for which we only yet know that they are $\QCMA$-hard. $\QCMA$-hardness for inverse-poly-gapped Hamiltonians was already shown in~\cite{Aharonov2022pursuitof} (in fact they even show it for 1D Hamiltonians), and rediscovered in this work.}
    \label{fig:reduction}
\end{figure}

\begin{theorem} $\mathsf{CGaLH}(k,\delta,\zeta)$ is $\mathsf{QCMA}$-hard under randomized reductions for $k\geq 2$, $\zeta \in (1/\poly(n),1-1/\poly(n))$ and $\delta = 1/\poly(n)$.
\label{thm:CGaLH}
\end{theorem}
\begin{proof} 
Let us first state a `basic version' reduction, which uses basis states as guiding states which trivially satisfy the conditions of Definition~\ref{def:cds}, for which we prove completeness and soundness, and finally improve its parameters in terms of the achievable fidelity and locality domains. 

\paragraph{The reduction} Let $\langle U_n, p_1,p_2 \rangle$ be a $\QCMA$ promise problem. By using Lemma~\ref{lem:UQCMA}, there exists randomized reduction to a $\UQCMA$ promise problem $\langle \hat{U_n}, \hat{p}_1,\hat{p}_2 \rangle$, $\hat{p}_1-\hat{p}_2 \geq 1/q(n)$ for some polynomial $q$, which uses witnesses $y \in \{0,1\}^{p(n)}$ for some polynomial $p(n)$ and uses at most $T = \poly(n)$ gates. We will now apply the following modifications to the $\UQCMA$ instance:
\begin{enumerate}
\item First, we force the witness to be classical by adding another register to which we `copy' all bits of $y$ (through CNOT operations), before running the actual verification protocol -- i.e.~we use the CNOT trick of Lemma~\ref{lem:CNOT_trick}, which diagonalizes the corresponding Marriot-Watrous operator in the computational basis.
\item We apply \emph{error reduction} to the circuit.  This is done by applying the so-called ``Marriot and Watrous trick" for error reduction, described in~\cite{Mariott2005quantum}, which allows one to repeat the verification circuit several times whilst re-using the same witness. It is shown in~\cite{Mariott2005quantum}, Theorem 3.3, that for any quantum circuit $V_n$ using $T = \poly(n)$ $2$-qubit gates which decides on acceptance or rejectance of an input $x$, $|x|=n$, using a $p(n)$-qubit witness $\ket{\psi}$ for some polynomial $p$, satisfying completeness and soundness probabilities $c$, $s$ such that $c-s \geq 1/q(n)$ there is another circuit $\tilde{V}_n$ that again uses a $p(n)$-qubit witness $\ket{\psi}$ but has completeness and soundness $1-2^{-r}$ and $2^{-r}$, respectively, at the cost of using $\tilde{T} = \mO( q^2 r T)$ gates.
\end{enumerate}
Let the resulting protocol be denoted by 
$\langle \tilde{U}_n, \tilde{c},\tilde{s} \rangle$, where $\hat{U}_n$ has an input register $A$, a witness register $W$ and ancilla register $B$, uses $\tilde{T} = \mO(q^2 r T)$ gates and has completeness and soundness $C = 1-2^{-r}$ and $\hat{s} = 2^{-r}$. We denote $y^*$ for the (unique) witness with acceptance probability $\geq C$ in the {\sc yes}-case. We keep $r$ as a parameter to be tuned later in our construction. We will also write $P(y):=\text{Pr}[\hat{U}\text{ accepts } (y)]$. Now consider the $4$-local Hamiltonian
\begin{align}
     H^x = H_{\sc yes} \otimes \ket{0} \bra{0}_D +   H_{\sc no} \otimes \ket{1}\bra{1}_D,
     \label{eq:full_H}
\end{align}
where $H_{\sc yes} = H^x_\text{FK}$ is the Hamiltonian given by Eq.~\eqref{eq:H_FK} using the circuit $\hat{U}_n$ and parameter $\epsilon$ and $H_{\sc no}$ is given by 
\begin{align}
    H_{\sc no} =  \sum_{i=0}^{R-1}  \kb{1}{1}_i + b I, 
    \label{eq:H_no}
\end{align}
where $R$ is the total size of the registers $A$, $W$, $B$ and the clock register $C$, and $b >0$ is yet another tunable parameter. Note that $H_{\sc no}$ has a \emph{unique} ground state with energy $b$ given by the all zeros state, and the spectrum after that increases in steps of $1$ (and so it in particular has a \emph{spectral gap} of $1$). We also have that  $\norm{H_{\sc no}} = R +b = \poly(n)$.
As a guiding state in the {\sc yes}-case will use the following basis state 
\begin{align}
    \ket{u_\text{yes}} =  \ket{x}_{A} \ket{y^*}_{W} \ket{0\dots0}_{B} \ket{0}_C \ket{0}_D,
    \label{eq:guidstate_yes}
\end{align}
which satisfies $\left(\bra{\eta(y^{ *})}\bra{0}_D\right)\ket{u_\text{yes}} = 1/\sqrt{(T+1)} = \mO(1/\poly(N))$, with $\ket{\eta(y^{ *})}$ being the history state of witness $y^{ *}$ for Hamiltonian $H_{yes}$. In the {\sc no}-case, we will show that the state
\begin{align}
    \ket{u_\text{no}} =  \ket{0 \dots 0}_{A W B C} \ket{1}_D,
    \label{eq:guidstate_no}
\end{align}
will be in fact the ground state. We will now show that setting $b := \mO(1/\tilde{T}^7)$ and $\epsilon := \mO(1/\tilde{T}^5)$, our reduction achieves the desired result.

\paragraph{Completeness}
Let us first analyse the {\sc yes}-case. By Lemma~\ref{lem:spcc}, we have that the eigenvalue $\lambda(y)$ corresponding to the witness $y^*$ is upper bounded by
\begin{align*}
    \lambda (y^*) \leq \epsilon \frac{2^{-r}}{\tilde{T}+1} + \mO( \tilde{T}^3\epsilon^2).
\end{align*}
On the other hand, we have that for any $y \neq y^*$ 
\begin{align*}
    \lambda (y) \geq 
       \epsilon \frac{1-2^{-r}}{\tilde{T}+1} - \mO(\tilde{T}^3\epsilon^2) = \Omega\left(\frac{1}{\tilde{T}^6}\right)
\end{align*}
for our choice of $\epsilon$ and $r \geq 1$. Hence, for our choice of $\epsilon$ we must have that the ground state $\ket{\psi}$ of $H_\text{yes}$ is unique and has a spectral gap that can be bounded as
\begin{align}
    \gamma(H_{\sc yes}) \geq  \epsilon\frac{1-2^{-r+1}}{\tilde{T}+1} - \mO \left( \tilde{T}^3\epsilon^2 \right)= \Omega \left(\frac{1}{ \tilde{T}^6} \right),
    \label{eq:spectral_gap}
\end{align}
for some $r \geq \Omega(1)$ (we will pick $r$ to be much larger later). Let us consider the fidelity of the history state $\ket{\eta(y^*)}$ with the actual ground state. First, we have that the energy of $\ket{\eta(y^*)}$ is upper bounded by 
\begin{align*}
  \bra{\eta(y^*)}  H_{\sc yes }\ket{\eta(y^*)} \leq  \epsilon \frac{2^{-r}}{\tilde{T}+1} = \mO\left( \frac{2^{-r}}{ \tilde{T}^6 } \right),
\end{align*}
which follows directly from Eq.~\ref{eq:energy_hs} and the fact that $P(y^*) \geq 1-2^{-r}$.  We can write $\ket{\eta(y^*)}$ in the eigenbasis of $H_{\sc yes}$ as
\begin{align*}
    \ket{\eta(y^*)} = \alpha \ket{\psi} +\sqrt{1-\alpha^2} \ket{\psi^\perp},
\end{align*}
for some real number $\alpha \in [0,1]$, where $\ket{\psi}$ is the actual ground state of $H_{\sc yes}$ and $\ket{\psi^\perp}$ another state orthogonal to $\ket{\psi}$. We have that the energy of $\ket{\eta(y^*)}$ is upper bounded by 
\begin{align*}
  \bra{\eta(y^*)}  H_{\sc yes }\ket{\eta(y^*)} \leq  \epsilon \frac{2^{-r}}{\tilde{T}+1} = \mO\left( \frac{2^{-r}}{ \tilde{T}^6 } \right).
\end{align*}
On the other hand, the energy of $\ket{\eta(y^*)}$ is lower bounded by 
\begin{align*}
    \bra{\eta(y^*)}  H_{\sc yes }\ket{\eta(y^*)}=  \alpha^2 \bra{\psi} H_{\sc yes} \ket{\psi} + (1-\alpha^2) \langle{\psi^\perp}| H_{\sc yes} |\psi^\perp\rangle \geq \Omega \left(\frac{1-\alpha^2 }{ \tilde{T}^6} \right),
\end{align*}
using the fact that $H_{\sc yes}$ is PSD. Combining the upper and lower bounds, we find
\begin{align}
    \alpha^2 = |\bra{\eta(y^*)} \ket{\psi}|^2 \geq  1- \mO\left(2^{-r}\right),
    \label{eq:overlap_yescase}
\end{align}
which can be made $\geq 1-2^{-c \tilde{T}}$ for some $r = c \tilde{T} + \mO(1)$. Hence, we have that the fidelity of $\ket{u_\text{yes}}$ with the unique ground state of $H$ can be lower bounded as
\begin{align*}
    \abs{\bra{u_\text{yes}}\ket{\psi}}^2  &\geq  1 - \left(\sqrt{1 - |\bra{u_\text{yes}}(\ket{\eta(y^*)}\ket{0})|^2} + \sqrt{1 - |(\bra{\eta(y^*)}\bra{0})\ket{\psi}|^2}\right)^2\\
    &\geq 1- \left( \sqrt{1-\frac{1}{\tilde{T}+1}} + 2^{-c \tilde{T}/2}  \right)^2 \\
    &\geq \Omega \left(\frac{1}{\tilde{T}}\right),
\end{align*}
as desired.

\paragraph{Soundness}  We have that all witnesses $y$ get accepted by $\hat{U}$ with at most an exponentially small probability, and hence have that $H_{\sc yes} \succeq \Omega(1/\tilde{T}^6)$. By our choice $b$ we have therefore ensured that the ground state in the {\sc no}-case must be the state given by Eq.~\eqref{eq:guidstate_no}, which has energy $b = \Omega(1/\tilde{T}^7)$.  Hence, the promise gap between {\sc yes} and {\sc no} cases is $\delta = \Omega(1/\tilde{T}^7) = \Omega(1/ q^2 T^8) = 1/\poly(n)$.\\

\noindent We will now use similar tricks as in~\cite{cade2023improved} to improve the basic construction in terms of the fidelity range and locality.

\paragraph{Increasing the fidelity range}
Note that in the {\sc no}-case we already have that the ground state is a semi-classical poly-sized subset state. However, in the {\sc yes}-case, the ground state is a history state with only inverse polynomial fidelity with the state $\ket{u_\text{yes}}$. To work around this, we apply the same trick as in~\cite{cade2023improved}:\footnote{This trick is not original to~\cite{cade2023improved} and is well known, see e.g.~\cite{caha2017feynman}.} by pre-idling the circuit with a polynomial number of identities, of which we denote the total number by $N$, and guiding state to
\begin{align}
     \ket{u_\text{yes}^{\text{new}}} =  \frac{1}{\sqrt{N}}\sum_{t = 0}^{N-1}\ket{x}_{A} \ket{y^*}_{W} \ket{0\dots0}_{B} \ket{t}_C \ket{0}_D,
     \label{eq:new_guidstate_yes}
\end{align}
which satisfies
\begin{align*}
     | \bra{u_\text{yes}^{\text{new}}}(\ket{\eta(y^*)}\ket{0})|^2 &= \frac{N}{N+\tilde{T}+1}
\end{align*}
Since the history state itself has an exponentially close fidelity with the ground state by equation Eq.~\eqref{eq:overlap_yescase}, we have that the guiding state itself has an inverse polynomially close to $1$ fidelity with the unique ground state $\ket{\psi}$ of large enough $N$. For the new pre-idled circuit we have to replace in all our results throughout our construction $\tilde{T}$ by $\tilde{T} + N$ and we have that the fidelity becomes
\begin{align*}
    \abs{\bra{u^{\text{new}}_\text{yes}}\ket{\psi}}^2  &\geq  1 - \left(\sqrt{1 - |\bra{u^{\text{new}}_\text{yes}}(\ket{\eta(y^*)}\ket{0})|^2} + \sqrt{1 - |(\bra{\eta(y^*)}\bra{0})\ket{\psi}|^2}\right)^2\\
    &\geq 1- \left( 1-\frac{N}{N+\tilde{T}+1} + 2^{-c (\tilde{T}+N)/2}  \right)^2 \\
    &\geq 1 - \frac{1}{r(n)},
\end{align*}
for any positive polynomial $r$ for some choice of $N \in \poly(\tilde{T})$. 

\paragraph{Classical evaluatability and quantum preparability}
 We will check each condition of Definition~\ref{def:cds} for $\ket{u_\text{yes}^{\text{new}}}$. Condition (i) follows directly from the definition of polynomially-sized subset states. For condition (ii) we have that $\bra{u} O \ket{u} = \frac{1}{|S|} \sum_{i,j \in S} \bra{i}O \ket{j}$ can be computed efficiently for any $O$ for which we have query access to its matrix elements, since each $\bra{i}O \ket{j}$ corresponds to a query to the element $O_{i,j}$. Hence, when $|S| = \poly(n)$ this can be done efficiently for any $k$. Finally, for condition (iii), we have that such states can be trivially prepared using $\poly(n)$ quantum gates by using a series of controlled rotations on each qubit at a time. For instance, a very simple application of the algorithm from Grover-Rudolph~\cite{grover2002creating} would suffice.

\paragraph{Reducing the locality}
Finally, we show how to reduce the locality of the constructed Hamiltonian. Assume that we have already increased the fidelity as above, so that the number of gates in the circuit is now $M = \tilde{T} + N$. The spectral gap of $H$, denoted as $\gamma(H)$, can be lower bounded as
\begin{align*}
    \gamma(H) &\geq \min \left[\gamma(H|_{x \in A_{\sc yes}}),\gamma(H|_{x \in A_{\sc no}}) \right]\\
    &= \min \left[\delta,\Omega\left(\frac{1}{M^6}\right)- \Omega\left(\frac{1}{M^7}\right) \right]\\
    &= \delta = \Omega\left(\frac{1}{M^7}\right) = 1/\poly(n).
\end{align*}
Since the ground state is unique and inverse-polynomially gapped (in both the {\sc yes}- and {\sc no}-case), we can apply Lemma~\ref{lem:2localsim} to obtain a $2$-local Hamiltonian $H'$ which $(\Delta, \eta, \epsilon)$-simulates $H$, where we can take $\Delta = 1/\poly(n) \geq \gamma(H)$ sufficiently large, $\eta, \epsilon = 1/\poly(n) \leq \delta$ sufficiently small to ensure that the ground energy remains below some $a'$ in a {\sc yes} instance and above some $b'$ in a {\sc no} instance, such that $b' - a' = \delta = 1/\poly(n)$ and so that $\|\mathcal{E}_{\text{state}}(\ket{g}) - \ket{g'}\| \leq \eta + \mathcal{O}(\gamma^{-1}\epsilon)$, where $\ket{g}$ is $\ket{u_\text{yes}}$ in a {\sc yes} instance or $\ket{u_\text{no}}$ in a {\sc no} instance, $\ket{g'}$ is the ground state of $H'$, and $\mathcal{E}_{\text{state}}(\ket{g})$ is as in Lemma~\ref{lem:2localsim}. That is, $H'$ approximates $H$ in the low energy spectrum (below $\Delta$) in a way that the eigenvalues are perturbed by at most some small inverse-polynomial, and where the ground state can be approximated by the old ground state, plus some semi-classical state added as a tensor product. Finally, note that we can obtain $\norm{H} \leq 1$ by simply scaling down by some polynomial, as required by the problem definition. Note that this will not change any of the statements as all relevant parameters and coefficients are (inverse) polynomials in $n$, albeit of very large degree.

Finally, we can ensure $\norm{H} \leq$ by scaling $H$ with an inverse polynomially large factor.
\end{proof}
Since polynomially-sized subset states are also samplable (see~\cite{gharibian2021dequantizing}), our proof would also go through if one considers a variant of the guidable local Hamiltonian problem which considers samplable states as in Definition~\ref{def:sampaccess} instead. We have the following corollary.
\begin{corollary} $\mathsf{CGaLH}^*(k,\delta,\zeta)$ is $\mathsf{QCMA}$-complete, where the hardness is under randomized reductions, for $k\geq 2$, $\zeta \in (1/\poly(n),1-1/\poly(n))$ and $\delta = 1/\poly(n)$.
\label{cor:CGaLHstar}
\end{corollary}
\begin{proof}
Hardness follows from Theorem~\ref{thm:CGaLH}. Containment follows trivially from the fact that the {\sc yes}-and {\sc no}-cases can be distinguished by using $\text{desc}(u)$ as a witness, and a verifier circuit that prepares the quantum state $\ket{u}$ (which can be done efficiently possible because of the extra condition on $u$) followed by quantum phase estimation to an accuracy strictly smaller than the promise gap $\delta$, see Theorem 2 in~\cite{Cade2022Complexity}.
\end{proof}
Now that we have established $\QCMA$-completeness for $\mathsf{CGaLH}^*$, we get $\QCMA$-completeness for $\mathsf{QGaLH}$ \emph{for free} for the same range of parameter settings, as the latter is a generalization of the former (containing $\mathsf{CGaLH}^*$ as a special case), and containment holds by the same argument as used in the proof of Corollary~\ref{cor:CGaLHstar}. However, with just a little bit of more work we can see that $\QCMA$-hardness for $\mathsf{QGaLH}$ actually persists for a larger range of parameter settings. For this, we will use the following lemma by~\cite{Lin2020nearoptimalground}.
\begin{lemma}[Ground state preparation with a-priori ground energy bound~\cite{Lin2020nearoptimalground}.] Suppose we have a Hamiltonian $H = \sum_k \lambda_k(H) \ket{\phi_q}\bra{\phi_q}$, where $\lambda_k(H) \leq \lambda_{k+1}(H)$, given through its $(\alpha,m,0)$-block-encoding $U_H$. That is, we have access to a $(n+m)$-qubit unitary operator $U$ such that that
\begin{align*}
    \alpha (\bra{0}^{\otimes m} \otimes I) U (\ket{0}^{\otimes m}  \otimes I)=  H.
\end{align*} 
Also suppose we have an initial state $\ket{\psi}$ prepared by some circuit $U_\text{prep}$ with the promise that $|\bra{\phi_0}\ket{\psi}|^2 \geq \Gamma$, and that we have the following bounds on the ground energy and spectral gap: $\lambda_0(H) \leq \mu - \Delta/2 < \mu + \Delta/2 \leq \lambda_1(H)$, for some $\mu,\Delta \in \mathbb{R}$. Then the ground state can be prepared to fidelity $1-\varepsilon$ with probability $1-\nu$ with the following costs:
\begin{enumerate}
    \item Query complexity: $\mO\left(\frac{\alpha}{\Gamma \Delta} \left(\log\left(\frac{\alpha}{\Delta}\right)\log\left(\frac{1}{\Gamma}\right)\log\left(\frac{\log(\alpha/\Delta)}{\nu}\right) + \log\left(\frac{1}{\varepsilon}\right)\right) \right)$ queries to $U_H$ and $\mO\left( \frac{1}{\Gamma} \log\left( \frac{\alpha}{\Delta}\right) \log\left( \frac{\alpha/\Delta}{\nu} \right) \right)$ queries to $U_\text{prep}$,
    \item Number of qubits: $\mO\left(n+m+\log(\frac{1}{\Gamma})\right)$,
    \item Other one- and two-qubit gates: $\mO\left(\frac{m \alpha}{\Gamma \Delta} \left(\log\left(\frac{\alpha}{\Delta}\right)\log\left(\frac{1}{\Gamma}\right)\log\left(\frac{\log(\alpha/\Delta)}{\nu}\right) + \log\left(\frac{1}{\varepsilon}\right)\right) \right)$
\end{enumerate}
\label{lem:gs_prep}
\end{lemma}

\begin{theorem}
$\mathsf{QGaLH}(k,\delta)$ is $\mathsf{QCMA}$-complete for $k\geq 2$, $\delta = 1/\poly(n)$ and $\zeta \in (1/\poly(n),1-1/\exp(n)$.
\label{thm:QGaLH}
\end{theorem} 
\begin{proof}
This follows immediately from the proof of Theorem~\ref{thm:CGaLH}, where the used history states themselves can be prepared by a quantum circuit of at most $\poly(n)$ gates. In the locality reduction, we can only ensure that we remain inverse polynomially close to the original ground state. However, due to Lemma~\ref{lem:gs_prep}, this fidelity is enough to guarantee the existence of a quantum circuit, still polynomial in $n$, that produces a new quantum state which is inverse exponentially close to the actual ground state. Let $U_\text{prep}$ be the quantum circuit that creates $\ket{u} = U_\text{prep} \ket{0}$. Let us assume the worst case setting in our construction of Theorem~\ref{thm:CGaLH}, where we have that $\Gamma  = \zeta = 1/\poly(n)$, $\alpha = \mO(1)$, $\Delta = \gamma(H) = \Omega(1/M^7)$, and $\mu = \lambda_0(H_{\sc no}) = \Omega(1/M^7)$. Let $m=\mO(1)$. Since the inverse fidelity $\varepsilon$ appears only logarithmically in Lemma~\ref{lem:gs_prep}, we can prepare a state that is exponentially-close in fidelity with (exponentially) high probability $1-\nu$ in 
\begin{align*}
     \mO\left(\poly(n) M^7 \left(\log\left(\poly(n) M^7\right)\log\left(\poly(n)\right)\log\left(\frac{\log(M^7)}{\nu}\right) + \log\left(\exp(n)\right)\right) \right) = \tilde{\mO}\left(\poly(n)\right)
\end{align*}
queries to $U_H$ (the block encoding of $H$) and single qubit gates, as well as 
\begin{align*}
    \mO\left( \poly(n) \log\left( M^7 \right) \log\left( \frac{M^7}{\nu} \right) \right) = \tilde{\mO}\left(\poly(n)\right)
\end{align*}
queries to $U_\text{prep}$. 
\end{proof}
There are a few more observations one can make about our proof of Theorem~\ref{thm:CGaLH}. First, if one instead uses the reduction to $\QCMA_1$ instead of the randomized reduction to $\UQCMA$,  the ground state energy is exactly zero in the {\sc yes}-case and inverse polynomially large in the {\sc no}-case. This way we are no longer able to apply the locality reductions (since we need a unique ground state), but this suggests that it might be possible to modify the construction such that it actually shows a result for a quantum satisfiability version for $\QCMA$ (which implies hardness for the more general local Hamiltonian version). In Appendix~\ref{app:GaQSAT} we show that this is indeed possible. 

\section{Classical containment via spectral amplification}
\label{sec:clas_cont}
To complement our quantum hardness results with classical containment results in certain parameter regimes, we will use a technique based on the dequantization of the quantum singular value transformation as described in~\cite{gharibian2021dequantizing}. Our algorithm differs conceptually from the one proposed in~\cite{gharibian2021dequantizing} in the following ways:
\begin{itemize}
    \item We consider a different (and less restrictive) input model: whereas~\cite{gharibian2021dequantizing} considers access to states of the form of Definition~\ref{def:sampaccess}, we use states that adhere to the requirements as in Definition~\ref{def:cds}.
    \item For our purposes, we only consider local Hamiltonians (which are Hermitian sparse matrices) and not arbitrary sparse complex matrices. This simplifies the algorithm in the sense that we can view functions on these Hamiltonians as acting on the \emph{spectrum} instead of the \emph{singular values}.
    \item We also simplify the algorithm by tailoring it exactly to ground state \emph{decision} instead of \emph{estimation} problems, which allows us to use a different function acting on $H$ as compared to~\cite{gharibian2021dequantizing} to solve the relevant problems.
\end{itemize}
Let us introduce and prove bounds on the complexity of the \emph{spectral amplification} algorithm in the next subsection. In the subsequent subsection, we will utilize this algorithm to put classical complexity upper bounds on $\mathsf{CGaLH}(k,\delta,\zeta)$ in specific parameter regimes. 
\subsection{Spectral amplification}
\label{ssec:spec_amp}
Let $H = \sum_{i=0}^{m-1} H_i$ be a Hamiltonian on $n$ qubits which is a sum of $k$-local terms $H_i$, which satisfies $\norm{H} \leq 1$. Since $H$ is Hermitian, we can write $H$ as 
\begin{align*}
    H = \sum_{i =0}^{2^n-1} \lambda_i \ket{\psi_i}\bra{\psi_i},
\end{align*}
where $\lambda_i \in [-1,1]$ (by assumption on the operator norm) denotes the $i$'th eigenvalue of $H$ with corresponding eigenvector $\ket{\psi_i}$. Consider a polynomial $P \in \mathbb{R}[x]$ of degree $d$, and write
\begin{align*}
    P(x) = a_0 + a_1 x + \dots + a_{d} x^{d}.
\end{align*}
The \emph{polynomial spectral amplification} of $H$ for $P$ is then defined as 
\begin{align*}
    P(H) &=  a_0 I +a_1 H +  \dots + a_{d} H^{d}\\
    &=  a_0 I + a_1 \sum_{i =0}^{2^n-1}  \lambda_i \ket{\psi_i}\bra{\psi_i} + \dots + a_d \sum_{i =0}^{2^n-1}  \lambda_i^{d} \ket{\psi_i}\bra{\psi_i}\\
    &= \sum_{i =0}^{2^n-1}  P(\lambda_i) \ket{\psi_i}\bra{\psi_i}.
\end{align*}
Now for $\alpha \in [-1,1]$, denote
\begin{align}
    \Pi_{\alpha} = \sum_{\{i : \lambda_i \leq \alpha\}} \ket{\psi_i}\bra{\psi_i}
    \label{eq:low_energy_projector}
\end{align}
for the projection on all eigenstates of $H$ which have eigenvalues at most $\alpha$, which we will call a \emph{low-energy projector} of $H$. Note that for any $\alpha \geq \lambda_0$, we must have that $\Pi_\text{gs} \Pi_\alpha = \Pi_\alpha \Pi_\text{gs} = \Pi_\text{gs}$. We can utilize such a projector to solve $\mathsf{CGaLH}(k,\delta,\zeta)$, simply by computing $\norm{{\Pi}_{\alpha} \ket{u}}$ for $\alpha = a$ given a classically evaluatable state $u$. To see why this works, note that in the {\sc yes}-case, for the witness $\text{desc}(u)$ we have that $\norm{{\Pi}_{a} \ket{u}} \geq \norm{{\Pi}_\text{gs} \ket{u}} \geq \sqrt{\zeta}$ and in the {\sc no}-case we have that $\norm{{\Pi}_{a} \ket{v}} = 0$ for all states, which means that the two cases are separated by $\sqrt{\zeta}$. However, it is unlikely that an efficient description exists of $\Pi_a$, and even if it did, it would not be $k$-local and therefore $\norm{{\Pi}_{a} \ket{u}}$ would not even be necessarily efficiently computable.  

The idea is now to approximate this low-energy projector $\Pi_\alpha$ by a polynomial in $H$. To see this, note that ${\Pi}_{\alpha}$ can be written exactly as
\begin{align*}
    \Pi_{\alpha} = \frac{1}{2} \left(1-\text{sgn} (H - \alpha I)\right),
\end{align*}
where $\text{sgn}(x)$ is the sign function, which for our purposes is defined on $\mathbb{R}:\rightarrow \mathbb{R}$ as 
\begin{align*}
\text{sgn}(x) = \begin{cases} 1 &\text{ if } x > 0,\\
-1 &\text{ if } x  \leq 0.
\end{cases}    
\end{align*}
From~\cite{LowHamiltonian2017} we can then use the polynomial approximation of the sign function, which can subsequently be shifted to obtain the desired approximate low-energy projector $\tilde{\Pi}_{a}$.
\begin{lemma}[Polynomial approximation to the sign function, from~\cite{LowHamiltonian2017}]
For all $\delta' > 0, \epsilon' \in (0,1/2)$ there exists an efficiently computable odd polynomial $P \in \mathbb{R} [x]$ of degree $d = \mO\left(\frac{\log(1/\epsilon')}{\delta'}\right)$, such that
\begin{itemize}
    \item for all $x \in [-2,2] : |P(x)| \leq 1$, and
    \item for all $x \in [-2,2]\setminus(-\delta', \delta'): |P(x)-\text{sgn}(x)| \leq \epsilon'$.
\end{itemize}
\label{lem:poly_approx_sgn}
\end{lemma}
Lemma~\ref{lem:poly_approx_sgn} is, for the required conditions, optimal in its parameters $\delta'$, and $\epsilon'$~\cite{LowHamiltonian2017}. Since Lemma~\ref{lem:poly_approx_sgn} holds on the entire interval $[-2,2]$, choosing any $\alpha \in [-1,1]$ and scaling the $\text{sgn}(x)$ function with the factor $1/2$ will ensure that the error, as in the lemma, will be $\leq \epsilon/2$. Let $q_\alpha(x): \mathbb{R} \rightarrow [0,1]$ defined as $q_\alpha(x) = \frac{1}{2}(1-\text{sgn}(x-\alpha))$ be this function, with polynomial approximation $Q_\alpha \in \mathbb{R}[x]$ of degree $d$. Note that $Q_\alpha$ can be written as a function of $P$ as $Q_\alpha(x) = \frac{1}{2}(1-P(x-\alpha))$.  We will write $\tilde{\Pi}_{\alpha} = Q_\alpha(H)$ for the corresponding polynomial approximation of the approximate low-energy ground state ``projector''. Note that $\tilde{\Pi}_{\alpha} $ is Hermitian (since $H$ is Hermitian), but that $\tilde{\Pi}_{\alpha}$ is no longer necessarily a projector and therefore $\tilde{\Pi}_{\alpha}^2 \neq  \tilde{\Pi}_{\alpha}$. If we now replace ${\Pi}_{\alpha}$ in $\norm{{\Pi}_{\alpha} \ket{u}}$ by $\tilde{\Pi}_{\alpha}$, we get $\norm{\tilde{\Pi}_{\alpha} \ket{u}} = \sqrt{\bra{u}\tilde{\Pi}_{\alpha}^\dagger \tilde{\Pi}_{\alpha} \ket{u}} = \sqrt{\bra{u}\tilde{\Pi}_{\alpha}^2 \ket{u}} = \sqrt{\bra{u} (Q_\alpha(H))^2 \ket{u}}$, which means that we have to evaluate up to degree $2d$ powers of $H$. The next lemma will give an upper bound on the number of expectation values that have to be computed when evaluating a polynomial of $H$ of degree $d$.

\begin{lemma} Given access to a classically evaluatable state $u$, a Hamiltonian $H = \sum_{i=0}^{m-1} H_i$, where each $H_i$ acts on at most $k$ qubits non-trivially, and a polynomial $P[x]$ of degree $d$, there exists a classical algorithm that computes $\bra{u}P(H)\ket{u}$ in $\mO(m^{d})$ computations of $\bra{u} O_i \ket{u}$, where the observables $\{O_i\}$ are at most $kd$-local.
\label{lem:powers_H}
\end{lemma}
\begin{proof}
We have that
\begin{align*}
     \bra{u} P(H) \ket{u} &=  \bra{u} \left(a_0 I +a_1 H +  \dots + a_{d} H^{d} \right) \ket{u}\\
     &= a_0 +a_1 \bra{u} H \ket{u}+  \dots + a_{d} \bra{u}H^{d} \ket{u}.
\end{align*}
Let $l \in [d]$ be the different powers for which we have to compute $\bra{u} H^l \ket{u}$.  We have that for each $l$ that 
\begin{align*}
    H^l &= \left(\sum_{i=0}^{m-1} H_i \right)^l
\end{align*}
consists of at $m^l$ terms when fully expanded and is Hermitian. However, when the sum is expanded, not every summand is necessarily Hermitian and therefore a local observable. However, we do have that for every $kl$-local summand $Q_{j,l}$ there exists another summand $Q_{j',l}$ which has all the terms of $Q_j$ but in reverse order, unless the $Q_j$ contains only powers of a single term. Grouping those together, we can write
\begin{align*}
    H^l &= \sum_{j=1}^{m + \frac{m^l-m}{2}} \hat{Q}_{j,l},
\end{align*}
where each $\hat{Q}_{j,l}$ is $2l$ local and has $\norm{\hat{Q}_{j,l}} \leq $2. We can now simply absorb the factor $2$ into the coefficients $a_{l'}$, $l' \in \{0,1,\dots,d\}$, such that all $ \hat{Q}_{j,l}$ satisfy $\norm{ \hat{Q}_{j,l}} \leq 1$. The total number of local observables that have to be computed is now equal to
\begin{align*}
    \sum_{l=1}^{d} m + \frac{m^l-m}{2}= \frac{m(m^d + dm - d -1)}{2(m-1)} = \mO(m^d),
\end{align*}
completing the proof.
\end{proof}
All that remains to show is that for constant promise gap $\delta$, using a good enough approximation $\tilde{\Pi}_\alpha$ with a suitable choice of $\alpha$, will ensure that we can still distinguish the two cases in the $\mathsf{CGaLH}(k,\delta,\zeta)$ problem in a polynomial (resp. quasi-polynomial number of computations in $m$ when $\zeta = \Omega(1)$ (resp. $\zeta = 1/\poly(n)$).

\begin{theorem} Let $H = \sum_{i=0}^{m-1} H_i$ be some Hamiltonian, and $\text{desc}(u)$ be a description of a classically evaluatable state $u \in \mathbb{C}^{2^n}$. Let $a,b \in [-1,1]$ such that $b-a\geq \delta$, where $\delta > 0$ and let $\zeta \in (0,1]$. Consider the following two cases of $H$, with the promise that either one holds:
\begin{enumerate}[label=(\roman*)]
    \item $H$ has an eigenvalue $ \leq a$, and $\norm{\Pi_\text{gs} \ket{u}}^2 \geq \zeta $ holds, or
    \item all eigenvalues of $H$ are $\geq b$.
\end{enumerate}
Then there exists a classical algorithm that is able to distinguish between cases (i) and (ii) using
    \begin{align*}
        \mO\left( m^{ c \left( \log(1/\sqrt{\zeta}))/\delta \right) } \right) 
    \end{align*}
    computations of local expectation values, for some constant $c > 0$.
\label{thm:CED}
\end{theorem}
\begin{proof}
Let $\tilde{\Pi}_{\alpha} := Q_\alpha(H)$, where $Q$ is a polynomial of degree $d$, be the approximate low-energy projector that approximates $\Pi_{\alpha} = \frac{1}{2} \left(1-\text{sgn} (H-(\alpha I))\right)$.  We set $\alpha := \frac{a+b}{2}$, $\delta':= \delta/2$ and $\epsilon'= 1/10$.  We propose the following algorithm:
\begin{enumerate}
    \item Compute $\norm{\tilde{\Pi}_{a} \ket{u}}$ using a polynomial of degree $2d$ where $d = \mO(\log(1/\epsilon'))/\delta'$, for $\epsilon' := \frac{1}{10}\sqrt{\zeta} $ and $\delta' = \delta/2$. 
    \item If $\norm{\tilde{\Pi}_{\alpha} \ket{u}} \geq \frac{9}{10} \sqrt{\zeta}$, output (i) and output (ii) else. 
\end{enumerate}
Clearly, by Lemma~\ref{lem:powers_H}, we have that this can be done in at most 
    \begin{align*}
        \mO\left( m^{ c \left( \log(1/\sqrt{\zeta}))/\delta \right) } \right) 
    \end{align*}
computations of expectation values of local observables, for some constant $c$. Let us now prove the correctness of the algorithm. Note that we can write  $\tilde{\Pi}_{\alpha}$ as
\begin{align*}
    \tilde{\Pi}_{\alpha} = \sum_{i =0}^{2^n-1} Q(\lambda_i) \ket{\psi_i}\bra{\psi_i}, 
\end{align*}
where we have that
\begin{align*}
      \begin{cases} 1-\sqrt{\zeta}/2 \leq Q(\lambda_i) \leq 1 &\text{ if } \lambda_i \leq a, \\
     0 \leq Q(\lambda_i) \leq \zeta/2 &\text{ if } \lambda_i \geq  b,\\
    0 \leq Q(\lambda_i) \leq 1 &\text{ else,} 
    \end{cases}
\end{align*}
by Lemma~\ref{lem:poly_approx_sgn}. Let us analyse both case (i) and (ii) separately.
\begin{enumerate}[label=(\roman*)]
    \item $H$ has an eigenvalue $\leq a$, and $\norm{\Pi_\text{gs} \ket{u}}^2 \geq \zeta $ holds:
    \begin{align*}
        \norm{\tilde{\Pi}_{\alpha} \ket{u}} &\geq \norm{\tilde{\Pi}_{\alpha} \Pi_\text{gs}\ket{u}}\\
        &= \norm{\Pi_{\alpha} \Pi_\text{gs}\ket{u}   -(\Pi_{\alpha}-\tilde{\Pi}_{\alpha}) \Pi_\text{gs}\ket{u}}\\
        &= \norm{\Pi_\text{gs}\ket{u}  - \left(\sum_{i:\lambda_i \leq \alpha} \ket{\psi_i}\bra{\psi_i} - \sum_{i =0}^{2^n-1} Q(\lambda_i) \ket{\psi_i}\bra{\psi_i} \right)\Pi_\text{gs}\ket{u}}\\
        &= \norm{\Pi_\text{gs}\ket{u}  - \left(\sum_{i:\lambda_i \leq \alpha} (1-Q(\lambda_i)) \ket{\psi_i}\bra{\psi_i} - \sum_{i: \lambda_i > \alpha} Q(\lambda_i) \ket{\psi_i}\bra{\psi_i} \right)\Pi_\text{gs}\ket{u}}\\
        &\geq \norm{\Pi_\text{gs}\ket{u}  - \left(\sum_{i:\lambda_i \leq \alpha} \frac{1}{10} \ket{\psi_i}\bra{\psi_i} \right)\Pi_\text{gs}\ket{u}}\\
        &= \norm{\Pi_\text{gs}\ket{u}  - \frac{1}{10} \left(\sum_{i:\lambda_i \leq \alpha}  \ket{\psi_i}\bra{\psi_i} \right)\Pi_\text{gs}\ket{u}}\\
        &= \norm{\Pi_\text{gs}\ket{u}  - \frac{1}{10} \Pi_\alpha \Pi_\text{gs}\ket{u}}\\
        &= \norm{\Pi_\text{gs}\ket{u}  - \frac{1}{10} \Pi_\text{gs}\ket{u}}\\
        &= (1-\frac{1}{10}) \norm{\Pi_\text{gs}\ket{u}}\\
        &\geq \frac{9}{10} \sqrt{\zeta}.
    \end{align*}
    \item all eigenvalues of $H$ are $\geq b$: \\
We must have that
\begin{align*}
        \norm{\tilde{\Pi}_{\alpha} \ket{u}} \leq \frac{1}{2} \sqrt{\zeta},
    \end{align*}
    since $\lambda_i \geq b$ for all $i \in \{0,\dots,2^n-1 \}$.
\end{enumerate}
Hence, we have that the promise gap between both cases is lower bounded by
\begin{align*}
    \frac{9}{10} \sqrt{\zeta} - \frac{1}{2} \sqrt{\zeta} = \frac{2}{5} \sqrt{\zeta},
\end{align*}
which is $1/\poly(n)$ when $\zeta \geq 1/\poly(n)$.
\end{proof}

\begin{remark} It should be straightforward to adopt the same derivation as above to a more general setting by considering sparse matrices, a promise with respect to fidelity with the low-energy subspace (i.e.~all states with energy $\leq \lambda_0 + \gamma$ for some small $\gamma$), as well as $\epsilon>0$ for $\epsilon$-classically evaluatable states (see Definition~\ref{def:cds}). However, this would likely put constraints on $\gamma$ and $\epsilon$, where $\epsilon$ in principle has to scale inversely proportional to the number of local terms in the Hamiltonian.
\end{remark}

\begin{figure}[h!]
    \centering\includegraphics[width=0.6\linewidth]{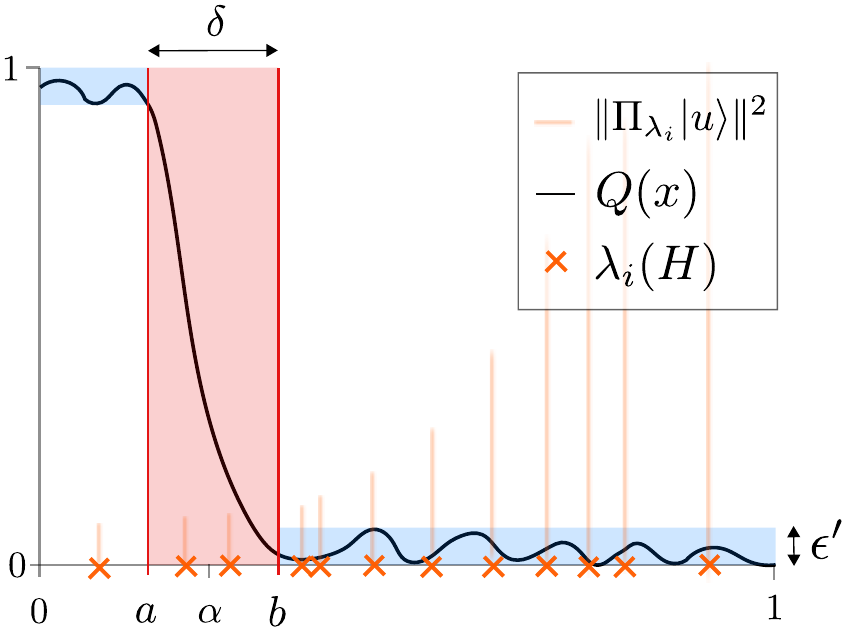}
    \caption{Illustration of the approximate low-energy projector $\Pi_\alpha$ in both the {\sc yes}-case with $\alpha = \frac{a+b}{2}$. The orange crosses correspond to the energy values, and the attached shaded lines indicate the fidelity of the guiding state with the space spanned by all eigenstates $\ket{\psi_l}$ of $H$ that have energy at most $\lambda_i$. The polynomial approximation of the shifted sign function is displayed as $Q_\alpha(x)$, and the $\epsilon'$-error approximation regimes are indicated with the blue-shaded areas. In the red regime we do not have tight bounds on the error, except that the function values are in $[0,1]$. For small enough $\epsilon$, in the {\sc yes}-case the contribution of the ground state to the value of $\norm{\tilde{\Pi}_\alpha \ket{u}}^2$ should be larger than that computed in the {\sc no}-case due to contributions of higher energy values, as a result from an inexact implementation of the low-energy projector. In the {\sc no}-case, all energy values will be larger than $b$.}
    \label{fig:AGPS}
\end{figure}

\subsection{Classical hardness and containment}
All results in this section also hold when `$\mathsf{CGaLH}$' is replaced by `$\mathsf{CGaLH}^*$', as the containment trivially follows since $\mathsf{CGaLH}$ generalises $\mathsf{CGaLH}^*$ and the hardness construction uses a diagonal (i.e.~classical) Hamiltonian, of which the ground states are basis states and can thus be prepared on a quantum computer. To be able to make completeness statements when we consider $\NP$, let us start by first proving a (straightforward) hardness result. 

\begin{lemma} $\mathsf{CGaLH}(k,\delta,\zeta)$ is $\NP$-hard for $k\geq 2$, $\delta \leq \mO(1)$ and $\zeta \leq 1$, where $k,\delta,\zeta$ can also be functions of $n$.
\label{lem:CGaLH_NP-hard}
\end{lemma}
\begin{proof} We will prove this by a reduction from gapped $3$-SAT. Let $\gamma$-$3$-SAT be a promise decision problem where we are given a formula $\phi(x) = \frac{1}{m}\sum_{i=0}^{m-1} C_i$ with $C_i = x_{i_1} \lor x_{i_2} \lor x_{i_3}$, with the promise that either $\phi(x) = 1$ (output {\sc yes}) or  $\phi(x) \leq \gamma$ (output {\sc no}), where $\gamma \in (0,1)$. From (one of) the (equivalent) PCP theorem(s) we know that there exists a constant $\gamma \in (0,1)$ for which deciding on the correct output (we are allowed to output anything if the promise doesn't hold) is $\NP$-hard~\cite{HastadSome1997}. Next, we apply the gadget from~\cite{Ansoteguireducing2021}, which maps $\phi(x)$ to a $2$-SAT instance with formula $\phi'(x) = \frac{1}{10m}\sum_{i=0}^{m-1} \sum_{j \in [10]} C'_{i,j} $. Here we have that $\phi'(x)$ has the property that for every clause $C_i$ there are $10$ corresponding clauses $C'_{i,j}$, $j\in[10]$ such that, if a given assignment $x$ satisfies a clause $C_i$ of $\phi(x)$, then exactly $7$ clauses of $C'_{i,j}$ can be satisfied, and for those $C_i$ that are not satisfied by $x$, at most $6$ clauses of $C'_{i,j}$ are satisfied. Note that if $\phi(x) = 1$, we then must have that $\phi'(x) = 7/10$, and that if $\phi(x) \leq \gamma$, we have that $\phi'(x) \leq 7\gamma/10 + 6(1-\gamma)/10 = (\gamma + 6)/10$. Hence, it is still $\NP$-hard to distinguish between those cases, and the promise gap between the {\sc yes}- and the {\sc no}-case is $\frac{1}{10}(1-\gamma)=:\gamma'$, which is some constant.  Let us now map $\phi'(x)$ into a $2$-local diagonal Hamiltonian $H'$ such that $\bra{x}H'\ket{x} = \phi(x)$, which can be done by a representation of the clauses as diagonal matrices. By our choice of $\phi'(x)$, we have already ensured that the Hamiltonian is (sub)-normalized. To turn the problem into a minimization problem, one can simply invert the spectrum by letting $H = I-H'$ (note that $H' \succeq 0)$. The eigenvectors of $H$ are basis vectors -- and thus themselves classically evaluatable states for which $\zeta = 1 = \mO(1)$ -- and its eigenvalues are precisely the function evaluations of $1-\phi'(x)$. Hence, setting $a := 3/10$ and $b := (4-\gamma)/10$ gives us $\delta = \gamma' = \mO(1)$.
\end{proof}
Theorem~\ref{thm:CED} now gives us a very easy way to establish the following upper bounds on $\mathsf{CGaLH}(k,\delta,\zeta)$ when the required precision $\delta$ is only constant. Combined with Lemma~~\ref{lem:CGaLH_NP-hard}, we obtain the following result, reminiscent of Theorem 5 in~\cite{gharibian2021dequantizing}.
\begin{theorem}
$\mathsf{CGaLH}(k,\delta,\zeta)$ is $\NP$-complete for $k = \mO(\log(n))$, and constants $\delta \in (0,1]$ and $\zeta \in(0,1]$. Furthermore, when $\zeta = 1/\poly(n)$ we have that $\mathsf{CGaLH}(k,\delta,\zeta)$ is in $\NqP$.
\label{thm:cp_NP_NqP}
\end{theorem}
\begin{proof}
$\NP$-Hardness follows from Lemma~\ref{lem:CGaLH_NP-hard}. The containment statements follow from Theorem~\ref{thm:CED}, in which the proposed algorithm for $m = \mO(n^k)$, $\delta \in (0,1]$ constant runs in polynomial time when $\zeta$ is constant and in quasi-polynomial time when $\zeta = 1/\poly(n)$ (using the fact that $\mO(n^{\log(n)})= 2^{\mO( \log^c (n))} $ for some constant $c >0$).
\end{proof}
Moreover, by a little more careful inspection one can show that the problem's hardness depends on how $\zeta$ and $\delta$ relate to one another, as shown in the following theorem.

\begin{theorem} Let $f(n): \mathbb{N} \rightarrow \mathbb{R}_{> 0}$, $g(n): \mathbb{N} \rightarrow \mathbb{R}_{> 0}$ be some functions with the property that there exists some constant $n_0$ (which is known), such that for all $n \geq n_0$ we have that $1/g(n)-1/f(n) >0$. Then we have that $\mathsf{CGaLH}(k,\delta, \zeta)$ is $\NP$-complete for $k\geq 2$, $\delta = 1/g(n)$ and $\zeta \geq 1-1/f(n)$. 
\label{thm:CGaLH_close}
\end{theorem}
\begin{proof}
\noindent\textit{Hardness:}
$\NP$-hardness follows again trivially from Lemma~\ref{lem:CGaLH_NP-hard}, by taking $g(n) = \mO(1)$ and letting $f(n)$ be some arbitrarily large function such that $f(n) \gg g(n)$, which gives a loose lower bound on $\zeta$ which can be as large as exactly $1$.

\noindent\textit{Containment:}
To prove this we split the regime into the $n \geq n_0$ case and the $n < n_0$ case, giving separate algorithms for both cases.

\paragraph{\bf $n < n_0$ case:}
In this setting we have that $n \in [1,n_0]$, which is a constant. Hence, we can simply diagonalize the full Hamiltonian and compute its ground state energy in time upper bounded by some constant.

\paragraph{\bf $n\geq n_0$ case:} 
As always, denote $\Pi_\text{gs}$ for the projector on the ground space of $H$. The verifier expects to be given a $\text{desc}(u)$ such that $\norm{\Pi_\text{gs} \ket{u}}^2 \geq 1-1/f(n)$ and checks if $\bra{u}H\ket{u} \leq a + 1/f(n)$. Let us now check completeness and soundness of this simple protocol. By the definition of the problem, we must have that $\norm{H} \leq 1$. We can write $\ket{u} = \alpha_1 \ket{\phi_0}  + \alpha_2 \ket{\phi_0^\perp}$ for $|\alpha_1|^2 + |\alpha_2|^2 = 1$. Here $\ket{\phi_0}$ lives in the ground space of $H$ and $\ket{\phi_0^\perp}$ in the subspace orthogonal to the ground state. Note that $|\alpha_2|^2 = 1-\norm{\Pi_\text{gs} \ket{u}}^2 \leq \frac{1}{f(n)}$.  Therefore, we must have that in the {\sc yes}-case 
\begin{align*}
    \bra{u}H\ket{u} &= |\alpha_1|^2 \bra{\phi_0} H \ket{\phi_0} + |\alpha_2|^2 \bra{\phi_0^\perp} H  \ket{\phi_0^\perp}\\
    &\leq a+|\alpha_2|^2 \bra{\phi_0^\perp} H  \ket{\phi_0^\perp}\\
    &\leq a+ \norm{H}/f(n)\\
    &\leq a + 1/f(n) .
\end{align*}
In the {\sc no}-case, we can simply evoke the variational principle
\begin{align*}
    \bra{u}H\ket{u} & \geq \lambda_0(H) \geq b = a + 1/g(n) > a + 1/f(n), 
\end{align*}
for all $\ket{u}$, by the assumption of the functions for $n \geq n_0$. Therefore, the two cases are separated, and can therefore be distinguished from one another (using the fact that for our definition of classically evaluatable states the expectation of local observables can be computed exactly\footnote{In practice, one would want the difference to be large enough to be able to detect them with machine precision.}).
\end{proof}

\section{Quantum-classical probabilistically checkable proofs}
\label{sec:QCPCP}
In this section we initiate the study of a new complexity class that sits right between the classical and quantum PCPs. First, let us recall some basic definitions and facts about PCPs.
\begin{definition}[Probabilistically checkable proofs (PCPs)] 
Let $n\in \mathbb{N}$ be the input size and $q: \mathbb{N} \rightarrow \mathbb{N},r :\mathbb{N} \rightarrow \mathbb{N}$. A promise problem $A = (A_\text{yes},A_\text{no})$ has a $(r(n),q(n))$-$\PCP$ verifier if there exists a polynomial-time probabilistic algorithm $V$ which takes an input $x \in \{0,1\}^n$, and has random access to a string $\pi \in \{0,1\}^{*}$ of length at most $q(n) 2^{r(n)}$, uses at most $r(n)$ random coins and makes at most $q(n)$ non-adaptive queries to locations of $\pi$, such that 
\begin{itemize}
\item\textbf{Completeness.} If $x\in A_\text{yes}$, then there is a proof $\pi$ such that $V^\pi(x)$ accepts with certainty.
\item \textbf{Soundness.} If $x \in A_\text{no}$, then for all proofs $\pi$ we have that $V^\pi(x)$ accepts with probability at most $\frac{1}{2}$.
\end{itemize}
A promise problem $A = (A_\text{yes},A_\text{no})$ belongs to $\PCP[q,r]$ if it has a $(r(n),q(n))$-$\PCP$ verifier.
\end{definition}
The celebrated PCP theorem states that $\NP = \PCP[\mO(1),\mO(\log n)]$~\cite{Arora1998proof,Arora1998probabilistic}. It also implies that there exists a constant $\alpha$ such that it is $\NP$-hard to decide for a constraint satisfaction problem with `promise gap' $\alpha$. Dinur~\cite{Dinur2007thepcp} showed that this implication can be obtained directly, by reducing from a constraint satisfaction problem with inverse polynomial promise gap to one with a constant promise gap, whilst retaining $\NP$-hardness. This type of reduction is commonly referred to as \textit{gap amplification}.

\

\noindent A quantum variant to $\PCP$ can naturally be defined as follows~\cite{Aharonov2008The,aharonov2013guest}.

\begin{definition}[Quantum Probabilistically Checkable Proofs ($\QPCP$)] Let $n\in \mathbb{N}$ be the input size and $p,q : \mathbb{N} \rightarrow \mathbb{N}$, $c,s : \mathbb{R}_{\geq 0} \rightarrow \mathbb{R}_{\geq 0} $ with $c-s >0$. A promise problem $A = (A_\text{yes},A_\text{no})$ has a $(p(n),q(n),c,s)$-$\QPCP$-verifier if there exists a quantum algorithm $V$ which acts on an input $\ket{x}$ and a polynomial number of ancilla qubits, and takes as additional input a quantum state $\ket{\xi} \in \left(\mathbb{C}^{2}\right) ^{\otimes p(n)}$, from which it is allowed to access at most $q(n)$ qubits, followed by a measurement of the first qubit after which it accepts only if the outcome is $\ket{1}$, such that
\begin{itemize}
\item[~] \textbf{Completeness.} If $x\in A_\text{yes}$, then there is a quantum state $\ket{\xi}$ such that the verifier accepts with probability at least $c$,
\item[~] \textbf{Soundness.} If $x \in A_\text{no}$, then for all quantum states $\ket{\xi}$ the verifier accepts with probability at most $s$.
\end{itemize}
A promise problem $A = (A_\text{yes},A_\text{no})$ belongs to $\QPCP[p,q,c,s]$ if it has a $(p(n) , q(n),c,s)$-$\QPCP$ verifier. If $p(n) \leq \poly(n)$, $c=2/3$, and $s=1/3$, we simply write $\QPCP[q]$. 
\label{def:QPCP}
\end{definition}

\noindent And likewise, there is a formulation of the \textit{quantum} PCP conjecture, using the above notion quantum probabilistically checkable proofs (QPCPs).
\begin{conjecture}[QPCP conjecture - proof verification version] There exists a constant $q \in \mathbb{N}$ such that 
\begin{align*}
    \QMA = \QPCP[q].
\end{align*}
\label{conj:conj1}
\end{conjecture}
As in the classical PCP theorem, there exist equivalent formulations of the QPCP conjecture. In particular, one can formulate the PCP theorem in terms of \emph{inapproximability} of constraint satisfaction problems (CSPs), analogously to the classical setting. In the context of the quantum complexity classes (notably $\QMA$), `quantum' CSPs are generalized by local Hamiltonian problems. The formulation of the quantum PCP conjecture in terms of inapproximability of local Hamiltonians is:

\begin{conjecture}[QPCP conjecture - gap amplification version] There exists a (quantum) reduction from the local Hamiltonian problem with promise gap $1/\poly(n)$ to another instance of the local Hamiltonian problem with promise gap $\Omega(1)$.
\label{conj:conj2}
\end{conjecture}
\noindent It is well known that, at least under \emph{quantum reductions}, both conjectures are in fact equivalent:
\begin{fact}[\cite{Aharonov2008The}]Conjecture~\ref{conj:conj1} holds if and only if conjecture~\ref{conj:conj2} holds.
\label{fact:QPCP}
\end{fact}
\subsection{Quantum-classical PCPs}
We now consider the notion of a quantum PCP conjecture that, intuitively, conjectures the existence of polynomial-time quantum verifiers for $\QCMA$ problems that need only check a constant number of (the now classical) bits of the proof to satisfy constant completeness and soundness requirements. More formally, this reads as:

\begin{definition}[Quantum-Classical Probabilistically Checkable Proofs ($\QCPCP$)] Let $n\in \mathbb{N}$ be the input size and $p, q : \mathbb{N} \rightarrow \mathbb{N}$, $c,s : \mathbb{R}_{\geq 0} \rightarrow \mathbb{R}_{\geq 0} $ with $c-s >0$. A promise problem $A = (A_\text{yes},A_\text{no})$ has a $(p(n),q(n),c,s)$-$\QCPCP$-verifier if there exists a quantum algorithm $V$ which acts on an input $\ket{x}$ and a polynomial number of ancilla qubits, plus an additional bit string $y \in \{0,1\}^{p(n)}$ from which it is allowed to read at most $q(n)$ bits (non-adaptively), followed by a measurement of the first qubit, after which it accepts only if the outcome is $\ket{1}$, and satisfies:
\begin{itemize}
\item[~] \textbf{Completeness.} If $x\in A_\text{yes}$, then there is a $y \in \{0,1\}^{p(n)}$ such that the verifier accepts with probability at least $c$,
\item[~] \textbf{Soundness.} If $x \in A_\text{no}$, then for all $y \in \{0,1\}^{p(n)}$ the verifier accepts with probability at most $s$.
\end{itemize}
A promise problem $A = (A_\text{yes},A_\text{no})$ belongs to $\QCPCP[p,q,c,s]$ if it has a $(p(n) , q(n),c,s)$-$\QCPCP$ verifier. If $p(n)=\mathcal{O}(\poly(n))$, $c=2/3$, and $s=1/3$, we simply write $\QCPCP[q]$. 
\label{def:QCPCP}
\end{definition}

We remark that there are likely several ways to characterise a PCP for $\QCMA$, with some being more or less natural than others. With that said, we believe that the above characterisation is well-motivated for the following reasons:
\begin{enumerate}
    \item It is a natural definition following the structure of a $\QPCP$ verifier as in Definition~\ref{def:QPCP}, now with proofs given as in the standard definition of $\QCMA$ (see Definition~\ref{def:QCMA}). 
    \item $\QCPCP[\mO(1)]$ captures the power of $\BQP$ as well as $\NP$ (via the PCP theorem), which are both believed to be strictly different complexity classes. Since techniques used to prove the PCP theorem are difficult (or impossible) to translate to the quantum setting~\cite{aharonov2013guest}, studying $\QCPCP[\mO(1)]$ might provide a fruitful direction with which to obtain the first non-trivial lower bound on the complexity of $\QPCP[\mO(1)]$. Indeed, the currently best known lower bound on the complexity of $\QPCP[\mO(1)]$ is $\NP$ via the PCP theorem. 
\end{enumerate}
\ 
Given this definition for QCPCPs, our `quantum-classical' PCP conjecture is naturally formulated as follows.
\begin{conjecture} [quantum-classical PCP conjecture] There exists a constant $q\in \mathbb{N}$ such that 
\begin{align*}
 \QCMA = \QCPCP[q].
\end{align*}
\label{conj:QCPCP}
\end{conjecture}
\begin{figure}
  \centering
\includegraphics[scale=1.7]{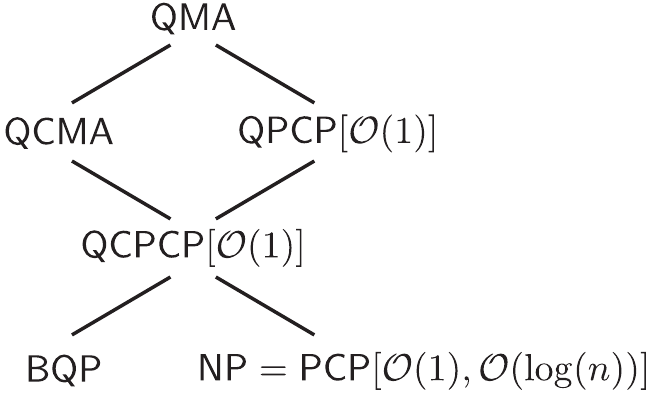}
  \caption{ Known inclusions between some complexity classes and our proposed class $\QCPCP[q]$, with $q= \mO(1)$. A line drawn from complexity class $A$ to another class $B$, where $B$ is placed above $A$ means that $A \subseteq B$. Note that the complexity of our proposed class $\QCPCP[\mO(1)]$ is non-trivial, as it contains both $\NP$ and $\BQP$ for which it is believed that both $ \BQP \not\subset  \NP$ and $ \BQP \not\supset  \NP$.}
  \label{fig:PCPs_inclusions}
\end{figure}
If true, this conjecture would give a  `$\QCMA$ lower bound' on the power of quantum PCP systems, showing that a PCP theorem holds for (quantum) classes above $\NP$, taking a step towards proving the quantum PCP conjecture. If it is false, but the quantum PCP conjecture is true, then this suggests that $\QPCP$ systems must take advantage of the quantumness of their proofs to obtain a probabilistically checkable proof system. In particular, since $\QCMA \subseteq \QMA$, this would imply the existence of a quantum PCP system for every problem in $\QCMA$, but \emph{not} a quantum-classical one, even though the problem admits a classical proof that can be efficiently verified when we are allowed to look at all of its bits. 

Note that we have that $\QCPCP[\mO(1)]$ trivially contains $\mathsf{CGaLH}(k,\zeta,\delta)$ whenever $k = \mO(\log n)$ $\zeta = \Omega(1)$ and $\delta = \Omega(1)$, since we have shown in Theorem~\ref{thm:cp_NP_NqP} that this problem is in $\NP$ and therefore admits a classical PCP system to solve the problem. This would not apriori be clear if we considered guiding states as in Definition~\ref{def:sampaccess}, as we do not know a PCP for the class $\MA$.

\subsubsection{Useful facts about $\QCPCP$}
We begin by showing two basic properties of $\QCPCP$, which we make use of later on. First, we show the simple fact that $\QCPCP$, just as is the case for $\PCP$, allows for error reduction at the cost of extra queries to the proof. 
\begin{proposition}[Strong error reduction for $\QCPCP$] The completeness and soundness parameters in $\QCPCP$ can be made exponentially close in some $t \in \mathbb{N}$ to $1$ and $0$, respectively, i.e.~$c=1-2^{-\mO(t)}$ and $s=2^{-\mO(t)}$, by making $tq$ queries instead of $q$ to the classical proof.
\label{prop:error_reduction}
\end{proposition}
\begin{proof}
This follows from a standard parallel repetition argument, running the $\QCPCP$ protocol $t$ times in parallel and taking a majority vote on the outcomes. Since the proof is classical, it does not matter if the same parts of the proof are queried multiple times by different runs, unlike the case when the proof is quantum. By taking a majority vote on the $t$ outcomes, this yields the desired result by a Chernoff bound.
\end{proof}
\begin{figure}
    \centering
    \includegraphics[width=\linewidth]{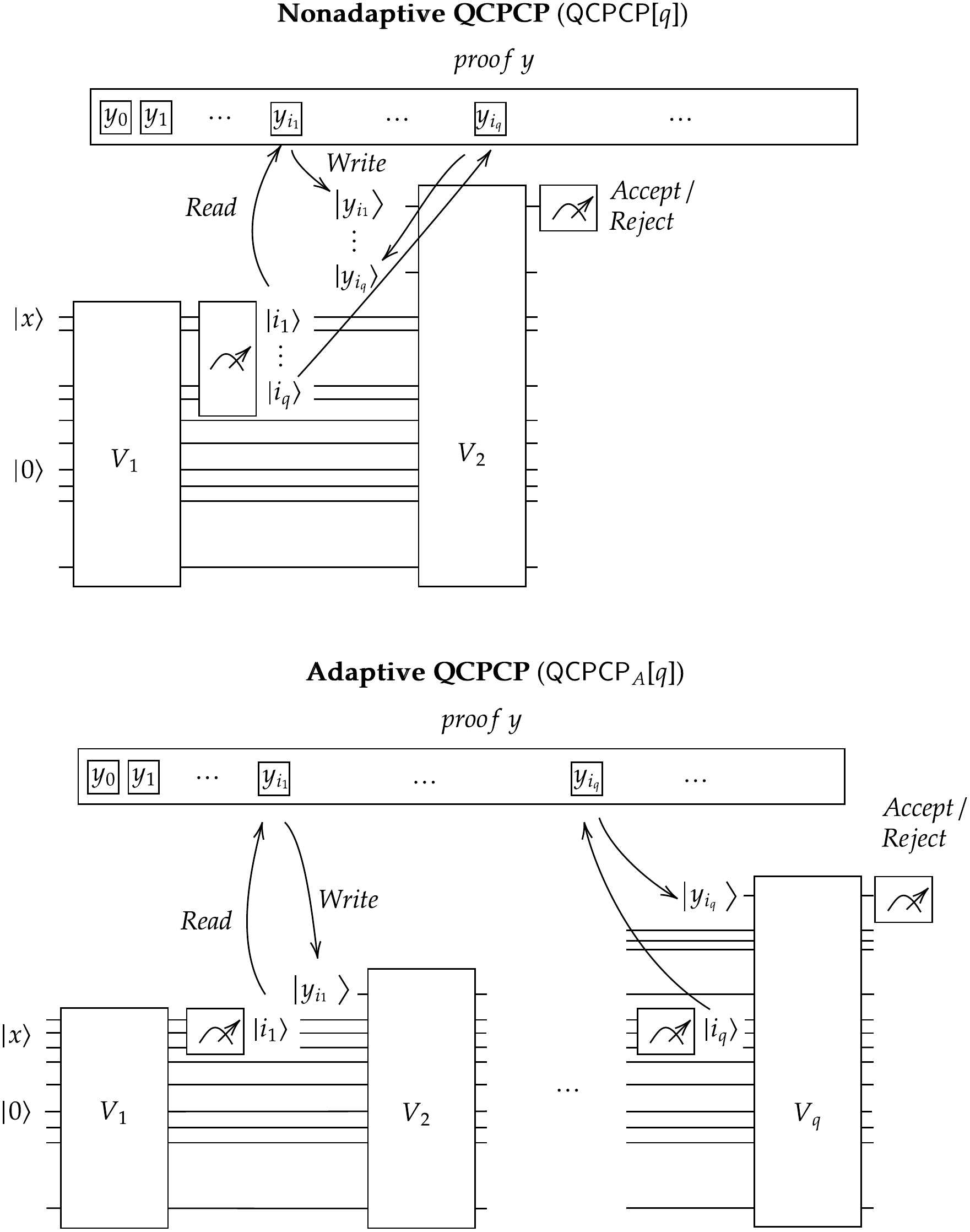}
    \caption{General quantum circuits for nonadaptive QCPCPs ($\QCPCP[q]$, top) versus adaptive ones ($\QCPCP_A[q]$, bottom). When we talk about making a `query' to the classical proof we mean a single read/write procedure, where a single proof bit is read from the proof and then used to initialise a qubit into the basis state corresponding to the value of the bit. To determine which part of the proof must be read, a measurement of some designated index qubits is performed, which outputs a basis state(s) corresponding to index (indices) of the proof. The entire post-measurement state is allowed to be an input to a subsequent quantum circuit.}
    \label{fig:QCPCP_overview}
\end{figure}
Second, we show that when one is interested in $\QCPCP[\mO(1)]$, the non-adaptiveness restriction in the definition does not limit the power of the class. The proof of this is similar to how one would prove it for $\PCP[\mO(1),\mO(\log n)]$, but encounters one difficulty. In the classical setting, the probabilistic PCP verifier can be replaced by a deterministic one that has an auxiliary input for a string which is taken uniformly at random. This way, one can `fix the randomness' in the verifier, and the queries to the proof form a decision tree on the choices over the values of the proof bits, of which only one path can correspond to the actual proof. However, in the quantum setting there is no such equivalent notion of `fixing quantumness', and therefore at every step the circuit might output different values for the indices of the proof that are supposed to be queried, even when all preceding steps gave the same indices and yielded identical values for the proof queries as compared to another run. Fortunately, it is easy to fix this by simply running the circuit many times in a simulated manner, using randomly chosen `fake' proof bits. The probability that one matches the actual proof bits with the randomly chosen ones is small, but still constant when the number of queries to the proof is constant. We formalize this in the following theorem.

\begin{theorem} Let $ \QCPCP_A$ be just as $ \QCPCP$ but with the power to make adaptive queries to the proof. We have that
\begin{align*}
    \QCPCP[\mO(1)] = \QCPCP_A[\mO(1)].
\end{align*}
\label{thm:QCPCP_A}
\end{theorem}
\begin{proof} The `$\subseteq$' is trivial since adaptive queries generalise non-adaptive ones, so we only have to show `$\supseteq$'. Let $V$ be the circuit description of an adaptive $\QCPCP[\mO(1)]$ protocol which makes $q$ queries to a classical proof $y$. W.l.o.g., we can view $V$ as a quantum circuit that consists of applying a unitary $V_0$ to the input and some ancilla qubits all initialized in $\ket{0}$, followed by a measurement to determine which bit of the proof to read, after which a new qubit is initialized in the basis state corresponding to the proof bit value, followed by another unitary $V_1$ that acts on all qubits, and so on, until at the end a measurement of a single designated output qubit is performed to decide whether to accept or reject (see Figure~\ref{fig:QCPCP_overview}). We can simulate $V$ on proof $y$ non-adaptively in the following way. For a total of $T = C 2^q$ times, with $C > 1$ some constant, we pick  a $z \in \{0,1 \}^q$ uniformly at random, and run circuit $V$ where the answers to the proof queries are taken to be the bits of $z$. The final output qubit is measured as usual. This yields a total of $T$ tuples $(z,i_1,\dots,i_q,o)$, where the first entry indicates the value of the sampled $z \in \{0,1\}^q$, the following $q$ entries the $q$ indices of the bits that were supposed to be queried by $V$, and the last one the outcome of the final measurement, which has $o = 1$ if the circuit accepted and $o = 0$ if it rejected. For all $T$ runs of $V$, we can check if each $(z,i_1,\dots,i_q)$ is consistent with the actual proof $y$, i.e.~whether $z_{i_l} = y_{i_l}$ for all $l \in [q]$, which requires at most $q T = q C 2^q = \mO(1)$ queries to $y$. Since the probability that a randomly selected $z$ is consistent with the proof $y$ is $1/2^q$, we have that the probability that at least one of the $z$ values in the tuples is consistent with the proof satisfies
\begin{align*}
\mathbb{P}[\text{At least one $z_{i_1},\dots,z_{i_q}$ consistent with $y$}] \geq
1-\left( 1-\frac{1}{2^q} \right)^T \geq 1-\frac{1}{1+\frac{T}{2^q}} = \frac{T}{2^q + T} = \frac{C}{1+C}.
\end{align*}
Now we loop over all collected tuples, and for the first tuple we encounter that is consistent with the proof we check $o$: if $o=1$ we accept, and if $o=0$ we reject. If all of the tuples are inconsistent with the proof, we reject. Hence, if the $\QCPCP_A$ protocol has completeness $c$ and soundness $s$, this $\QCPCP$ protocol has completeness $c' = 2C / 3(1+C) $ and soundness $1/3$. We then have that the {\sc yes}- and {\sc no}-cases are separated by at least
\begin{align*}
    \frac{2 C }{3(1+C)} -\frac{1}{3}=  \Omega(1).
\end{align*}
if $C >1$. This can easily be boosted such that the completeness and soundness are $\geq 2/3$ and $\leq 1/3$ at the cost of a constant multiplicative factor in the number of queries, by adopting the strong error reduction of Proposition~\ref{prop:error_reduction}. Hence,  $\QCPCP[\mO(1)] \supseteq \QCPCP_A[\mO(1)]$, completing the proof.
\end{proof}

\begin{remark} The above proof actually holds for all quantum algorithms that have `non-quantum' query access to a classical string and make only a constant of such queries.
\end{remark}
There is an interesting observation one can make when looking more closely at the proof of Theorem~\ref{thm:QCPCP_A}: the resulting non-adaptive protocol employs a quantum circuit only for generating indices/output data needed for queries from the proof. Hence, once the queries are made, all subsequent checks can be performed classically. This means a quantum-classical PCP that is restricted to only using quantum circuits before the proof is accessed is just as powerful as the more general definition as per Definition~\ref{def:QCPCP}. In the next subsection we will see that this idea allows one to put a non-trivial upper bound on the complexity of QCPCPs that make only a constant number of queries to the proof.

\subsection{Upper bound on $\QCPCP$ with constant proof queries}
Here we show that $\QCPCP$ with a constant number of proof queries is contained in $\BQP^{\NP[1]}$, i.e.~in $\BQP$ with only a single query to an $\NP$-oracle. The proof is rather long, but the idea is simple: just as is the case for $\QPCP$,  a \emph{quantum} reduction can be used to transform a $\QCPCP$ system into a local Hamiltonian problem. However, since the proof is now classical, we can either use the CNOT-trick of Lemma~\ref{lem:CNOT_trick} to ensure that the resulting Hamiltonian problem is diagonal in the computational basis, or directly learn a diagonal (i.e.~classical) Hamiltonian that captures the input/output behaviour of the $\QCPCP$-circuit on basis state inputs.\footnote{We use the language of Hamiltonians here to be consistent with the $\QPCP$ setting, but this is not necessary. Since the Hamiltonian is classical, we might as well say that we learn a polynomial $P : \{0,1\}^{p(n)} \rightarrow \mathbb{R}_{\geq 0}$.} The main technical work required is to derive sufficient parameters in the reduction, thereby ensuring that the reduction succeeds with the desired success probability. We first prove a lemma which upper bounds the norm difference between a certain `learned' Hamiltonian satisfying certain accuracy constraints, and the actual Hamiltonian.

\begin{lemma} Let $H = \sum_{i \in [m]} w_i H_i$ be a $k$-local Hamiltonian consisting of weights $w_i \in [0,1]$ such that $\sum_{i \in [m]} w_i = W$, and $k$-local terms $H_i$ for which $\norm{H_i} \leq 1$ for all $i \in [m]$. Let $\Omega_{\geq \gamma} = \{i | w_i \geq \gamma \}$ and $\Omega_{<\gamma} =[m] \setminus \Omega_{\geq \gamma}$, for some parameter $\gamma \in [0,1]$. Suppose $\tilde{H} = \sum_{i \in \Omega_{\geq \gamma}} \tilde{w}_i \tilde{H}_i$ is another Hamiltonian such that, for all $i \in \Omega_{\geq \gamma}$, we have $\abs{\tilde{w}_i -w_i} \leq \epsilon_0$ and  $\norm{H_i - \tilde{H}_i } \leq \epsilon_1$. Then
\begin{align*}
    \norm{H - \tilde{H}} \leq m (\gamma+ \epsilon_0) + (W+m \epsilon_0) \epsilon_1
\end{align*}
\label{lem:H_apx}
\end{lemma}
\begin{proof} 
\begin{align*}
    \norm{H - \tilde{H}} &= \norm{\sum_{i \in \Omega_{\geq \gamma}} \tilde{w}_i \tilde{H}_i - \sum_{i \in [m]} w_i H_{i}}\\
    &\leq \sum_{i \in \Omega_{\geq \gamma}}\norm{\tilde{w}_i \tilde{H}_i - w_i H_{i}} + \norm{\sum_{i \in \Omega_{< \gamma}} w_i H_{i}}\\
    &\leq m\gamma+ \sum_{i \in \Omega_{\geq \gamma}}\norm{\tilde{w}_i \tilde{H}_i - w_i H_{i}}\\
    &\leq m\gamma+  \sum_{i \in \Omega_{\geq \gamma}} \norm{\tilde{w}_i \left(\tilde{H}_i - {H}_{i}\right) + H_{i} \left(  \tilde{w}_i -  w_i  \right)} \\
    &\leq m\gamma+ \sum_{i \in \Omega_{\geq \gamma}} \norm{\tilde{w}_i \left(\tilde{H}_i - {H}_{i}\right)} + \norm{H_{i} \left(  \tilde{w}_i -  w_i  \right)} \\
    &\leq  m\gamma+ \sum_{i \in \Omega_{\geq \gamma}} \tilde{w}_i  \norm{\tilde{H}_i - {H}_{i}} + \norm{H_{i}}  \abs{\tilde{w}_i -  w_i  } \\
    &\leq m\gamma+ \sum_{i \in \Omega_{\geq \gamma}} \left(w_i + \epsilon_0 \right) \epsilon_1 + \epsilon_0 \\
    &\leq  m (\gamma+ \epsilon_0) + (W+m \epsilon_0) \epsilon_1,
\end{align*}    
where in going from line $1$ to line $2$ we used the triangle inequality and the definition of $\Omega_{\geq \gamma}$, from line $3$ to $4$ again the definition of $\Omega_{\geq \gamma}$, from line $4$ to line $5$ the triangle inequality, from line $5$ to $6$ the absolute homogeneity of the norm, from $6$ to $7$ and $7$ to $8$ the properties on norm and absolute value differences as stated in the lemma.
\end{proof}
We will also use the fact that the learning of the Hamiltonian parameters in the reduction can be viewed as the well-known `Double dixie cup' problem~\cite{newman1960double}. In this problem, which is a generalization of the coupon collector problem, a collector wants to obtain $m$ copies of each element from a set of $n$ elements, via a procedure where every item is sampled with equal probability. In our setting, we have that the probability over the items is non-uniform. However, it is straightforward to obtain an upper bound on the non-uniform double dixie cup problem in terms of the expectation value in the uniform setting, as illustrated by the following lemma.

\begin{lemma}[Upper bound on the non-uniform double dixie cup problem] Given samples from the set $N = [n]$, according to a distribution $\mathcal{P}$, consider the subset $M_{\gamma} \subseteq N$ such that
\begin{align*}
   M_{\gamma} = \{ i \in N : \mathcal{P}(i) \geq \gamma \},
\end{align*}
for some $\gamma \in [0,1]$. Let $T^\mathcal{P}_m (M)$ be the random variable indicating the first time that all elements in $M_{\gamma}$ have been sampled at least $m$ times when sampling from $N$ over the distribution $\mathcal{P}$. Write $T_m(S)$ when the distribution over some set $S$ is uniform. Then we have that
\begin{align*}
    \mathbb{E}[T^\mathcal{P}_m(M_{\gamma})] \leq \mathbb{E}[T_m(\lceil 1/\gamma \rceil)],
\end{align*}
where 
\begin{align*}
       \mathbb{E}[T_m(\lceil 1/\gamma \rceil)]=  \lceil 1/\gamma \rceil \ln \lceil 1/\gamma \rceil + (m-1)\lceil 1/\gamma \rceil
\ln \ln \lceil 1/\gamma \rceil + \mO\left(\lceil 1/\gamma \rceil\right) .
\end{align*}
\label{lem:ddcp}
\end{lemma}
\begin{proof} 
W.l.o.g., let the first $|M_{\gamma}|$ items of $\lceil 1/\gamma \rceil$ correspond to the items in $M_{\gamma}$. We have for all $i \in M_{\gamma}$ that $\mathcal{P}(i) \geq \gamma \geq \mathbb{P}[\text{sample $i$ from $\lceil 1/\gamma \rceil$}]$. Hence, if we have seen all elements of $\lceil 1/\gamma \rceil$ at least $m$ times, we also have seen all elements of $M_\gamma$ at least $m$ times, so from both facts together it follows that
\begin{align*}
    \mathbb{E}[T^\mathcal{P}_m(n)] \leq \mathbb{E}[T_m(S_{\lceil 1/\gamma \rceil})],
\end{align*}
where 
\begin{align*}
        \mathbb{E}[T_m(S_{\lceil 1/\gamma \rceil})] =  \lceil 1/\gamma \rceil \ln \lceil 1/\gamma \rceil + (m-1)\lceil 1/\gamma \rceil
\ln \ln \lceil 1/\gamma \rceil + \mO\left(\lceil 1/\gamma \rceil\right) 
\end{align*}
is by the result of~\cite{newman1960double}.
\end{proof}

We are now in a position to show the quantum reduction, which proceeds by learning the probability distribution over which indices of the proof string are queried by the verifier circuit, as well as the probability that the verification circuit $V$ of the $\QCPCP$-verifier accepts when those proof bits take on particular values.

\begin{lemma} Let $q \in \mathbb{N}$ be some constant and $x$ an input with $|x| = n$. Consider a $\QCPCP[q]$ protocol with verification circuit $V_x$ (which is $V$ but with the input $x$ hardcoded into the circuit), and proof $y \in \{0,1\}^{p(n)}$, and let 
\begin{align*}
     \mathcal{P}_x(i_1,\dots,i_q) = \mathbb{P}[V_x \text{ queries the proof at indices } (i_1,\dots,i_q)]
\end{align*}
and
\begin{align*}
    \lambda_{x,(i_1,\dots,i_q)}(z) = \mathbb{P}[V_x \text{ accepts given proof bits $i_1,\dots, i_q$ are queried and are given by $y_{i_1} = z_1,\dots, y_{i_q} = z_q$}].
\end{align*}
Let $\Omega = \{(i_1,\dots,i_q) : i_j \in [p(n)], \forall j \in [q]\}$, $\Omega_{\geq \gamma}  = \{ (i_1,\dots,i_q) \in \Omega | \mathcal{P}_x (i_1,\dots,i_q) \geq \gamma \}$ and $\Omega_{< \gamma} = \Omega \setminus \Omega_{\geq \gamma}$, for some parameter $\gamma \in [0,1]$. Then there exists a quantum algorithm that, for all $(i_1,\dots,i_q) \in \Omega_{\geq \gamma}$ and all $z\in \{0,1\}^q$, provides estimates $\tilde{\mathcal{P}}_x(i_1,\dots,i_q)$ and $\tilde{\lambda}_{x,(i_1,\dots,i_q)}(z)$ such that
\begin{align*}
    \abs{\tilde{\mathcal{P}}_x(i_1,\dots,i_q) - \mathcal{P}_x (i_1,\dots,i_q) } \leq \epsilon_0,
\end{align*}
and
\begin{align*}
    \abs{\tilde{\lambda}_{x,(i_1,\dots,i_q)}(z) - {\lambda}_{x,(i_1,\dots,i_q)}(z)} \leq \epsilon_1,
\end{align*}
with probability $1-\delta$, and runs in time $ \poly(n,1/\gamma,1/\epsilon_0,1/\epsilon_1,1/\delta)$.
\label{lem:reduction}
\end{lemma}

\begin{proof} 
The idea of this quantum algorithm is that it runs $V_x$ many times gathering statistics about which indices are most likely to be queried by $V_x$. The $\QCPCP$ protocol would return bits from the proof $y$ according to these indices, instead, the quantum algorithm provides a string $z \in \{0,1\}^q$ and records the output generated by $V_x$ on input $z$, generating statistics for the probability of acceptance. For every run, counted by $t \in [T]$ for some $T \in \mathbb N$, this generates a tuple $O^{t,z} = ((i^{t,z}_1,\dots,i^{t,z}_k),o^{t,z})$,  in which the proof $y$ was supposed to be queried at indices $i_1,\dots,i_q$, and in which those bits were assigned the values $y_{i_1} = z_1,\dots,y_{i_q} = z_q$, and where $o$ is the accept/reject measurement outcome. It repeats this process $T$ times for every $z$, where we will specify $T$ later ensuring that the statistics gathered are accurate with high probability. \\

The resulting algorithm can be specified as follows:
\begin{enumerate}
    \item For $z \in \{0,1\}^q$:
    \begin{enumerate}
        \item Run $V_x$ for a total of $T$ times to obtain samples $\{O^{t,z}\}_{t \in [T]}$.
        \item For all observed $(i_1^{t,z},\dots,i_q^{t,z})$, set 
        \begin{align*}
            \tilde{\lambda}_{x,(i_1,\dots,i_q)}(z) := \frac{\text{\text{\# times $o^{t,z} =1 \text{ and } i_1,\dots,i_q$ observed }}}{\text{\# times $i_1,\dots,i_q$ observed }}.
        \end{align*}
    \end{enumerate}
    \item Set 
    \begin{align*}
        \tilde{\mathcal{P}}_x(i_1,\dots,i_q) = \sum_{z \in \{0,1\}^{q}} \frac{\text{\# times $(i_1^{t,z},\dots,i_q^{t,z})$ observed}}{2^q T},
    \end{align*}
    \item For any estimated $\tilde{\mathcal{P}}_x(i_1,\dots,i_q) \leq \gamma$ remove both $\tilde{\mathcal{P}}_x(i_1,\dots,i_q)$ and associated $\tilde{\lambda}_{x,(i_1,\dots,i_q)}(z)$ for all $z$, and output all of the remaining ones.
\end{enumerate}

Let us now show that there exists a $T$ not too large such that the criteria of the theorem are satisfied. Since the $\mathcal{P}_x (i_1,\dots,i_q)$ form a discrete distribution over the set $\Omega$,\footnote{Note that $|\Omega| = \binom{n}{q} \leq \left(\frac{e n}{q} \right)^q$ which for constant $q$ is polynomial in $n$.}  we know by a standard result in learning theory (see for example~\cite{canonne2020short}) that a total of 
\begin{align*}
    \Theta \left(\frac{|\Omega| + \log (1/\delta_0)}{\epsilon_0^2} \right)
\end{align*}
samples of $O^{t,z}$ (the `$z$'-value is in fact irrelevant here) suffices to get, with probability at least $1-\delta_0$, estimates $\tilde{\mathcal{P}}_x(i_1,\dots,i_q)$ which satisfy
\begin{align*}
        \abs{\tilde{\mathcal{P}}_x(i_1,\dots,i_q) - \mathcal{P}_x (i_1,\dots,i_q) } \leq \epsilon_0.
\end{align*}
To learn estimates $\tilde{\lambda}_{x,(i_1,\dots,i_q)}(z)$ for a single index configuration $(i_1,\dots,i_q)$ and proof configuration $z$, Hoeffding's inequality tells us that we only need
\begin{align*}
     \mO \left( \frac{ \log \left(1/\delta_1\right) }{\epsilon_1^2}\right)
\end{align*}
samples of $O^{t,z}$ to have that $\abs{\tilde{\lambda}_{x,(i_1,\dots,i_q)}(z) - {\lambda}_{x,(i_1,\dots,i_q)}(z)} \leq \epsilon_1,$ with probability $1-\delta_1$. This means that any index configuration $(i_1,\dots,i_q)$ such that $\mathcal{P}_x (i_1,\dots,i_q) \geq \gamma$ needs to appear $\mO \left( \frac{ \log \left(1/\delta_1\right) }{\epsilon_1^2}\right)$ many times, to get a good estimate of $\tilde{\lambda}_{x,(i_1,\dots,i_q)}(z)$. 
Lemma~\ref{lem:ddcp} shows that the expected number of samples needed such that this condition is met is upper bounded by 
\begin{align*}
   \left\lceil \frac{1}{\gamma} \right\rceil \ln \left\lceil \frac{1}{\gamma} \right\rceil + \left( \mO \left( \frac{ \log \left(1/\delta_1\right) }{\epsilon_1^2}\right)-1\right)\left\lceil \frac{1}{\gamma} \right\rceil  \ln \ln \left\lceil \frac{1}{\gamma} \right\rceil  + \mO\left( \left\lceil \frac{1}{\gamma} \right\rceil \right),
\end{align*}
which by Markov's inequality means that
\begin{align*}
    \frac{1}{\delta_\lambda}\left( \left\lceil \frac{1}{\gamma} \right\rceil \ln \left\lceil \frac{1}{\gamma} \right\rceil + \left( \mO \left( \frac{ \log \left(1/\delta_1\right) }{\epsilon_1^2}\right)-1\right)\left\lceil \frac{1}{\gamma} \right\rceil  \ln \ln \left\lceil \frac{1}{\gamma} \right\rceil +  \mO \left(\left\lceil \frac{1}{\gamma} \right\rceil\right)\right)
\end{align*}
samples of $O^{t,z}$ suffice to turn this into an algorithm that achieves success probability $\geq 1-\delta_\lambda$. To ensure that our entire algorithm succeeds with probability $1-\delta$, we require that
\begin{align*}
    (1-\delta_\lambda)^{2^q } (1-\delta_0) (1-\delta_1)^{2^q \left\lceil \frac{1}{\gamma} \right\rceil} \geq 1-\delta,
\end{align*}
which can be achieved by setting $\delta_\lambda = \delta/(2^{q+2})$,  $\delta_0 = \delta/4$ and $\delta_1 = \delta/( \lceil 1/\gamma\rceil 2^{q+2})$. Both the statistics for probabilities over the set of indices as well as the output probabilities are gathered at the same time. This means that the requirements on the number samples needed for both estimations can be met at the same time, therefore the total number of samples $T$ that we must take satisfies 
\begin{align*}
    T & \geq \max \{\Theta \left(\frac{\left\lceil \frac{1}{\gamma} \right\rceil + \log (1/\delta_0)}{\epsilon_0^2} \right), \frac{2^q}{\delta_\lambda} \left( \left\lceil \frac{1}{\gamma} \right\rceil \ln \left\lceil \frac{1}{\gamma} \right\rceil + \left(\mO \left( \frac{ \log \left(1/\delta_1\right) }{\epsilon_1^2}\right)-1\right)\left\lceil \frac{1}{\gamma} \right\rceil  \ln \ln \left\lceil \frac{1}{\gamma} \right\rceil +  \mO \left(\left\lceil \frac{1}{\gamma} \right\rceil\right)\right) \}\\
 & \geq \max \{\Theta\left(\frac{\left\lceil \frac{1}{\gamma} \right\rceil + \log \left(\frac{1}{\delta}\right)}{\epsilon_0^2} \right), \frac{2^{2(q+1)}}{\delta}\left( \left\lceil \frac{1}{\gamma} \right\rceil \ln \left\lceil \frac{1}{\gamma} \right\rceil + \mO \left( \frac{ q \log \left(\left \lceil \frac{1}{\gamma} \right\rceil /\delta \right) }{\epsilon_1^2}\right) \left\lceil \frac{1}{\gamma} \right\rceil  \ln \ln \left\lceil \frac{1}{\gamma} \right\rceil  \right)\}
\end{align*}
which yields a total runtime of $\mO (\poly(n,\lceil 1/\gamma \rceil,1/\delta,1/\epsilon_1,1/\epsilon_0))$ when $q = \mO(1)$. 
\end{proof}
\begin{lemma}
Let $q \in \mathbb{N}$ be some constant, then there exists a quantum algorithm that can reduce any problem solvable by a $\QCPCP[q]$ protocol, without access to the proof $y$, to a diagonal Hamiltonian $\tilde{H}_x$ with the following properties:
\begin{itemize}
    \item $x \in P_{\sc yes} \Rightarrow \exists y \in \{0,1\}^{p(n)} : \bra{y }\tilde{H}_x \ket{y } \leq \frac{1}{3} + \epsilon$ 
    \item $x \in P_{\sc no} \Rightarrow \forall y \in \{0,1\}^{p(n)} : \bra{y }\tilde{H}_x \ket{y } \geq \frac{2}{3} - \epsilon$\,.
\end{itemize}
This reduction succeeds with probability $1-\delta$ and runs in time $\poly(n, 1/\epsilon)$.
\label{lem:reduction_to_H}
\end{lemma}

\begin{proof}
Let $A$ be any promise problem in $\QCPCP[q]$ and $x$ an input, and denote by $V_x$ the corresponding $\QCPCP[q]$-verifier circuit which is $V$ but with the input $x$ hardcoded into the circuit. Using the notation of Lemma~\ref{lem:reduction}, we know that the $\QCPCP[q]$ protocol has the following property\footnote{We write $1-\bra{y}H_x\ket{y}$ only to follow the convention that the local Hamiltonian problem is usually stated as a minimization problem.}
\begin{align*}
    \mathbb{P}[\QCPCP[q]\text{ protocol accepts } y] &= \sum_{(i_1,\dots,i_q)\in \Omega} \mathcal{P}_x (i_1,\dots,i_q) \lambda_{x,(i_1,\dots,i_q)}(y_{i_1},\dots,y_{i_q})\\
    &=  \bra{y} \sum_{(i_1,\dots,i_q)\in \Omega}\mathcal{P}_x (i_1,\dots,i_q) \sum_{z \in \{0,1\}^q} \lambda_{x,(i_1,\dots,i_q)} (z) |z\rangle\langle z|_{i_1,\ldots,i_q} \ket{y}\\
    &=1-\bra{y} H_x \ket{y},
\end{align*}
where $H_x$ is a $q$-local Hamiltonian that is diagonal in the computational basis, and takes the form
\begin{align*}
    H_x = \sum_{(i_1,\dots,i_q)\in \Omega}\mathcal{P}_x (i_1,\dots,i_q) H_{x,(i_1,\dots,i_q)}\,,
\end{align*}
with $q$-local terms given by
\begin{align*}
    H_{x,(i_1,\dots,i_q)}  = \sum_{z \in \{0,1\}^q} \big(1-\lambda_{x,(i_1,\dots,i_q)} (z)\big) |z\rangle\langle z|_{i_1,\ldots,i_q}.
\end{align*}
Note that this Hamiltonian has the following two properties inherited from the $\QCPCP[q]$ protocol:
\begin{itemize}
    \item $x \in P_{\sc yes} \Rightarrow \exists y \in \{0,1\}^{p(n)} : \bra{y }H_x \ket{y } \leq \frac{1}{3}$ 
    \item $x \in P_{\sc no} \Rightarrow \forall y \in \{0,1\}^{p(n)} : \bra{y }H_x \ket{y } \geq \frac{2}{3}$\,.
\end{itemize}

By Lemma~\ref{lem:reduction}, there exists a polynomial-time quantum algorithm that given $V_x$ and a parameter $\gamma$ provides estimates $\tilde{\mathcal{P}}_x(i_1,\dots,i_q)$ and $\tilde{\lambda}_{x,(i_1,\dots,i_q)}(z)$ such that
\begin{align*}
    \abs{\tilde{\mathcal{P}}_x(i_1,\dots,i_q) - \mathcal{P}_x (i_1,\dots,i_q) } \leq \epsilon_0,
\end{align*}
and
\begin{align*}
    \abs{\tilde{\lambda}_{x,(i_1,\dots,i_q)}(z) - {\lambda}_{x,(i_1,\dots,i_q)}(z)} \leq \epsilon_1,
\end{align*}
for all $(i_1,\dots,i_q)$ such that $\mathcal{P}_x (i_1,\dots,i_q) \geq \gamma$. This algorithm succeeds with probability $1-\delta$ and runs in time $\mO(poly(n,\lceil1/\gamma\rceil,1/\delta, 1/\epsilon_0 , 1/\epsilon_1))$. Given these estimates, the quantum algorithm constructs the Hamiltonian
\begin{align*}
    \tilde{H}_x = \sum_{(i_1,\dots,i_q)\in \Omega_{\geq \gamma}} \tilde{\mathcal{P}}_x(i_1,\dots,i_q) \tilde{H}_{x,(i_1,\dots,i_q)},
\end{align*}
where
\begin{align*}
     \tilde{H}_{x,(i_1,\dots,i_q)}  = \sum_{z \in \{0,1\}^q} (1-\tilde{\lambda}_{x,(i_1,\dots,i_q)}) (z) |z\rangle\langle z|_{i_1,\ldots,i_q}\,.
\end{align*}
By Lemma~\ref{lem:H_apx} we can upper bound the difference between the Hamiltonian and the learned Hamiltonian as follows
\begin{align*}
    \norm{\tilde{H}_x-H_x} \leq |\Omega| (\gamma + \epsilon_0) + (1 + |\Omega|\epsilon_0)\epsilon_1
\end{align*}
Now we set $\gamma = \epsilon/(4|\Omega|)$, $\epsilon_0 = \epsilon/(4|\Omega|)$, and $\epsilon_1 = \epsilon/4$, such that $\norm{\tilde{H}_x-H_x} \leq \epsilon$. Then, conditioning on the quantum algorithm succeeding, we have that
\begin{itemize}
    \item $x \in P_{\sc yes} \Rightarrow \exists y \in \{0,1\}^{p(n)} : \bra{y }\tilde{H}_x \ket{y } \leq \frac{1}{3} + \epsilon$ 
    \item $x \in P_{\sc no} \Rightarrow \forall y \in \{0,1\}^{p(n)} : \bra{y }\tilde{H}_x \ket{y } \geq \frac{2}{3} - \epsilon$\,.
\end{itemize}
This algorithm succeeds with probability $\geq 1-\delta$ and runs in time $\poly(n, 1/\epsilon)$, when $q = \mO(1)$.
\end{proof}

We now have all the ingredients to obtain the main result of this section.
\begin{theorem}For all constant $q \in \mathbb{N}$, we have that
\begin{align*}
    \QCPCP[q] \subseteq \BQP^{\NP[1]}.
\end{align*}
\label{thm:QCPCP_NP_qr}
\end{theorem}

\begin{proof}
Let $q\in \mO(1)$ be any constant, and consider any promise problem $A$ in $\QCPCP[q]$ with a $\QCPCP[q]$-protocol that uses a verification circuit $V_x$ (which is again $V$ with input $x$ hardwired into the circuit), which makes $q = \mO(1)$ non-adapative queries to a proof $y \in \{0,1\}^{p(n)}$. By Lemma~\ref{lem:reduction_to_H} there exists a polynomial-time quantum algorithm that reduces $V_x$ to a diagonal local Hamiltonian $\tilde{H_x}$ without access to the proof $y$, such that
\begin{itemize}
    \item $x \in P_{\sc yes} \Rightarrow \exists y \in \{0,1\}^{p(n)} : \bra{y }\tilde{H}_x \ket{y } \leq \frac{1}{3} + \epsilon$ 
    \item $x \in P_{\sc no} \Rightarrow \forall y \in \{0,1\}^{p(n)} : \bra{y }\tilde{H}_x \ket{y } \geq \frac{2}{3} - \epsilon$\,.
\end{itemize}
Note that $\frac{1}{3} + \epsilon$ and $\frac{2}{3} - \epsilon$ are separated by a constant for any constant $\epsilon < 1/6$. Since our reduction created a \emph{diagonal} (and thus classical) Hamiltonian, which is a $\mathsf{CGaLH}(q,\zeta,\delta)$ instance with $\zeta =1$ and $\delta = \Omega(1)$, the corresponding $q$-local Hamiltonian problem can be solved in $\NP$.\footnote{We would like to stress here that the $\NP$ containment only holds because all complexity classes considered are defined as classes of promise problems and not just languages. However, the promise on the final diagonal Hamiltonian eigenvalue problem can be removed by using that (i) every diagonal entry of each local term is only learned to a certain number of bits of precision and (ii) the eigenvalues of a sum of diagonal matrices are given by the sums of the eigenvalues (here we have to consider the full matrix, so also with the tensored identities which are left out in the notation). This way, one knows exactly what values the eigenvalues of the final Hamiltonian can take, and one can simply modify the completeness and soundness parameters such that they are both represented exactly by a certain number of bits and have a promise gap which is given by the difference between two successive numbers in the used binary representation. The promise gap will shrink considerably (and no longer be constant), but this is fine as the diagonal local Hamiltonian problem will still be in $\NP$ (it is even for an exponentially small promise gap).} The $\BQP^{\NP[1]}$ protocol would then consist of performing the quantum reduction to obtain the $\mathsf{CGaLH}(q,\zeta,\delta)$ instance, followed by a single call to an $\NP$ oracle. This protocol succeeds with the same success probability as the reduction, which ensures completeness $\geq 1-\delta \geq 2/3$ and soundness $\leq \delta \leq 1/3$ for any $\delta \leq 1/3$. 
\end{proof}

\noindent The above theorem yields an interesting implication in a world where the quantum-classical PCP conjecture is true, as illustrated by the following corollary.

\begin{corollary}
If Conjecture~\ref{conj:QCPCP} is true, then we have that $\NP^\BQP \subseteq \BQP^\NP$.
\label{cor:NPBQP}
\end{corollary}
\begin{proof}
    We only have to show that $\NP^\BQP \subseteq \QCMA$, since if Conjecture~\ref{conj:QCPCP} is true it implies that $\QCMA \subseteq \QCPCP \subseteq \BQP^{\NP[1]} \subseteq \BQP^\NP$ by Theorem~\ref{thm:QCPCP_NP_qr}. Let $q(n):\mathbb{N} \rightarrow \mathbb{N}$ be a polynomial and $M^\Pi(x,y)$ be a deterministic polynomial-time verification circuit that uses as an additional input a proof $y$ and make queries to a $\BQP$ oracle $\Pi = (\Pi_\text{yes},\Pi_\text{no},\Pi_\text{inv})$ at most $q(n)$ times as a black box. We define $z \in \{0,1,\perp \}^{q(n)}$ as the string that describes the sets that each query input belonged to (here `$\perp$' indicates $\Pi_\text{inv}$). Define $I = \{i : z_i =0 \text{ or } z_i = 1\}$. Since $M^\Pi$ is deterministic, we have that any string $z'$ that matches $z$ on all indices $I$ must produce the same output, so all of such $z'$s can be considered correct query strings. For each $i \in I$, a $\QCMA$ verifier can, conditioned on all previous $z_j$ for $j <i$ being computed such that it matches the first $i-1$ bits of a correct query string,  perform the $\BQP$ computation required to compute $z_i$ with success probability exponentially close to $1$ (this follows by the fact that $\BQP$ allows for probability amplification). If $i \in [q(n)]\setminus I$, any output will match in a correct query string in this case. Since there are a total of $q(n) \leq \poly(n)$ such queries, the overall success probability of simulating a query sequence can be made $\geq 2/3$. The $\QCMA$ protocol then simply simulates $M^\Pi$ by executing all operations in $M$ directly and replacing every oracle call by a direct $\BQP$-computation, which can be easily made to succeed with success probability $\geq 2/3$, ensuring completeness and soundness.
\end{proof}

An implication of Corollary~\ref{cor:NPBQP} is that it can be used to show that under the assumption $\NP \subseteq \BQP$ and the quantum-classical PCP conjecture being true, we have that $\PH \subseteq \BQP$, where the class $\PH$ is the union of all complexity classes in the polynomial hierarchy, i.e.~$\PH = \Pclass^{\NP^{\NP^{\dots}}}$. This follows from 
\begin{align*}
    \NP^\NP \subseteq \NP^\BQP \subseteq \BQP^\NP \subseteq \BQP^\BQP = \BQP,
\end{align*}
where the first and the third `$\subseteq$' are by assumption, the second is by the assumption of Conjecture~\ref{conj:QCPCP} to be true and the last equality follows from the fact that $\BQP$ is self-low. We then have that $\PH \subseteq \BQP$ follows by induction, just as is the case for $\BPP$~\cite{zachos1988probabilistic}.\footnote{See also~\url{https://blog.computationalcomplexity.org/2005/12/pulling-out-quantumness.html}.} Moreover, this would also imply that under these assumptions $\QCMA \subseteq \BQP$,  since
\begin{align*}
    \QCMA \subseteq \QCPCP[\mO(1)] \subseteq \BQP^\NP \subseteq \BQP^\BQP \subseteq \BQP.
\end{align*}
Both of these implications would provide further evidence that it is unlikely that $\NP \subseteq \BQP$.

It is known that there exists an oracle relative to which the conclusion of Corollary~\ref{cor:NPBQP} is not true, i.e.~there exists an oracle $A$ relative to which $\NP^{\BQP^A} \not\subset \BQP^{\NP^A}$~\cite{aaronson2021acrobatics}. Nevertheless, this does not necessarily mean the premise (i.e.~the quantum-classical PCP conjecture) is false: one can also easily construct an oracle separation between $\PCP$ and $\NP$, and both classes are now known to be equal~\cite{fortnow1994role}. However, this suggests that, if Conjecture~\ref{conj:QCPCP} is true, showing so requires non-relativizing techniques, just as was the case for the PCP theorem.

\section{Implications to the quantum PCP conjecture}
\label{sec:PCP}
In this final section, we consider some implications from all previous sections to the quantum PCP conjecture. We find that our results give insights into the longstanding open question of whether the reduction from a $\QPCP$ verifier to a local Hamiltonian with constant promise gap can be made classical. Furthermore, we give a no-go result for the existence of quantum gap amplification procedures exhibiting certain properties (unless $\QCMA = \NP$ or $\QCMA \subseteq \NqP$), and our results allow us to pose a conjecture which generalizes NLTS (now Theorem~\cite{anshu2022nlts}) and provides an alternative to NLSS~\cite{gharibian2021dequantizing}.\footnote{See also~\cite{coble2023local}, which proposes a closely-related conjecture independently of this work.}

\subsection{`Dequantizing' the $\QPCP$-to-local-Hamiltonian quantum reduction}
Recall Fact~\ref{fact:QPCP}, which states that the two types of QPCP conjectures (i.e.~the proof verification and local Hamiltonian problem formulations) are known to be equivalent under \emph{quantum} reductions. It has also been a longstanding open question whether the $\QPCP$-to-local-Hamiltonian reduction can be made classical~\cite{aharonov2013guest}. However, from our definition of $\QCPCP[q]$ it follows that it is unlikely that a reduction having exactly the same properties as the known quantum reduction exists, as illustrated by the following theorem.

\begin{theorem}[No-go for classical polynomial-time reductions] Let $q\in \mathbb{N}$ constant. For any $\epsilon <1/6 $ there cannot exist a classical polynomial-time reduction from a $\QPCP[q]$ circuit $V_x$ (which is the $\QPCP[q]$ verifier circuit $V$ with the input $x$ hardwired into it) to a $\mO(q)$-local Hamiltonian $H_x$ such that, given a proof $\ket{\psi}$,
\begin{align*}
    \abs{\mathbb{P}[V_x \text{ accepts } \ket{\psi}] - \left(1 - \bra{\psi}H_x \ket{\psi}\right)}\leq \epsilon,
\end{align*}
unless $\QCPCP[q] \subseteq \NP$ (which would imply $\BQP \subseteq \NP$).
\label{thm:clas_red_no_go}
\end{theorem}
\begin{proof}
This follows directly from the fact that a $\QPCP[q]$ protocol can simulate a $\QCPCP[q]$ protocol by using the exact same verification circuit $V_x$ and asking for basis states as quantum proofs (this can be forced by measuring the $q$ qubits before any quantum operation acts upon them). Since one can compute $\bra{y} H_x \ket{y}$ efficiently classically for any basis state $y$ and local Hamiltonian $H_x$, we have that the behaviour of the $\QCPCP$-system given a proof $y$ can be evaluated up to precision $\epsilon$ by computing $\bra{y} H_x \ket{y}$, where $H_x$ is the $\mO(q)$-local Hamiltonian induced by the $\QPCP$ verification circuit $V_x$. For any constant $\epsilon < (c-s)/2$ completeness and soundness are ensured.
\end{proof}

\subsection{Gap amplifications}
\noindent We now consider the implications of Section~\ref{sec:clas_cont} to gap amplifications of guidable local Hamiltonian problems.
\begin{theorem}[No-go's for quantum-classical gap amplification] 
There cannot exist
\begin{enumerate}
    \item A polynomial time classical reduction from an instance of $\mathsf{CGaLH}(k,\zeta,\delta)$ with $k\geq 2$, some constant $\zeta > 0$, and $\delta = 1/\poly(n)$ to some $\mathsf{CGaLH}(k',\zeta',\delta')$ with $k'\geq 2$, some constant $\zeta' > 0$, and $\delta' = \Omega(1)$,
\end{enumerate}
unless $\QCMA = \NP$, and
\begin{enumerate}
\setcounter{enumi}{1}
\item A quasi-polynomial time classical reduction from an instance of $\mathsf{CGaLH}(k,\zeta,\delta)$ with $k\geq 2$,  $\zeta =1/\poly(n)$, and $\delta = 1/\poly(n)$ to some $\mathsf{CGaLH}(k',\zeta',\delta')$ with $k'\geq 2$,  $\zeta' =1/\poly(n)$, and $\delta' = \Omega(1)$,
\end{enumerate}
unless $\QCMA \subseteq \NqP$.
\label{thm:QCPCP_nogos}
\end{theorem}
\begin{proof}
These all follow directly from Theorem~\ref{thm:cp_NP_NqP}.
\end{proof}

One can also interpret the no-go results for gap amplifications (points $1$ and $2$ in the above theorem) in a more general setting: if one wants to prove the QPCP conjecture through a gap amplification procedure a la Dinur, the procedure needs to have the property that it doesn't preserve `classically evaluatable' properties of eigenstates (it cannot even maintain an inverse polynomial fidelity with such states) unless at the same time showing that $\QCMA = \NP$ (or $\QCMA \subseteq \NqP$, which is also very unlikely)! Hence, this result can be viewed as a `$\QCMA$-analogy' to the result from~\cite{Aharonov2019stoquastic}, where the authors showed that the existence of quantum gap amplifications that preserve stoquasticity of Hamiltonians would imply that $\NP = \MA$. We also point out that it is possible that -- even though the complexity of $\mathsf{QGaLH}$ and $\mathsf{CGaLH}^*$ was in the \emph{inverse polynomial} precision regime the same when $1/\poly(n) \leq \zeta \leq 1-1/\poly(n)$ -- it might very well be that their complexities will differ when considering a \emph{constant} precision, as our containment results of Section~\ref{sec:clas_cont} crucially use the properties of classically evaluatable states. Finally, we note that many of the above results can also be easily adopted to the $\MA$ setting, for which all obtained results and proofs are given in Appendix~\ref{app:SPCP}. 

\subsection{Classically evaluatable states and $\QPCP$}
Finally, we close by formulating a new conjecture which can be viewed as a strengthening of the NLTS theorem, or as an alternative to the NLSS conjecture of~\cite{gharibian2021dequantizing} in light of our results, and which must hold if the quantum PCP conjecture is true and $\QMA \neq \NP$.

\begin{conjecture}[NLCES (no low-energy classically-evaluatable states) conjecture] There exists a family of local Hamiltonians $\{H_n\}_{n \in \mathbb{N}}$, where each $H_n$ acts on $n$ qubits, and a constant $\beta > 0$,  such that for sufficiently large $n$ we have that for all classically evaluatable states $\ket{u} \in \mathbb{C}^{2^n}$ as in Definition~\ref{def:cds} it holds that
$\bra{u}H\ket{u} \geq \lambda_0(H_n) + \beta$.
\label{conj:NLCES}
\end{conjecture}
Taking into account our results about the containment of the constant-gapped classically guidable local Hamiltonian problem in $\NP$ -- namely the insight that what really matters is the \textit{fidelity} of a classically evaluatable state with the low-lying energy subspace of the Hamiltonian, and not the energy of the classically evaluatable state itself -- we can also define a stronger version of the NLCES conjecture, which must hold if the quantum PCP conjecture holds. 

\begin{conjecture}[Strong-NLCES conjecture] There exists a family of local Hamiltonians $\{H_n\}_{n \in \mathbb{N}}$, where each $H_n$ acts on $n$ qubits, and a constant $\beta > 0$, such that for sufficiently large $n$ we have that for all classically evaluatable states $\ket{u} \in \mathbb{C}^{2^n}$, as in Definition~\ref{def:cds}, we have that $\|\Pi_{\lambda_0(H_n) + \beta} \ket{u}\|^2 = o(1/\poly(n))$. Here $\Pi_{\lambda_0(H_n) + \beta}$ is the projector onto the space spanned by eigenvectors of $H$ with energy less than $\lambda_0(H_n) + \beta$.
\label{conj:strong_NLCES}
\end{conjecture}

Note that the NLCES Conjecture is strictly weaker than the Strong-NLCES conjecture, and that both do not necessarily imply the $\QPCP$ conjecture.

\section*{Acknowledgements}
The authors thank Sevag Gharibian and Fran{\c{c}}ois Le Gall for useful discussions and comments on an earlier version of the manuscript, with a special thanks to Fran{\c{c}}ois Le Gall for pointing us towards looking into the implications to the quantum PCP conjecture. We also thank Harry Burhman for his help regarding promise classes, Jonas Helsen for helpful discussions, and Ronald de Wolf for providing feedback on the introduction. We thank the anonymous reviewers for their helpful comments, and in particular an anonymous STOC reviewer who provided an annotated version of an earlier version of this manuscript with very detailed and helpful comments. CC was supported by Fermioniq B.V., and thanks QuSoft and CWI for their accommodation whilst this work was completed. MF and JW were supported by the Dutch Ministry of Economic Affairs and Climate Policy (EZK), as part of the Quantum Delta NL programme.

\medskip

\printbibliography

\appendix 

\section{Perfect sampling access of MPS and stabilizer states}
\label{app:samp}
In this appendix we show that both matrix product states (MPS) and stabilizer states are samplable states, by checking all three conditions of Definition~\ref{def:sampaccess}.
\paragraph{Matrix product states:}
Let $u$ be a $N=2^n$-dimensional vector described by an MPS of $n$ particles, bounded bond dimension $D$ and local particle dimension $d$.
\begin{enumerate}[label=(\roman*)]
    \item Let $\hat{i}$ be the bit representation of $i$. The algorithm $Q_u$ can simply be the evaluation of $\text{Tr}[A_1^{(s_1)} A_2^{(s_2)} \dots A_n^{(s_n)}] $ for $s = \hat{i}$, which can be done via a naive matrix multiplication algorithm in time $\mO(nD^3)$, and thus clearly runs in time $\mO\left(\poly(\log(N)\right)$ when $d = \mO(\poly(n)$, $D = \mO(\poly(n)$.
    \item We will use that expectation values that are a tensor product of $1$-local observables can be computed efficiently for a MPS im time $\mO(nd^2 D^3)$~\cite{verstraete2008matrix}. We assume that $m$ is already known (see item (iii)), and that our MPS is therefore normalized. The algorithm $\mathcal{S}\mathcal{Q}_u$ works as follows: one computes the probability that the first qubit is $1$ by computing the expectation value of the $1$-local projector $ \Pi_1 = \ket{1}\bra{1}_1$. Let $p_1 = \bra{u} \Pi_1 \ket{u}$ and $p_0 = 1-p_1$. The algorithm now samples a bit $j_1 \in \{0,1\}$ according to distribution $\{p_0, p_1\}$, and computes the expectation value of the $2$-local projector $\Pi_{j_1 1}$ to obtain $p_{j_1,1}$ and $p_{j_1,0}$, from which again a bit is sampled according to the distribution $p_{j_1,0},p_{j_1,1}$. This procedure is repeated for all $n-2$ remaining sites, which yields a sample $j$ with probability $|u_j|^2$. The total time complexity of this procedure is $\mO(n^2 d^2 D^3) = \mO(\poly(\log N))$, when $d = \mO(\poly(n))$ and $D = \mO(\poly(n))$, as desired.
    \item $m$ can easily be computed by considering the overlap of the MPS with itself, which can be done in time $\mO(npD^3)$ as the overlap can be viewed as the expectation value of a $0$-local observable. 
\end{enumerate}

\paragraph{Stabilizer states:}
Let $u \in \mathbb{C}^{2^n}$, $N=2^n$, be a stabilizer state on $n$ qubits.
\begin{enumerate}[label=(\roman*)]
    \item This follows from the fact that basis states are stabilizer states, and that there exists an algorithm $\mathcal{Q}_u$ that computes inner products between stabilizer states in time $\mO(n^3) = \mO(\poly(\log N))$~\cite{aaronson2004improved}.
    \item This follows from the fact that stabilizer states can be strongly simulated (i.e. marginals can be computed), which allows for weak simulation as shown in~\cite{Terhal2002adaptive} at overhead $n$ for the cost of strong simulation. Using the strong simuation algorithm as in~\cite{aaronson2004improved}, this gives an algorithm $\mathcal{S}\mathcal{Q}_u$ that runs in time $\mO(n^3) =\mO(\poly(\log N))$.
    \item $m=1$ by definition.
\end{enumerate}

\section{MPS to circuit construction}
\label{app:MPS_to_Circuit}
In this section, we show that any MPS on $n$ qubits with bond dimension $D$ can be implemented on a quantum computer up to distance $\epsilon$, with respect to the $2$-norm,
in $\mO(nD\log(D)^2 \log(D n/\epsilon))$ one- and two-qubit gates and a $\mO(n poly(D))$-time classical pre-calculation. The result is based on a result from \cite{schon2005sequential}. For completeness we will first repeat their result. Let $\mathcal H_A = \mathbb C^D$ and $\mathcal H^B = \mathbb C^2$ be the Hilbert spaces characterising a $D$-dimensional ancillary system and a single qubit, respectively. Then every MPS of the form 
\[
    \ket{\psi} = \bra{\phi_F}V_{n}\dots V_{1}\ket{\phi_I}
\]
with arbitrary maps $V_{k}: \mathcal H_A \mapsto \mathcal H_A \otimes \mathcal H_{B}$, and $\ket{\phi_I}, \ket{\phi_F}\in \mathcal H_A$ is equivalent to a state 
\[
     \ket{\psi} = \bra{\tilde{\phi_F}}\tilde{V}_{n}\dots \tilde{V}_{1}\ket{\tilde{\phi}_I}
\]
with $\tilde{V}_{k}: \mathcal H_A \mapsto \mathcal H_A \otimes \mathcal H_{B}$ isometries and such that the ancillary register decouples in the last step
\[
    \tilde{V}_{n}\dots \tilde{V}_{1}\ket{\tilde{\phi}_I} = \ket{\phi_F} \otimes \ket{\psi}.
\]
Note that this is the canonical form of the MPS and can be found using $\mO(n poly(D))$ classical pre-calculation time. The isometries are of size $2 D \times D$ acting on the auxiliary system sequentially and create one qubit each. Every $\tilde{V}_{k}$ can be embedded into a unitary $U_{k}: \mathcal H_A \otimes \mathcal H_{B} \mapsto \mathcal H_A \otimes \mathcal H_{B}$ of size $2D \times 2D$, acting on the auxiliary system and a qubit initialised in $\ket{0}$ such that $U_{k}\ket{\tilde{\phi}_{k}}\otimes \ket{0} = \tilde{V}_{k}\ket{\tilde{\phi}_k}$. This gives the quantum circuit
\[
 U_{n} \dots U_{1} \ket{\tilde{\phi}_I} \ket{0}^{\otimes n} = \ket{\tilde{\phi}_F}\ket{\psi}.
\]
$\ket{\tilde{\phi_I}}$ is a state in $\mathcal H_A$ which can be generated on $\lceil\log(D)\rceil$ qubits, up to normalisation. By the Solovay-Kitaev theorem, this state can be prepared up to  distance $\epsilon$ by a circuit of 
\[
\mO(\lceil\log(D)\rceil^2 D \log(\lceil\log(D)\rceil^2 D/\epsilon)) = \mO(D\log(D)^2 \log(1/\epsilon)
\] 
two and one-qubit gates. The unitaries $U_{k}$ act on $\lceil\log(D)\rceil + 1$ qubits hence they can be approximated up to error $\epsilon$ in 
\[
\mO((\lceil\log(D)\rceil + 1)^2 (D +1)\log((\lceil\log(D)\rceil + 1)^2 (D + 1)/\epsilon)) = \mO(D\log(D)^2 \log(D/\epsilon)).
\]
Note that because every unitary incurs an error $\epsilon$ the entire error can be bounded by $n \epsilon$, setting individual error to $\epsilon' = \frac{\epsilon}{n}$ ensures that the generated state is at most $\epsilon$ far from the desired state. This results in a circuit of complexity: $\mO(n D\log(D)^2 \log(n D /\epsilon))$ generating the MPS up to normalisation.  

\section{Results for $\MA$}
\label{app:MA}
Let us define yet another class of guidable local Hamiltonian problems, which constrains the considered Hamiltonians to be of a specific form. 
\begin{definition}[Classically Guidable Local Stoquastic Hamiltonian Problem] The \textbf{Classically Guidable Local Stoquastic Hamiltonian Problem}, shortened as $\mathsf{CGaLSH}(k,\delta,\zeta)$, has the same input, promise, extra promise and output as $\mathsf{CGaLH}(k,\delta,\zeta)$ but with the extra constraint that the considered Hamiltonian is stoqastic. For our purposes this means that all $m$ $k$-local terms $H_i$ of the Hamiltonian $H = \sum_{i=0}^{m-1} H_i$ have real, non-positive off-diagonal matrix elements in the computational basis.
\label{def:GaLSH}
\end{definition}

By adopting the same proof structure as we used to prove Theorem~\ref{thm:CGaLH}, we can obtain a similar result for $\mathsf{CGaLSH}$. For this, we first define a coherent description of $\MA$, denoted as $\MA_q$.
\begin{definition}[Coherent classical verifier~\cite{bravyi2006complexity}] A coherent classical verifier is a tuple $V = (n,n_w,n_o,n_+,U)$, where
\begin{align*}
    n \quad&= && \text{number of input bits}\\
    n_w \quad&= && \text{number of witness bits}\\
    n_0 \quad&= && \text{number of ancillas } \ket{0}\\
    n_+ \quad&= && \text{number of ancillas } \ket{+}\\
    U \quad&= && \text{quantum circuit on $n+n_w+n_0+n_+$ qubits with X, CNOT and Toffoli gates.}
\end{align*}
The acceptance probability of a coherent classical verifier $V$ on input string $x \in \{0,1\}^n$ and witness state $\ket{\psi} \in \left(\mathbb{C}^2\right)^{\otimes n_w}$ is defined as 
\begin{align*}
    \text{P}[V;x,\psi] = \bra{\psi_\text{in}} U^\dagger \Pi_\text{out} U \ket{\psi_\text{in}},
\end{align*}
where $\ket{\psi_\text{in}}=\ket{x} \ket{0}^{\otimes n_0} \ket{+}^{\otimes n_+}$ is the initial state and $\Pi_\text{out} = \bra{0} \ket{0} \otimes I_\text{else}$ project the first qubit onto the state $\ket{0}$.
\label{def:CohV}
\end{definition}
\begin{definition}[$\MA_q$] The class $\MA_q[c,s]$ is the set of all languages $L \subset \{ 0,1\}^*$ for which there exists a (uniform family of) coherent classical verifier circuit $V$ such that for every $x\in \{0,1\}^*$ of length $n=|x|$,
\begin{itemize}
    \item if $x \in L $ then there exists a $\poly(n)$-qubit witness state $\ket{\psi_x}$ such that $V(x,\ket{\psi_x})$ accepts with probability $\geq c(=2/3)$,
    \item if $x \notin L$ then for every purported $\poly(n)$-qubit witness state $\ket{\psi}$, $V(x,\ket{\psi})$ accepts with probability $\leq s(=1/3)$.
\end{itemize}
\end{definition}
\begin{lemma}[\cite{bravyi2006complexity}] $\MA = \MA_q$
\label{lem:MAisMAq}
\end{lemma}

\subsection{Proof of $\MA$-hardness of the Classically Guidable Local Stoquastic Hamiltonian problem}
\label{app:CGLSH-MAhard}

\begin{theorem} $\mathsf{CGaLSH}(k,\delta,\zeta)$ is $\MA$-hard for $k\geq 7$, $\zeta  \in (1/\poly(n),1-1/\poly(n))$ and $\delta = 1/\poly(n)$.
\label{thm:CGaLSH}
\end{theorem}
\begin{proof}
This follows by using a nearly identical construction as have used to prove Theorem~\ref{thm:CGaLH}, but then starting from a $\MA_q$ verification circuit (Definition~\ref{def:CohV}) instead of a quantum circuit and with some minor changes. Let us go through the relevant steps of the construction and verify that they preserve stoquasticity. We will not concern ourselves with reducing the locality, and hence the parts of the construction needed to achieve this in the $\QCMA$ setting are left out.
\begin{itemize}
    \item We have that $\MA_q [c,s]=\MA_q[1,s]$, since this holds for $\MA$ and the $\MA_q$ verifier has the same acceptance probability as the $\MA$ verifier. 
    \item The CNOT- and Marriot and Watrous-tricks can be applied as the CNOT gate is part of the allowed gate-set in Definition~\ref{def:CohV}.
    \item One can use Kitaev's clock construction with a small penalty $\epsilon$ on $H_\text{out}$ (Lemma~\ref{lem:spcc} would also work for Kitaev's construction, see~\cite{Deshpande2020}), where the only terms that are off-diagonal in the computational basis come from $H_\text{prop}$ and $H_\text{out}$. Inspection of these terms confirms that the Hamiltonian is stoquastic~\cite{bravyi2006complexity}. Since the Toffoli gate is a $3$-local gate, we have that our Hamiltonian becomes $6$-local.
    \item The block-encoding trick also preserves stoquasticity, as both blocks themselves are stoquastic. This does increase the locality of the Hamiltonian by $1$, and therefore it becomes $7$-local.
    \item Finally, the pre-idling can be done since the identity is trivially in any gate set.
\end{itemize}
\end{proof}
Somewhat surprisingly, a variant of the guidable local Hamiltonian problem for stoqastic Hamiltonians was already considered in a work by Bravyi in 2015~\cite{BravyiMonte2015}.\footnote{See also~\cite{liu2020stoqma}, which considers $\StoqMA$ circuit problems where the witnesses are promised to be of a more restricted form.} However, the considered promise on the guiding state is different, as can be read from the following definition.
\begin{definition}[\cite{BravyiMonte2015}] Let $H$ be a stoquastic Hamiltonian. We will say that $H$ admits a guiding state if and only if there exists a pair of normalized $n$-qubit states $\psi,\phi$ with non-negative amplitudes in the standard basis such that $\psi$ is the ground state of $H$, the function $x \rightarrow \bra{x}\ket{\phi}$ is computable by a classical circuit of size $\poly(n)$, and
\begin{align*}
    \bra{x}\ket{\phi} \geq \frac{\bra{x}\ket{\psi}}{\poly(n)} \quad \text{for all } x \in \{0,1\}^n. 
\end{align*}
\label{def:BravyiGLSH}
\end{definition}
Using the Projection Monte Carlo algorithm with a variable number of walkers, Bravyi was able to show $\MA$-containment of a guided Stoquastic Hamiltonian problem that uses guiding states satisfying Definition~\ref{def:BravyiGLSH}.  Note that Definition~\ref{def:BravyiGLSH} implies that the fidelity satisfies $|\bra{\psi} \ket{\phi}|^2 = \Omega \left(1/\poly(n) \right)$, but that the converse is not necessarily true. Therefore, Bravyi's result cannot be directly used to obtain $\MA$-containment of  $\mathsf{CGaLSH}(k,\delta,\zeta)$. We leave this as an open problem for future work.

\subsection{PCP statements for $\MA$}
\label{app:SPCP}
For $\MA$ it is not so clear how to define a $\PCP$ statement in a proof verification version, as the verifiers used in the original $\PCP$ system are already probabilistic. In~\cite{Aharonov2019stoquastic}, the authors define stoqastic $\PCP$ ($\SPCP$) in the following way:
\begin{conjecture}[SPCP conjecture - frustration-free Hamiltonian version~\cite{Aharonov2019stoquastic}] There exist constants $\epsilon > 0$, $k',d' >0$ and an efficient gap amplification procedure that reduces the problem of deciding if a uniform stoquastic $d$-degree $k$-local Hamiltonian is frustration-free or at least inverse polynomially frustrated, to the problem of deciding if a uniform stoquastic $d'$-degree $k'$-local Hamiltonian is frustration-free or at least $\epsilon$ frustrated.  
\label{con:SPCP_frustration}
\end{conjecture}
\begin{theorem}[\cite{Aharonov2019stoquastic}] If conjecture~\ref{con:SPCP_frustration} is true, then $\MA=\NP$.
\end{theorem}
We obtain similar results for our guidable \textit{stoquastic} Local Hamiltonian problem, which differs from~\cite{Aharonov2019stoquastic} in the fact that the Hamiltonian does not have to be uniform and the energy decision parameters can be arbitrary real numbers (instead of zero or bounded away from zero, as is the case in deciding on the frustration of the Hamiltonian). However, this comes with the extra constraint that there has to exist a classically describable guiding state, as per Definition~\ref{def:cds}. Formally, we can define another version of a $\SPCP$-conjecture in terms of guidable Hamiltonians in the following way.
\begin{conjecture}[SPCP conjecture - guidable Hamiltonian version] $\mathsf{CGaSLH}(k,\zeta,\delta)$ with $k\geq 2$, some constant $\zeta > 0$, and $\delta = \Omega(1)$ is $\MA$-hard under classical poly-time reductions.
\label{conj:SPCP}
\end{conjecture}
Just as in the $\QCMA$-setting (see Theorem~\ref{thm:QCPCP_nogos}), we obtain certain no-go results which must hold if Conjecture~\ref{conj:SPCP} is true.
\begin{theorem} Conjecture~\ref{conj:SPCP} (SPCP conjecture) is true if and only if $\MA = \NP$.
\end{theorem}
\begin{proof}
The '$\Longrightarrow$' implication follows from Theorem~\ref{thm:cp_NP_NqP}, since $\mathsf{CGaSLH}(k,\zeta,\delta)$ is a special case of $\mathsf{CGaLH}(k,\zeta,\delta)$. `$\Longleftarrow$' follows from the fact that the NP-hard Hamiltonian in Lemma~\ref{lem:CGaLH_NP-hard} is diagonal, which therefore is also stoquastic.
\end{proof}
Furthermore, we can also give no-go results on the existence of stoquastic gap amplification procedures.
\begin{theorem} There cannot exist 
\begin{enumerate}
    \item {\bf Stoquastic gap amplification:} a polynomial-time classical reduction from an instance of $\mathsf{CGaSLH}(k,\zeta,\delta)$ with $k\geq 2$, some constant $\zeta > 0$, and $\delta = \Theta(1/\poly(n))$ to some $\mathsf{CGaSLH}(k',\zeta',\delta')$ with $k'\geq 2$, some constant $\zeta' > 0$, and $\delta' = \Omega(1)$,
\end{enumerate}
unless $\MA = \NP$, and
\begin{enumerate}
\setcounter{enumi}{1}
\item {\bf Stoquastic gap amplification (2):} a quasi-polynomial time classical reduction from an instance of $\mathsf{CGaSLH}(k,\zeta,\delta)$ with $k\geq 2$,  $\zeta =\Omega(1/\poly(n))$, and $\delta = \Theta(1/\poly(n))$ to some $\mathsf{CGaSLH}(k',\zeta',\delta')$ with $k'\geq 2$,  $\zeta' =\Omega(1/\poly(n))$, and $\delta' = \Omega(1)$,
\end{enumerate}
unless $\MA \subseteq \NqP$.
\end{theorem}
\begin{proof}
This follows again from Theorem~\ref{thm:cp_NP_NqP}.
\end{proof}
Hence, these results might provide yet another way, next to~\cite{Aharonov2019stoquastic}, to derandomize the class $\MA$ (i.e. show that $\MA=\NP$). 

\section{QCMA versions of Quantum SAT}
\label{app:GaQSAT}
The canonical $\NP$-complete problem is satisfiability, and whilst the general local Hamiltonian problem can be viewed as the `quantum analogue' of satisfiability, one might argue that an \textit{even closer} analogue would be the slightly different problem of \textit{quantum satisfiability}, introduced in~\cite{Bravyiefficient2006}. In quantum satisfiability, shortened as $k$-$\mathsf{QSAT}$, the Hamiltonian is given by a sum of local projectors and the task is to decide whether there exists a quantum state $\ket{\xi}$ such that all terms project $\ket{\xi}$ to zero or the expectation value of $H$ is greater or equal than some inverse polynomial, for all quantum states. $k$-$\mathsf{QSAT}$ was shown to be $\QMA_1$-complete for $k\geq 3$~\cite{Bravyiefficient2006,Gossetquantum2013}, and it can be solved in linear time classically when $k \leq 2$~\cite{Aradlinear2015}. 

In the case of $\QMA$ it is very much unclear whether $\QMA = \QMA_1$, but for $\QCMA$ we know that in fact $\QCMA = \QCMA_1$, as mentioned before. The simple observation that in our proof of Theorem~\ref{thm:CGaLH} using $\QCMA_1$ impies that the ground state energy is exactly zero in the {\sc yes}-case and inverse polynomially large in the {\sc no}-case, hints that our construction actually shows a result for a quantum satisfiability version for $\QCMA$ (which implies hardness for the more general local Hamiltonian version), and indeed with some very minor modifications this can be shown to be the case. First we formally define the `$\QCMA$ versions' of satisfiability.  

\begin{definition}[Guidable Quantum Satisfiability] Guidable Quantum Satisfiability Problems are problems defined by having the following input, promise, output and some extra promise to be precisely defined below for each of the problems separately:\\
\textbf{Input:} A collection $\{\Pi_i : i=1,\dots,m\}$ of $k$-local projectors acting on $n$ qubits, a precision parameter $\delta >0 $. Let the Hamiltonian $H=\sum_{i=1}^m \Pi_i$ be the sum of all these projectors.\\
\textbf{Promise:} Either there exists a state $\ket{w}$ s.t. $\Pi_i \ket{w} = 0$ for all $i=1,\dots, m$ (i.e. $\lambda_0(H) = 0)$, or otherwise $\sum_{i=1}^m |\braket{y}{\Pi_i | y}| \geq \delta$ for all states $\ket{y} \in \mathbb{C}^{2^n}$ (i.e. $\lambda_0(H) \geq \delta$).\\
\textbf{Extra promises:} 
Denote  $\Pi_\text{gs}$ for the projection on the subspace spanned by the ground state of $H$.  Then for each problem class, we have that either one of the following promises hold, each giving a different problem:
\begin{enumerate}
    \item {\bf $\mathsf{CGa}\text{-}\mathsf{QSAT}(k,\delta,\zeta)$ \textit{Classically Guidable Quantum Satifisfiability Problem} }: 
    
    There exists a classically evaluatable state $u\in \mathbb{C}^{2^n}$ for which $\norm{\Pi_\text{gs} u}^2 \geq \zeta$.
    
    \item {\bf $\mathsf{QGa}\text{-}\mathsf{QSAT}(k,\delta,\zeta)$ \textit{Quantumly Guidable Quantum Satifisfiability Problem}}: 
    
    There exists a unitary $V$ implemented by a quantum circuit composed of at most $T=\poly(n)$ gates from a fixed gate set $\mathcal{G}$ that produces the state $\ket{\phi}=V\ket{0}$,  which has $\norm{\Pi_\text{gs} \ket{\phi}}^2 \geq \zeta$.
\end{enumerate}
\textbf{Output:} \begin{itemize}
    \item If $\lambda_0(H) = 0$, output {\sc yes}.
    \item If $\lambda_0(H) \geq \delta$, output {\sc no}.
\end{itemize}
\label{def:GQKSAT}
\end{definition}
\begin{theorem} $\mathsf{CGa}\text{-}\mathsf{QSAT}(k,\delta,\zeta)$ is $\QCMA$-complete for $k\geq 4$, $\delta = 1/\poly(n)$ and $\zeta \in (1/\poly(n),1-1/\poly(n))$. $\mathsf{QGa}\text{-}\mathsf{QSAT}(k,\delta,\zeta)$ is $\QCMA$-complete for $k\geq 4$, $\delta = 1/\poly(n)$ and $\zeta \in (1/\poly(n),1-1/\exp(n)]$.
\end{theorem}
\begin{proof}
We will follow the same hardness construction as was used to prove Theorem~\ref{thm:CGaLH}, making the following three observations: (i) we can instead start from a $\QCMA_1$-hard promise problem $\langle U_2,1,p_2\rangle$ by Lemma~\ref{lem:QCMA1}. (ii) We apply a slightly modified form of \emph{error reduction} to the circuit. Since $s' = 1 - \Omega(w^{-2}) > \frac{1}{2}$, we have to resort to a ``unanimous vote" instead of a ``majority vote" in our construction to exponentially suppress all witnesses in the {\sc no}-case, and all but one witness (the one that achieves perfect completeness) in the {\sc yes}-case close to zero, again using the ``Marriot and Watrous trick" for error reduction to preserve the witness size~\cite{Mariott2005quantum}. By only accepting when \textit{all} repetitions accept, one can quickly verify that the probability of acceptance for the witness that was originally accepted with probability $1$ remains $1$, and that for all other witnesses the probability of acceptance becomes suppressed exponentially close to zero. More precisely, by repeating $k$ times we have that the probability of acceptance in the {\sc no}-case is at most $(1-\Omega(1/w^2))^k \leq e^{-\Omega(k/w^2)} = e^{-p(n)}$ for $p$ a polynomial such that $k > p(n) \mO(w^2)$, implying that we need to repeat the verification circuit only polynomially many times. Note that after these two changes the protocol is still in $\QCMA_1$, albeit with a new soundness parameter which is now exponentially close to zero. (iii) Definition~\ref{def:GQKSAT} requires that the Hamiltonian is an unweighted sum of projectors. Note that the Hamiltonian of Eq.~\eqref{eq:H_FK} from~\cite{kempe20033local} is of this form except for the fact that the terms $H_\text{prop}$, $H_\text{clock}$ and $H_\text{out}$ are weighted with some factor. For all integer factors (which is only the $T^{12}$ in front of $H_\text{clock}$), we can simply replace the factor with a sum containing this number of terms. For all non-integer terms (which are $1/2$ and $\epsilon$), we simply add $1/f $ times all other terms, where $f \in \{1/2,\epsilon\}$ is the factor. For $\epsilon$ we then simply choose a value small enough such that $1/\epsilon$ is integer when using Lemma~\ref{lem:spcc}. (iv) $H_\text{no}$ of Eq.~\eqref{eq:H_no} is a sum of projectors, and the block-encoding trick used to construct $H$ of Eq.~\eqref{eq:full_H} ensures that all terms remain projectors. (v) The ground state energy is (before the locality reduction through the perturbative gadgets, which we do not apply here) precisely $0$ in the {\sc yes}-case and inverse polynomially far from zero in the {\sc no}-case. This also ensures containment through the use of quantum phase estimation. Also, note that we do not have to scale down the Hamiltonian as $\norm{H}\leq 1$ is not required by Definition~\ref{def:GQKSAT}.

\end{proof}

\end{document}